\documentclass[12pt]{article}

\usepackage[utf8]{inputenc} 

\usepackage{amsmath,amsthm,amsfonts,amssymb,mathrsfs}
\usepackage{textcomp,enumerate,bbm,latexsym, bm}
\usepackage{enumitem} 
\usepackage{graphicx} 
\usepackage{url,verbatim}
\usepackage[normalem]{ulem}
\usepackage{algpseudocode}
\usepackage{natbib}
\usepackage{multirow}
\usepackage{booktabs}
\usepackage{chngcntr, subcaption}
\usepackage{authblk}
\usepackage{thmtools} 
\usepackage{thm-restate}
\usepackage{dsfont}
\usepackage{subcaption}

\usepackage[usenames,table]{xcolor}
\usepackage[colorlinks,linktoc=all]{hyperref}
\usepackage[all]{hypcap}
\hypersetup{citecolor=cyan}
\hypersetup{linkcolor=cyan}
\hypersetup{urlcolor=cyan}
\usepackage{cleveref}

\graphicspath{{./art/}}
\usepackage[ruled,vlined,linesnumbered]{algorithm2e}

\newcommand{\blind}{1}

\usepackage[small]{titlesec}
\titlelabel{\thetitle.\quad} 
\titleformat*{\section}{\bf\large\center} 
\titlespacing*{\section}
{0pt}{3ex}{2ex}
\titlespacing*{\subsection}
{0pt}{2.5ex}{1.5ex}
\titlespacing*{\subsubsection}
{0pt}{1ex}{1ex}

\allowdisplaybreaks  

\renewcommand{\Pr}{\mathbb{P}}

\def\weightedsumnoation{S_{n, \vec{\omega}}}
\def\weightedsum{\sum_{i=1}^n\omega_iX_i}
\def\kappaweight{\sum_{i = 1}^n \omega_i^{\gamma}}

\crefname{figure}{Fig.}{Fig.}
\Crefname{figure}{Figure}{Figure}

\theoremstyle{plain}
\newtheorem{theorem}{Theorem}[section]
\newtheorem{lemma}[theorem]{Lemma}

\theoremstyle{definition}
\newtheorem{definition}{Definition}[section]
\Crefname{definition}{Definition}{Definitions}

\crefname{assumption}{assumption}{assumptions}
\crefname{theorem}{theorem}{theorems}

\theoremstyle{remark}
\newtheorem{remark}{Remark}[section]

\begin{document}
\def\spacingset#1{\renewcommand{\baselinestretch}%
{#1}\small\normalsize} \spacingset{1}

 \if1\blind
{
\title{\bf Aggregating Dependent Signals with Heavy-Tailed Combination Tests}
  \author{Lin Gui\vspace{-0.35cm}\\
    Department of Statistics, The University of Chicago \hspace{.2cm}\\
  \hspace{.2cm}\\ 
    Yuchao Jiang\hspace{.2cm}\\
    Departments of Statistics and Biology, Texas A\&M University \hspace{.2cm}\\
    \hspace{.2cm}\\ 
    Jingshu Wang\hspace{.2cm}\\
    Department of Statistics, The University of Chicago
    }
    \date{\vspace{-5ex}}
  \maketitle
} \fi

\if0\blind
{
  \bigskip
  \bigskip
  \bigskip
  \begin{center}
    {\LARGE\bf  Aggregating Dependent Signals with Heavy-Tailed Combination Tests}
\end{center}
  \medskip
} \fi

\bigskip
\begin{abstract}
Combining dependent p-values poses a long-standing challenge in statistical inference, particularly when aggregating findings from multiple methods to enhance signal detection. Recently, p-value combination tests based on regularly varying-tailed distributions, such as the Cauchy combination test and harmonic mean p-value, have attracted attention for their robustness to unknown dependence. This paper provides a theoretical and empirical evaluation of these methods under an asymptotic regime where the number of p-values is fixed and the global test significance level approaches zero.
We examine two types of dependence among the p-values. First, when p-values are pairwise asymptotically independent, such as with bivariate normal test statistics with no perfect correlation, we prove that these combination tests are asymptotically valid. However, they become equivalent to the Bonferroni test as the significance level tends to zero for both one-sided and two-sided p-values. Empirical investigations suggest that this equivalence can emerge at moderately small significance levels. Second, under pairwise quasi-asymptotic dependence, such as with bivariate t-distributed test statistics, our simulations suggest that these combination tests can remain valid and exhibit notable power gains over Bonferroni, even as the significance level diminishes. These findings highlight the potential advantages of these combination tests in scenarios where p-values exhibit substantial dependence. Our simulations also examine how test performance depends on the support and tail heaviness of the underlying distributions.
\end{abstract}

\vfill

\newpage
\spacingset{1.4}

\section{Introduction}
\label{sec:intro}
Combining dependent p-values to assess the global null hypothesis has long been a fundamental challenge in statistical inference. 
A common scenario arises when integrating the results of various methods on the same dataset to enhance signal detection power \citep{wu2016metacycle, Rosenbaumtestingtwice}. When individual p-values have arbitrary dependence, the Bonferroni test is the most common approach with a theoretical guarantee. However, it is often criticized for being overly conservative in practical applications.
 
Specifically, consider $n$ individual p-values $P_1,\ldots, P_n$. To test the global null hypothesis, i.e., all $n$ null hypotheses are true, the Bonferroni test calculates the combined p-value as $n\times\min\left(P_1, \ldots, P_n\right)$. Due to the scaling factor $n$, the Bonferroni combined p-value may exceed any of the individual p-values, leading to a loss of power during the combination process.

Recently, a novel approach gaining traction involves the combination of p-values through transformations based on heavy-tailed distributions \citep{liu2019acat,wilson2019harmonic}. Let $X_i$ be defined as $Q_F(1 - P_i)$, where $F(\cdot)$ represents the cumulative distribution function of a heavy-tailed distribution and $Q_F$ is its quantile function. The core idea is to compute the combined p-value based on the tail distribution of $S_n = \sum_{i = 1}^n X_i$, which under the global null is robust to dependence among the heavy-tailed variables ${X_1,\ldots, X_n}$. 
The Cauchy combination test, which sets $F$ as the standard Cauchy distribution, was first introduced in \citet{liu2019acat} for genome-wide association studies (GWAS) and has since been applied in genetic and genomic research, including spatial transcriptomics \citep{sun2020statistical}, ChIP-seq data \citep{qin2020lisa}, and single-cell genomics \citep{cai2022model}. Another popular method, the harmonic mean p-value \citep{wilson2019harmonic}, employs the Pareto distribution with shape parameter $\gamma = 1$ as $F$. 

Despite the growing popularity of these heavy-tailed combination tests in practical applications, there has been limited theoretical investigation and empirical evaluation of these methods. Existing studies \citep{liu2020cauchy,fang2021heavy} have provided asymptotic validity of these tests as the significance level $\alpha \to 0$ for pairwise bivariate normal test statistics. These results closely related to earlier findings on sums of regularly varying tail variables, showing that $\Pr\left(S_n>x\right)$ and $n\left\{1 - {F}(x)\right\}$ are asymptotically equivalent as $x \to +\infty$, provided that the variables ${X_1, \ldots, X_n}$ are pairwise quasi-asymptotically independent \citep{chen2009sums}. Intuitively, for heavy-tail distributed $X_1,\ldots, X_n$, their maximum typically dominates the sum, making the latter less sensitive to dependence among $X_1,\ldots, X_n$. Yet this same intuition raises doubts about the true benefits of these tests compared to the Bonferroni test. Additionally, the assumption of quasi-asymptotic independence, while covering any bivariate normal variables that are not perfectly correlated, remains more stringent than allowing arbitrary dependence. For example, bivariate t-distributed variables, which are frequently used as test statistics, are not quasi-asymptotically independent.  This raises questions about the robustness of these tests when faced with unknown dependence structures. 

This paper addresses these concerns through theoretical and empirical analyses. 
 Many applications employ heavy-tailed combination tests to aggregate results from different methods or studies, often in settings where the number of base hypotheses, $n$, is moderate rather than excessively large. 
Accordingly, we focus on scenarios where $n$ is fixed and analyze the asymptotic regime as the significance level $\alpha \to 0$. 
Our theoretical investigation shows that when test statistics are quasi-asymptotically independent, particularly when they follow a bivariate normal distribution with imperfect correlation, the rejection regions of heavy-tailed combination tests are asymptotically equivalent to those of the Bonferroni test as $\alpha$ approaches zero. This suggests that in the same asymptotic regime where combination tests have proven to be valid, they offer no real power advantage over Bonferroni's approach.
However, when the assumption of asymptotic independence is violated, such as when test statistics follow a multivariate t distribution, our empirical results indicate that combination tests still appear to be asymptotically valid when the tail index
$\gamma \leqslant 1$, despite the lack of a theoretical guarantee. More strikingly, they exhibit significantly greater power than the Bonferroni test, highlighting their potential advantages in settings where p-values are strongly dependent, a scenario that often arises when aggregating results from different methods applied to the same dataset. 
Furthermore, through simulations and real-world case studies,
we observe that the empirical validity and power of these tests are affected by both the heaviness and support of the heavy-tail distribution. 

\section{Model setup and theoretical results} \label{sec:setup}
\subsection{Model setup}
Consider $n$ test statistics $T_1,\ldots,T_n$, where each $T_i$ is for a base null hypothesis $H_{0,i}$. For each base hypothesis, we construct a one-sided or two-sided base p-value $P_i$ based on the distribution of $T_i$ under $H_{0,i}$.
We are interested in testing the global null hypothesis
\[
    H_0^\text{global}: H_{0,1}\cap \cdots \cap H_{0,n}.
\]
The test statistics $T_1,\ldots, T_n$ may exhibit unknown dependence structures among each other.

For the heavy-tailed combination tests, we apply a transformation of the p-values into quantiles of heavy-tailed distributions. Specifically, let $F$ denote the cumulative distribution function (CDF) of the heavy-tailed distribution and $Q_F$ represent its quantile function, defined as
\[
Q_F(t) = \inf\left\{x\in\mathbb{R}:\ t\leqslant F(x)\right\}.
\]
We define the individual transformed test statistics as $\left\{X_i = Q_F(1- P_i)\right\}_{i=1}^n$. A combination test can then be constructed based on the sum $S_n = X_1+\cdots+X_n$, the average $M_n=\left(X_1+\cdots+X_n\right)/n$, or more generally, any weighted sum $S_{n, \vec{\omega}}=\sum_{i=1}^n\omega_iX_i$ with non-random positive weights $\omega_i$s.

\subsection{Tail properties of the sum $S_n$} 
\label{sec:tail_prob}
We begin by reviewing existing theoretical results on the tail properties of $S_n$.
If $X_1, \ldots, X_n$ belong to the sub-exponential family, a major class of heavy-tailed distributions, it is well-known that the tail probability of $S_n=X_1+\cdots+X_n$
is asymptotically equivalent to the sum of individual tail probabilities under the assumption that the $X_i$s are mutually independent. That is,
\begin{equation}
\label{eq:asymp_tail_probability}
\lim_{x\to+\infty}\frac{\Pr\left(S_n>x\right)}{n\bar{F}(x)}=1
\end{equation}
where $\bar{F}=1-F$ denotes the tail probability \citep{embrechts2013modelling}. 
When the independence assumption fails, previous works \citep{chen2009sums,asmussen2011efficient,albrecher2006tail,kortschak2009asymptotic,geluk2006tail,tang2008insensitivity} have shown
that \eqref{eq:asymp_tail_probability} still holds for different subclasses of sub-exponential distributions under certain assumptions of the dependence structure.

Here, we restate several key results that form the foundation of the theoretical properties of the heavy-tailed combination tests, which will be detailed in \Cref{sec:def_combination_test}. 
For any variable $X$, we denote $X^+ = \max\left(X, 0\right)$ and $X^- = \max\left(-X, 0\right)$. 
To begin, we introduce the concepts of quasi-asymptotic independence and the consistently-varying subclass $\mathscr{C}$ of sub-exponential distributions, following \cite{chen2009sums}.

\begin{definition}[Quasi-asymptotic independence]
\label{def:quasi_indep}
Two non-negative random variables $X_{1}$ and $X_{2}$ with cumulative distribution functions $F_{1}$ and $F_{2}$, are quasi-asymptotically independent if
\begin{equation}
\label{eq:quasiasymp}
    \lim _{x \to +\infty} \frac{\Pr\left(X_{1}>x, X_{2}>x\right)}{\overline{F_{1}}(x)+\overline{F_{2}}(x)}=0
\end{equation}
More generally, two real-valued random variables, $X_{1}$ and $X_{2}$, are quasi-asymptotically independent if \eqref{eq:quasiasymp} holds with $\left(X_{1}, X_{2}\right)$ in the numerator replaced by $\left(X_{1}^{+}, X_{2}^{+}\right)$, $\left(X_{1}^{+}, X_{2}^{-}\right)$, and $\left(X_{1}^{-}, X_{2}^{+}\right)$.
\end{definition}

When $X_1$ and $X_2$ have the same marginal distribution, \eqref{eq:quasiasymp} can be rewritten as  $\Pr(X_1>x \mid X_2 > x)\overset{x \to +\infty}{\to} 0$, indicating that $X_1$ and $X_2$ are independent in the tail.

\begin{definition}[Consistently-varying class $\mathscr{C}$]
\label{def:classC}
A distribution with the cumulative distribution function $F(\cdot)$ is in class $\mathscr{C}$ if \begin{equation*}
    \lim_{y\to1^+}\liminf_{x\to+\infty}\frac{\bar{F}(xy)}{\bar{F}(x)}=1\text{ or }\lim_{y\to1^-}\limsup_{x\to+\infty}\frac{\bar{F}(xy)}{\bar{F}(x)}=1
\end{equation*}
\end{definition}

Theorem 3.1 in \citet{chen2009sums} 
established the asymptotic tail probability of $S_n$ for distributions within  $\mathscr{C}$, provided that quasi-asymptotic independence holds.

\begin{restatable}[Theorem 3.1 of \citet{chen2009sums}]{theorem}{SumTailProbEstimate}
\label{thm:sum_tail_prob_estimation}
Let $X_{1}, \ldots, X_{n}$ be $n$ pairwise quasi-asymptotically independent real-valued random variables with distributions $F_{1}, \ldots, F_{n} \in \mathscr{C}$, respectively. Denote $S_n=\sum_{i=1}^nX_i$. Then, it holds that 
\begin{equation}
\label{eq:sum_tail}
\lim_{x\to+\infty}\frac{\Pr\left(S_{n}>x\right)}{\sum_{i=1}^{n} \overline{F_{i}}(x)}=1 \ .
\end{equation}
\end{restatable}

The asymptotic equivalence \eqref{eq:sum_tail} can hold for broader subclasses of heavy-tailed distributions beyond $\mathscr{C}$ under stronger dependence assumptions. For instance, \citet{geluk2009asymptotic} provided the necessary 
dependence structure requirements for this equivalence to hold for dominated-varying tailed and long-tailed random variables. 
Additionally, \citet{asmussen2011efficient} verified this for log-normal distributions when coupled with a Gaussian copula. However, \citet{botev2017accurate} showed that convergence in \eqref{eq:sum_tail} can be extremely slow for log-normal distributions, requiring the tail probability to be as small as $10^{-233}$ to achieve reasonable approximations. 

Moreover, researchers have observed asymptotic equivalence between the tail probability of the $S_n$ and that of $\max(X_1, \ldots, X_n)$.
\begin{restatable}{corollary}{MaxEquivSum}
\label{cor:tail_prob_max_vs_sum}
With the same setting as in \Cref{thm:sum_tail_prob_estimation}, the tail probability of the sum and the maximum has the following relationship
\begin{equation*}
\label{eq:tail_prob_max_vs_sum}
\lim_{x\to+\infty}\frac{\Pr\left(\max_{i=1,\ldots,n}X_i>x\right)}{\sum_{i=1}^n\overline{F_{i}}(x)}=\lim_{x\to+\infty}\frac{\Pr\left(S_n>x\right)}{\sum_{i=1}^n\overline{F_{i}}(x)}=1.
\end{equation*}
\end{restatable}

\begin{remark}
We provide a proof of \Cref{cor:tail_prob_max_vs_sum} in Supplementary \Cref{sec:pf_sum_max_same_tail}, 
which essentially restates earlier results \citep{geluk2006tail,tang2008insensitivity,ko2008sums}, to facilitate understanding for interested readers.
\end{remark}

\begin{table}[t]
\centering
\caption{Regularly varying tailed distributions and their tail indices. $\Phi$ is the cumulative distribution function of a standard normal distribution. $\Gamma$ is the gamma function.  $J(s,x)=\int_x^\infty t^{s-1}e^{-t}dt$ is the incomplete gamma function and $I_x(a,b)=\int_0^xt^{a-1}(1-t)^{b-1}dt/\int_0^1t^{a-1}(1-t)^{b-1}dt$ is the regularized incomplete eta function, $\bar F_t(c)$ is the survival function at $c$ of the corresponding t distribution with the same degree of freedom $\gamma$}
\begin{tabular}{@{}lll@{}}
\hline
Distributions: Survival Function & Tail index & Support \\ \hline
Cauchy: $\arctan\left(1/x\right)/\pi$                & 1 & $\mathbb{R}$ \\
Log Cauchy: $\arctan\left(1/\log x\right)/\pi$         & 0 & $\mathbb{R}^+$ \\
Levy: $2\Phi\left(x^{-1/2}\right)-1$               & $1/2$ & $\mathbb{R}^+$ \\
Pareto: $\left(1/x\right)^\gamma$, $\gamma>0$      & $\gamma$ & $[1, +\infty)$\\
Fr\'echet: $1-e^{-x^{-\gamma}}$, $\gamma>0$ & $\gamma$ & $\mathbb{R}^+$ \\
Inverse Gamma: $1-J(\gamma,1/x)/\Gamma(\gamma)$, $\gamma>0$ & $\gamma$ & $\mathbb{R}^+$\\ 
Log Gamma: $1-J(1,\gamma\log x)$, $\gamma>0$ & $\gamma$ & $\mathbb{R}^+$\\ 
Student's t: $I_{\gamma/(x^2+\gamma)}\left(\gamma/2,1/2\right)/2$, $\gamma>0$  & $\gamma$ & $\mathbb{R}$\\ 
Left-truncated t: $I_{\gamma/(x^2+\gamma)}\left(\gamma/2,1/2\right)/(2\bar F_t(c))$, $\gamma>0$  & $\gamma$ & $[c, +\infty)$\\ 
\hline
\end{tabular}
\label{tab:regular_distributions}
\end{table}

\Cref{tab:regular_distributions} presents a list of common distributions in $\mathscr{C}$. All of these distributions also belong to a smaller subclass, the regularly varying tailed distributions $\mathscr{R}$, defined as follows:

\begin{definition}[Regularly varying tailed class $\mathscr{R}_{-\gamma}$]
\label{def:classR}
A distribution $F$ is in class $\mathscr{R}_{-\gamma}$ if for some $\gamma\geqslant0$ and any $y>0$
\[
\lim_{x\to+\infty}\frac{\bar{F}(xy)}{\bar{F}(x)}=y^{-\gamma}.
\]
\end{definition}

Following \citet{cline1983infinite}, the parameter $\gamma$ is referred to as the tail index, characterizing the tail heaviness \citep{teugels1987regular} of a distribution. Distributions
with a smaller $\gamma$ exhibit heavier tails.
For example, for the Student's t distribution, $\gamma$ is the same as the degree of freedom, with the Cauchy distribution being a special case with $\gamma = 1$.

In \Cref{tab:regular_distributions}, all distributions, except for the Student's t distributions that includes the Cauchy distribution, have a lower bound in their support. In contrast, the Student's t distributions have symmetric densities around the origin, and their supports cover the entire real line. As a consequence, when $p_i$ approaches $1$, the transformed test statistics $X_i$ can become substantially negative, which may affect both the power and type-I error control in the associated combination tests. To address this issue, we introduce the left-truncated Student's t distribution in \Cref{tab:regular_distributions}, defined as a conditional Student's t distribution with a left-bounded support interval of $[c, +\infty)$. Specifically, we define $F_{t, \gamma}(x)=\Pr(X \leqslant x)$ with $X$ following a Student's t distribution with degree of freedom $\gamma$. The cumulative distribution function of the left-truncated t distribution is
\[
F(x) = \Pr(X\leqslant x \mid X\geqslant c)=\frac{F_{t, \gamma}(x) - F_{t, \gamma}(c)}{1-F_{t, \gamma}(c)},\quad x\geqslant c.
\]
With this definition, the left-truncated t distribution remains a regularly varying tailed distribution with the same tail index $\gamma$, as proved in \Cref{prop:left_truncated_cauchy_cdf}. In our experiments, we vary the truncation level $c$ by setting $c$ as the $1-p_0$ quantile of the t distribution with the same tail index $\gamma$, and we refer to $p_0$ as the truncation threshold. This approach allows us to explore the effects of different levels of truncation on the performance of combination tests in practice.

\subsection{Asymptotic validity of the heavy-tailed combination tests}
\label{sec:def_combination_test}
The asymptotic validity of heavy-tailed transformation-based combination tests can be established based on \Cref{thm:sum_tail_prob_estimation}. 
 In particular, \citet{liu2020cauchy} demonstrated the asymptotic validity of the Cauchy combination test. Extending this work, \citet{fang2021heavy} expanded these results to cover regularly varying distributions under additional constraints. However, both results are only limited to two-sided p-values, which are always positively dependent.
 In this section, we present a unified theory for the asymptotic validity of the heavy-tailed combination tests that accommodates both one-sided and two-sided p-values. 

We first define combination tests applying the sum $S_n$, directly inspired by \Cref{thm:sum_tail_prob_estimation}.

\begin{definition}[Combination test]
\label{def:standard_combination_test}
Let $F$ be the cumulative distribution function of a distribution in $\mathscr{R}_{-\gamma}$. The combination test approximates the tail probability $\Pr(S_n > x)$ by $n\bar F(x)$. Specifically, the combined p-value is defined as $n\bar F(S_n)$, and the corresponding decision function at the significance level $\alpha$ is 
\begin{equation}
\label{eq:combination_test_standard}
\phi_{\text{std}}^F = 1_{\left\{S_n> Q_F\left(1-\alpha/n\right)\right\}}.
\end{equation}
\end{definition}

In addition to the sum $S_n$, the widely accepted Cauchy and harmonic combination tests, as introduced by \citet{liu2019acat} and \cite{wilson2019harmonic}, utilize the average $M_n$ and directly approximate the tail probability $\Pr(M_n > x)$ using $\bar F(x)$. Indeed, any regularly varying tailed distribution with tail index $\gamma = 1$ can be used to define a similar average-based combination test:

\begin{definition}[Average-based combination test]
\label{def:average_combination_test}
Let $F$ be the cumulative distribution function of a distribution in $\mathscr{R}_{-1}$. The average-based combination test approximates the tail probability $\Pr(M_n > x)$ by $\bar F(x)$. Specifically, the combined p-value is defined as $\bar F(M_n)$ and the corresponding decision function at the significance level $\alpha$ is 
\begin{equation}
\label{eq:combination_test_mean}
\phi_{\text{avg}}^F = 1_{\{M_n> Q_F(1-\alpha)\}}.
\end{equation}
\end{definition}

More generally, one can define a weighted combination test, which includes both the tests defined in \Cref{def:standard_combination_test,def:average_combination_test} as special cases. As noted in \cite{liu2020cauchy} and \cite{fang2021heavy}, the weighted test can incorporate prior information on the importance of each base hypothesis to enhance power. 

\begin{definition}[Weighted combination test] 
\label{def:weighted_combination_test}
Let $F$ be the cumulative distribution function of a distribution in $\mathscr{R}_{-\gamma}$ and let $\vec{\omega}=\left(\omega_1,\ldots,\omega_n\right)\in\mathbb{R}_+^n$ be a non-random weight vector associated with each hypothesis.
Define the weighted sum as $S_{n, \vec{\omega}}=\sum_{i=1}^n\omega_iX_i$ and let $\kappa = \sum_{i = 1}^n \omega_i^{\gamma}$ where $\omega_i^{\gamma}$ is the $\gamma$th power of $\omega_i$. 
Then the weighted combination test approximates the tail probability $\Pr(S_{n, \vec{\omega}} > x)$ by $\kappa \bar F(x)$. Specifically, the combined p-value is defined as $\kappa \bar F(S_{n, \vec{\omega}})$ and the corresponding decision function at the significance level $\alpha$ is 
\begin{equation}
\label{eq:combination_test_weighted}
 \phi_{\text{wgt}}^{F, \vec{\omega}} = 1_{\left\{S_{n, \vec{\omega}}> Q_F\left(1-\alpha/\kappa\right)\right\}}.
\end{equation}
 
\end{definition}

\begin{remark}
   The sum-based and average-based combination test in \Cref{def:standard_combination_test,def:average_combination_test} are special cases of the weighted combination tests with uniform weights $\omega_i = 1$ or $\omega_i = 1/n$. Although the weighted combination test is not scale-free regarding the weights, empirical simulations suggest that the weight scaling has minimal practical impact.
\end{remark}

The asymptotic validity of the combination tests in \citet{liu2020cauchy} and \citet{fang2021heavy} relies on pairwise bivariate normality of the test statistics $\{T_i\}_{i=1}^n$, ensuring pairwise quasi-asymptotic independence as required by \Cref{thm:sum_tail_prob_estimation}. Under the same assumption, we can establish the asymptotic validity for the combination tests defined in \Cref{def:standard_combination_test,def:average_combination_test,def:weighted_combination_test}. Additionally, the asymptotic result is uniform in the nuisance parameters, particularly pairwise correlation $\rho_{ij}$s, if we impose mild constraints on them.

\begin{restatable}{theorem}{AsympErrorControl}
\label{thm:asymp_type_I_error_control}
Assume that the test statistics $\{T_i\}_{i=1}^n$ are pairwise normal with correlations $\rho_{ij} \in [-\rho_0, \rho_0]$ ($\rho_0>0$) and are marginally following standard normal distributions under the global null. 
Then, the type-I error of the tests defined in \Cref{def:standard_combination_test,def:average_combination_test,def:weighted_combination_test} using two-sided p-values $\left\{P_i=2-2\Phi(|T_i|)\right\}_{i=1}^n$ satisfies
\begin{equation}
\label{eq:asymp_type_I_error_control}
\lim_{\alpha\to0^+}\sup_{\forall~i\neq j,~\rho_{ij}\in[-\rho_0,\rho_0]}\frac{\Pr_{H_0^\text{global}}\left(\phi_{\text{comb}}^F=1\right)}{\alpha}=1,
\end{equation}
where $\phi_{\text{comb}}^{F}$ is the test's decision function defined in \eqref{eq:combination_test_standard} to \eqref{eq:combination_test_weighted}. 
For the combination tests with one-sided p-values $\left\{P_i=1-\Phi(T_i)\right\}_{i=1}^n$, the relationship \eqref{eq:asymp_type_I_error_control} still holds with an additional assumption that the cumulative distribution function $F(\cdot)$ satisfies that $\bar F(x)\geqslant F(-x)$ for sufficiently large $x$. 
\end{restatable}

\begin{remark}
Our analysis considers fixed $n$. Prior work \citep{liu2020cauchy,long2023cauchy} established the asymptotic validity of the Cauchy combination test as $n\to\infty$, assuming $n$ grows at a slower rate than the decay of $\alpha\to0$. In addition, \citet{vovk2020combining} introduced an adjusted rejection threshold for the harmonic mean p-value to ensure validity as $n\to \infty$, even under arbitrary dependence among the p-values.
\end{remark}

The asymptotic validity of combination tests hinges on proving the pairwise asymptotic independence of the transformed statistics $\{X_i\}_{i = 1}^n$. \Cref{thm:asymp_type_I_error_control} provides a stronger asymptotic validity than previous studies \citep{liu2020cauchy,fang2021heavy} as uniform convergence is guaranteed over the set of correlation matrices. It also imposes minimal distributional requirements on $F$ and further addresses one-sided p-values. 
Unlike two-sided p-values, which are always non-negatively correlated under bivariate normality as stated in \Cref{prop:two_sided_p_vals_pos_dep}, one-sided p-values can exhibit negative correlations. To establish the test's asymptotic validity, an additional constraint is required that $\bar F(x)\geqslant F(-x)$ for sufficiently large $x$. 
This condition, met by all distributions in \Cref{tab:regular_distributions}, ensures that the left tail is either absent or lighter than the right tail.

\Cref{thm:asymp_type_I_error_control} requires no $(i,j)$ pair has perfect correlaion. When $\rho_{ij} = \pm 1$, though the transformed statistics $X_i$ and $X_j$ are no longer quasi-asymptotically independent, a weaker form of asymptotic validity still holds when the tail index $\gamma \leqslant 1$, as stated below. 
\begin{restatable}{corollary}{ValidityPerfectCorrelation}
\label{cor:weak_validity}
Under assumptions of \Cref{thm:asymp_type_I_error_control} while allowing $\rho_{ij}=\pm 1$ for any $(i,j)$ pairs, if additionally the tail index $\gamma \leqslant 1$, then the combination tests defined in \Cref{def:standard_combination_test,def:average_combination_test,def:weighted_combination_test} using two-sided p-values are still asymptotically valid satisfying
\begin{equation}
\label{eq:weaker_validity}
\limsup_{\alpha\to0^+}\sup_{\Sigma\in B_{\rho_0}}\frac{\Pr_{H_0^\text{global}}\left(\phi_{\text{comb}}^F=1\right)}{\alpha}\leqslant 1,
\end{equation}
where $\Sigma=(\rho_{ij})_{n\times n}$ is the correlation matrix of test statistics $T_i$s, and $B_{\rho_0}=\{\Sigma\in[0,1]^{n\times n}: \forall~i\neq j~\mathrm{either}~\rho_{ij}\in[-\rho_0,\rho_0]~\mathrm{or}~|\rho_{ij}|=1\}$.
For combination tests with one-sided p-values, \eqref{eq:weaker_validity} still holds
if $F(\cdot)$ has a lower bounded support or satisfies $\bar F(x) = F(-x)$ for all $x\in\mathbb R$. 
\end{restatable}

As a special case of \Cref{cor:weak_validity}, when $\rho_{ij}\equiv 1$ for all pairs of test statistics, it holds that
\begin{restatable}{corollary}{ValidityPerfectCorrelationII}
\label{cor:weak_validity2}
Under conditions of \Cref{cor:weak_validity}, if $\rho_{ij}\equiv 1$ for all $(i,j)$ pairs, then the combination tests defined in \Cref{def:standard_combination_test,def:average_combination_test,def:weighted_combination_test} using either one-sided or two-sided p-values satisfies 
\begin{equation*}
\label{eq:weaker_validity2}
\lim_{\alpha\to0^+}\frac{\Pr_{H_0^{\text{global}}}\left(\phi_{\text{comb}}^F=1\right)}{\alpha}=\frac{\left(\sum_{i=1}^n\omega_i\right)^\gamma}{\sum_{i=1}^n\omega_i^\gamma}.
\end{equation*}
In particular, when all weights are 1, the limit is $n^{\gamma-1}$. 
\end{restatable}

\subsection{Asymptotic equivalence to the Bonferroni test}
\label{sec:compared_to_bon}

In this subsection, we explore the relationship between the heavy-tailed combination tests and the Bonferroni test. We begin by defining the weighted Bonferroni test, a generalization of the standard Bonferroni test that incorporates pre-chosen weights.
\begin{definition}[Weighted Bonferroni test]
\label{def:bonferroni_test}
Let $P_1,\ldots,P_n$ be the p-values and $\vec{\omega}=\left(\omega_1,\ldots,\omega_n\right)\in\mathbb{R}_+^n$ be a non-random weight vector satisfying $\sum_{i = 1}^n \omega_i = 1$, then the weighted Bonferroni test at the significance level $\alpha$ has the decision function
\begin{equation}
\label{eq:bonferroni_test}
\phi_{\text{bon}}^{\vec{\omega}} = 1_{\left\{\min_{i=1,\ldots,n}P_i/\omega_i<\alpha\right\}}.
\end{equation}
\end{definition}
It has been shown that the weighted Bonferroni test controls type-I error under any dependence structure \citep{genovese2006false}. The standard Bonferroni test is a special case where $\omega_i = 1/n$.

Given that the set 
$\left\{X_i=Q_F(1-P_i)>Q_F\left(1 - \alpha/n\right)\right\}=\left\{P_i<\alpha/n\right\}$, 
 the decision function of the Bonferroni test can be rewritten as 
 \[
 \phi_{\text{bon}} = 1_{\{n\min_{i=1,\ldots,n}P_i<\alpha\}} =1_{\left\{\max_{i=1,\ldots,n}X_i>Q_F\left(1- \alpha/n\right)\right\}}.
 \]
Thus, \Cref{cor:tail_prob_max_vs_sum} implies that the type-I error of the Bonferroni test and the standard combination tests are asymptotically the same. Given this, we investigate whether the combination tests are indeed asymptotically equivalent to the Bonferroni test. We find out that the rejection regions of the weighted combination tests converge to those of a weighted Bonferroni test as $\alpha\to0$.

\begin{restatable}{theorem}{bonferroni}
\label{thm:same_as_bonferroni}
Assume that the test statistics $\{T_i\}_{i=1}^n$ are pairwise normal with correlations $\rho_{ij} \in [-\rho_0, \rho_0]$ and have a common marginal variance $1$. Means of marginal normals are all finite. Then for two-sided p-values, when $\alpha\to 0$, any weighted heavy-tailed combination test defined in 
\Cref{def:weighted_combination_test} is asymptotically equivalent to a weighted Bonferroni test. Namely,
\begin{equation*}
\label{eq:asymptotic_equialence_to_bonferroni}
\lim_{\alpha\to0^+}\sup_{\forall~i\neq j,~\rho_{ij}\in[-\rho_0,\rho_0]}\frac{\Pr\left(\phi_{\text{wgt}}^{F,\vec{\omega}}\neq\phi_{\text{bon}}^{ {\vec{\omega_*}}}\right)}{\min\left\{\Pr\left(\phi_{\text{wgt}}^{F,\vec{\omega}}=1\right),\Pr\left(\phi_{\text{bon}}^{ {\vec{\omega_*}}}=1\right)\right\}}=0 \ ,
\end{equation*}
where $\phi_{\text{wgt}}^{F,\vec{\omega}}$ is defined in \eqref{eq:combination_test_weighted}, $\phi_{\text{bon}}^{\vec {\omega_*}}$ is defined in \eqref{eq:bonferroni_test}, and $ \vec{\omega_*} = (\omega_{*,1}, \ldots, \omega_{*,n})$ with $\omega_{*,i} = \omega_i^\gamma/\sum_{i=1}^n\omega_i^\gamma$. For one-sided p-values, the conclusion retains when further assuming that the CDF $F(\cdot)$ satisfies that $\bar F(x)\geqslant F(-x)$ for sufficiently large $x$.
\end{restatable}

\Cref{thm:same_as_bonferroni} establishes the asymptotic equivalence between the combination tests and the Bonferroni test under any hypothesis configuration, provided that the test statistics are pairwise normal and not perfectly correlated. As the significance level $\alpha$ approaches zero, the rejection regions of both the combination tests and the Bonferroni test shrink, and the differences between these rejection regions diminish at a higher order. This equivalence does require that the test statistics are not perfectly correlated, so that they are quasi-asymptotically independent. 

\begin{figure}[t!]
    \centering
    \includegraphics[width=0.7\textwidth]{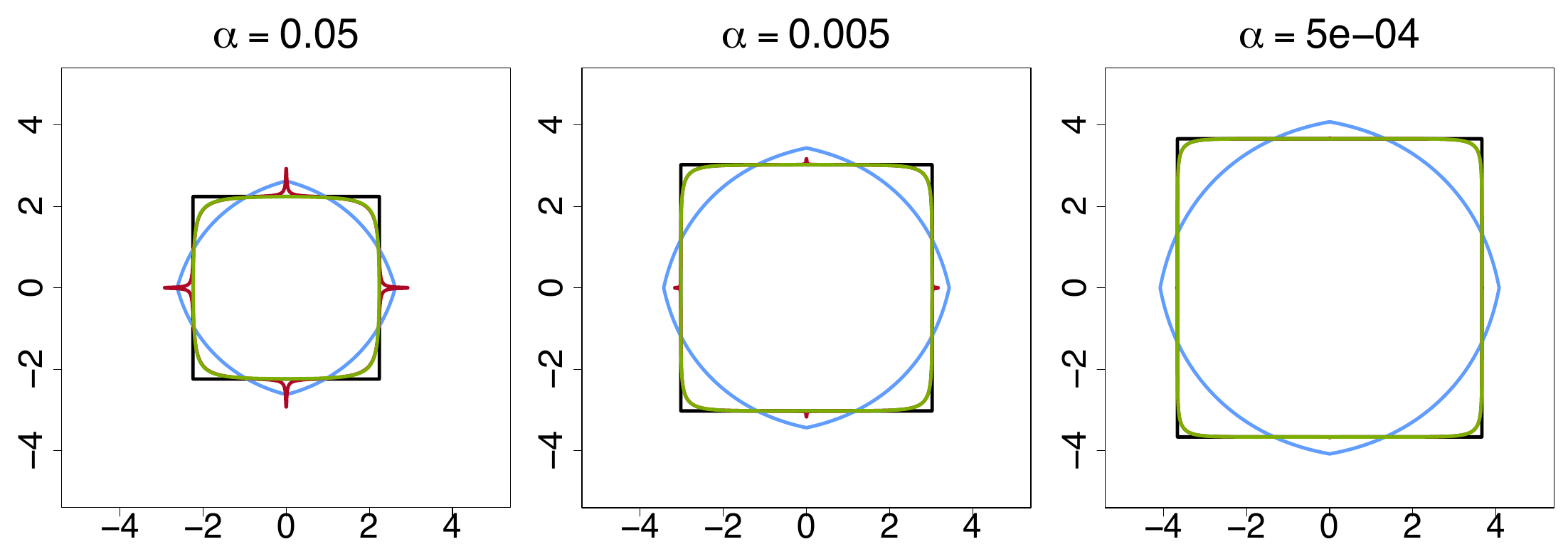} 
    \caption{Rejection regions for Bonferroni (black), Fisher (blue), Cauchy (red), and Fr\'echet $\gamma=1$ (green) combination tests for a two-sided test in test statistics space when the number of base hypotheses $n=2$ and at different significance level $\alpha$. The boundaries of the rejection regions are shown with different colored lines, and the rejection regions are the areas outside of these boundaries that do not include the origin.}
\label{fig:rejection_region_comp}
\end{figure}

To provide an intuitive understanding of \Cref{thm:same_as_bonferroni}, \cref{fig:rejection_region_comp} compares the rejection regions of various tests in the test statistics space for two-sided p-values with $n = 2$. 
The key takeaway is that the heavy-tailed nature of the transformation distribution yields nearly square rejection regions, which closely resemble those of the Bonferroni test as $\alpha$ decreases. In contrast, for combination tests relying on light-tailed distributions, such as Fisher's combination method, different rejection region shapes persist regardless of how small $\alpha$ becomes. Thus, in the asymptotic regime where these heavy-tailed combination tests are proven valid and when the individual test statistics are not perfectly correlated, there is no power gain over the Bonferroni test.

\section{Empirical evaluations of the heavy-tailed combination tests under asymptotic independence}
\label{sec:effect_of_heaviness}
\subsection{Empirical validity of the combination tests}
\label{subsec:validity_simulations}
The theoretical results in \Cref{sec:setup} provide valuable insight into the heavy-tailed combination tests. However, it is unclear to what extent these asymptotic results align with their practical performance at finite significance levels. We aim to conduct an empirical evaluation of the tests' validity, focusing on commonly used finite significance levels. 

For a comprehensive study, we vary the significance level $\alpha$, number of hypotheses, tail heaviness and support of the distribution, and the level of dependence among the p-values. 
Specifically, we generate test statistics as $z$-values sampled from a multivariate normal distribution with mean $\vec{\mu}=\vec{0}_n$ and covariance matrix $\Sigma_\rho$. The covariance matrix $\Sigma_\rho\in\mathbb{R}^{n\times n}$ has 1s on the diagonal and a common value $\rho$ off the diagonal, representing varying degrees of dependence. We assess performance at three values of $\rho$, $0,0.5,$ and $0.99$, in line with no, moderate, and strong dependence. We calculate two-sided p-values from the $z$-values and conduct the combination tests based on different heavy-tailed distributions from four distribution families, the Student's t, Fr\'echet, Pareto, and inverse Gamma distributions. Each family has have a tunable tail index $\gamma$ quantifying the tail heaviness, with a larger $\gamma$ corresponding to a lighter tail. We vary this $\gamma$ from $0.7$ to $1.5$ by $0.01$ for all four distribution families. We also include the Bonferroni test and the Cauchy combination test as baselines. For significance levels, we adopt $\alpha=0.05$ and $5\times10^{-4}$ to account for different testing scenarios. The standard $0.05$ is commonly used for a single global null hypothesis, while $5\times10^{-4}$ reflects the stricter threshold needed in genetic applications, where multiple testing adjustments lower the effective significance level for individual p-values. For the number of hypotheses, we consider $n = 5$ and $100$. Each scenario is replicated $10^6$ times to calculate the empirical type-I errors of the tests. 

\begin{figure}[t!]
     \centering
     \begin{subfigure}[h]{0.8\textwidth}
         \centering
         \includegraphics[width=\textwidth]{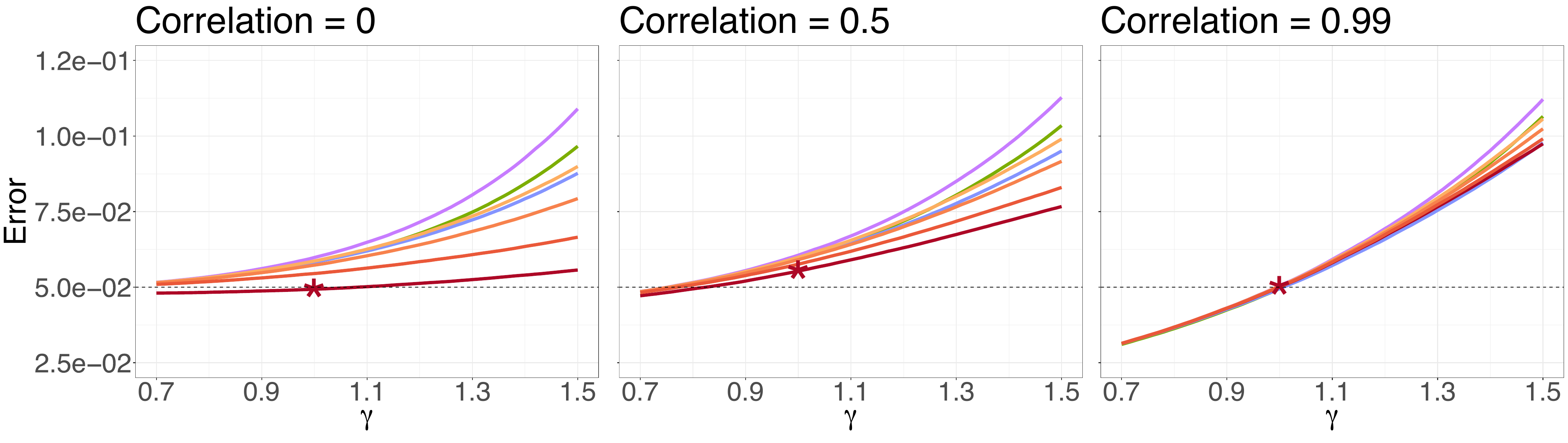}
         \caption{$\alpha=0.05$}
         \label{fig:error_5_0.05}
     \end{subfigure}
     \hfill
     \begin{subfigure}[h]{0.8\textwidth}
         \centering
         \includegraphics[width=\textwidth]{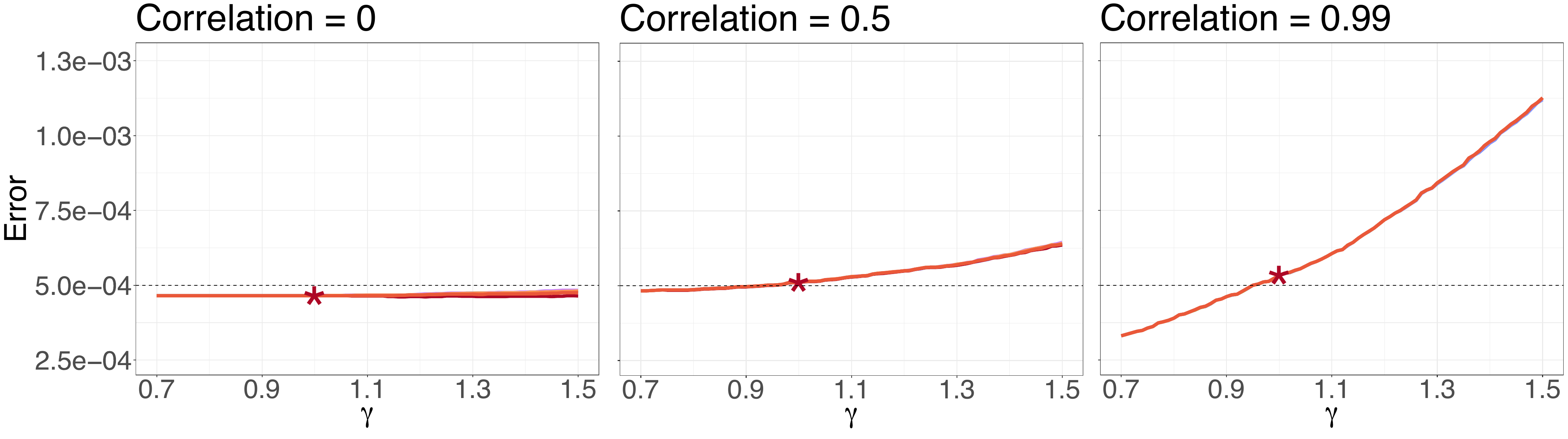}
         \caption{$\alpha=5\times10^{-4}$}
         \label{fig:error_5_5e-4}
     \end{subfigure}
\caption{The type-I error of the combination test when $n = 5$ with different distributions: Cauchy (star point), inverse Gamma (blue), Fr\'echet (green), Pareto (purple), student t (red), left-truncated t with truncation threshold $p_0=0.9$ (dark orange), left-truncated t with truncation threshold $p_0=0.7$ (orange), left-truncated t with truncation treshold $p_0=0.5$ (light orange). The vertical axis represents the empirical type-I error, and the horizontal axis stands for the tail index $\gamma$.}
\label{fig:n_5_gamma_vs_error}
\end{figure}

\Cref{fig:n_5_gamma_vs_error} and \ref{fig:n_100_gamma_vs_error} present the results for $n = 5$ and $100$. When $\alpha=0.05$ and $\gamma = 1$, only the Cauchy combination test can strictly control empirical type-I error under independence, and no method achieves strict control when correlation $\rho_{ij} = 0.5$. Smaller $\alpha$ improves error control and leads to a flatter curve across $\gamma$, consistent with the theoretical limit. Regarding the impact of tail heaviness on validity, differences between various distribution families diminish as $\alpha$ decreases, making the empirical validity of the tests primarily dependent on the tail index $\gamma$.  A larger $\gamma$ corresponds to a lighter tail, which results in poorer type-I error control for any finite $\alpha$. Empirically, a type-I error control is approximately achieved when $\gamma \leqslant 1$.  
Distribution support also plays a role in type-I error control. Tests based on t distributions, which allow negative transformed statistics, outperform those using distributions with only positive support at $\alpha = 0.05$.
To examine this further, left-truncated t-distributions with different truncation thresholds $p_0 = 0.5, 0.7$ and $0.9$ are adopted.
As shown in \cref{fig:n_5_gamma_vs_error} and \ref{fig:n_100_gamma_vs_error}, their empirical type-I errors fall between those of the original t distributions and other distribution families. This suggests that a wider support to the left of the real line tends to reduce the type-I error of the combination tests.

\begin{figure}[t!]
     \centering
     \begin{subfigure}[h]{0.8\textwidth}
         \centering
         \includegraphics[width=\textwidth]{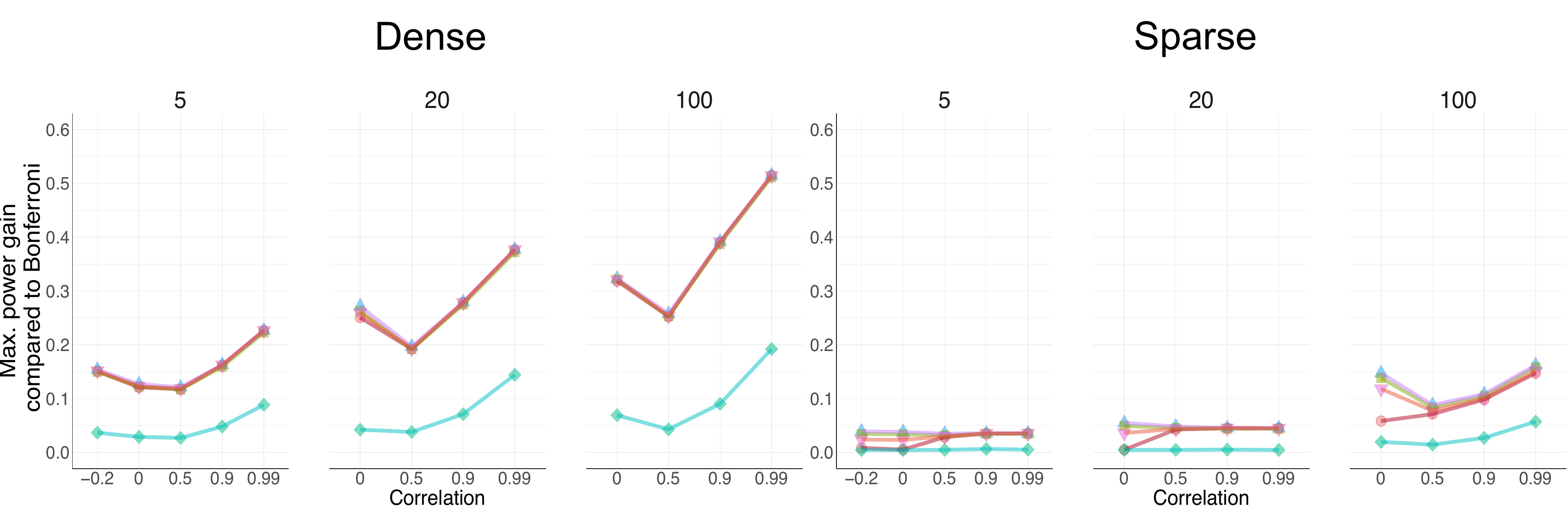}
         \caption{$\alpha=0.05$}
         \label{fig:power_0.05}
     \end{subfigure}
     \hfill
     \begin{subfigure}[h]{0.8\textwidth}
         \centering
         \includegraphics[width=\textwidth]{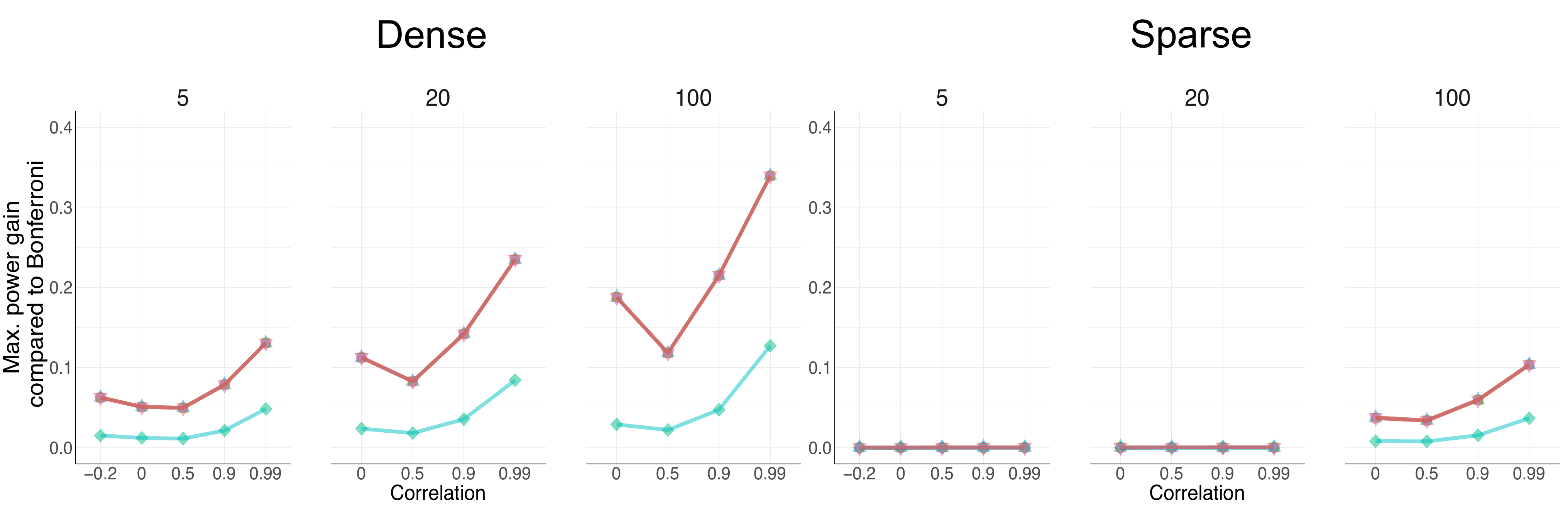}
         \caption{$\alpha=5\times10^{-4}$}
         \label{fig:power_5e-4}
     \end{subfigure}
\caption{Power comparison with the Bonferroni test of the combination test with different distributions: Levy (turquoise with diamond dot), Cauchy (red with round dot), Fr\'echet $\gamma=1$ (green with square dot), Pareto $\gamma=1$ (purple with triangular dot), left-truncated $t_1$ with truncation threshold $p_0=0.9$ (dark orange with inverted-triangle dot). Left plots correspond to dense signals and right ones correspond to sparse signals. }
\label{fig:power_gain}
\end{figure}

Additionally, we have investigated the type-I error control of the combination tests when the base p-values are negatively correlated by generating one-sided p-values. Results are shown in \Cref{tab:comp_neg_dep}. We observe that when the p-values are negatively correlated, the Cauchy combination test can be even more conservative than the Bonferroni test due to its unbounded support. This undesired conservativeness can be mitigated by using a left-truncated t-distribution with a moderate truncation threshold. For more details, see Supplementary \Cref{sec:negative_p_values}.

\subsection{Empirical comparison with the Bonferroni test}
\label{sec:power_with_bonferroni}
Theoretically, we have shown that the combination tests are asymptotically equivalent to the Bonferroni test for pairwise normal test statistics. Empirically, we aim to compare their power at finite significance levels and determine how small $\alpha$ needs to be for the asymptotic results to appear. Specifically, we evaluate significance levels $\alpha = 0.05$ and $5\times 10^{-4}$ while also approximating the asymptotic setting by letting $\alpha$ approach 0.

We start with assessing the power of the combination tests and the Bonferroni test at finite $\alpha$s. Specifically, we define power as $\Pr_{H_{1,\text{global}}}(\text{reject global null})$. We adopt the same simulation settings as in \Cref{subsec:validity_simulations}, generating one-sided p-values to obtain both positive and negative correlated p-values. We introduce both sparse and dense signals in the mean vector $\vec{\mu}$ and consider three different numbers of hypotheses $n=5,20,100$. The dense signals are generated as $\vec{\mu}=\vec{\mu}_n=(\mu,\mu,\ldots,\mu)\in\mathbb{R}^{n}$. For sparse signals, we employ $\vec{\mu}=({\vec{0}_4},\mu)\in\mathbb{R}^{5}$, $\vec{\mu}=({\vec{0}_{19}},\mu)\in\mathbb{R}^{20}$, and $\vec{\mu}=({\vec{0}_{95}},\vec{\mu}_5)\in\mathbb{R}^{100}$ as signal vectors. 
The parameter $\mu$ ranges from 0 to 6, ensuring that all testing methods can reach a power of $1$, in increments of $0.5$. For the covariance matrix $\Sigma_\rho$, we select $\rho=0,0.5,0.9,0.99$ and also consider the negative correlation $\rho=-0.2$, to ensure the covariance matrix is positive definite, for $n=5$. Each scenario is replicated $10^6$ times to calculate the empirical power of the tests. 

\Cref{fig:power_gain} displays the maximum power difference between the combination tests using the Cauchy, truncated $t_1$, Pareto, Fr\'echet, and Levy distributions, compared to the Bonferroni test when allowing $\mu$ to increase until all methods reach a power of $1$. The truncation threshold for the $t_1$ distribution is set at $p_0 = 0.9$. The Cauchy, truncated $t_1$, Fr\'echet, and Pareto distributions share a tail index $\gamma = 1$, whereas the Levy distribution has a tail index of $0.5$, resulting in a smaller power difference compared to the Bonferroni test. 

Our findings reveal that combination tests can achieve higher power at finite significance levels, particularly in situations where signals are dense. This remains the case for the Cauchy combination test even when p-values are negatively correlated, a setting in which it tends to be overly conservative. This likely stems from the nature of the combination test, which synthesizes signals from multiple sources rather than relying on a single dominant signal. These results suggest that the onset of asymptotic equivalence may occur at much smaller values of $\alpha$ compared to that for asymptotic validity, especially when signals are dense.

\begin{figure}[t!]
\centering
\includegraphics[width=0.8\textwidth]{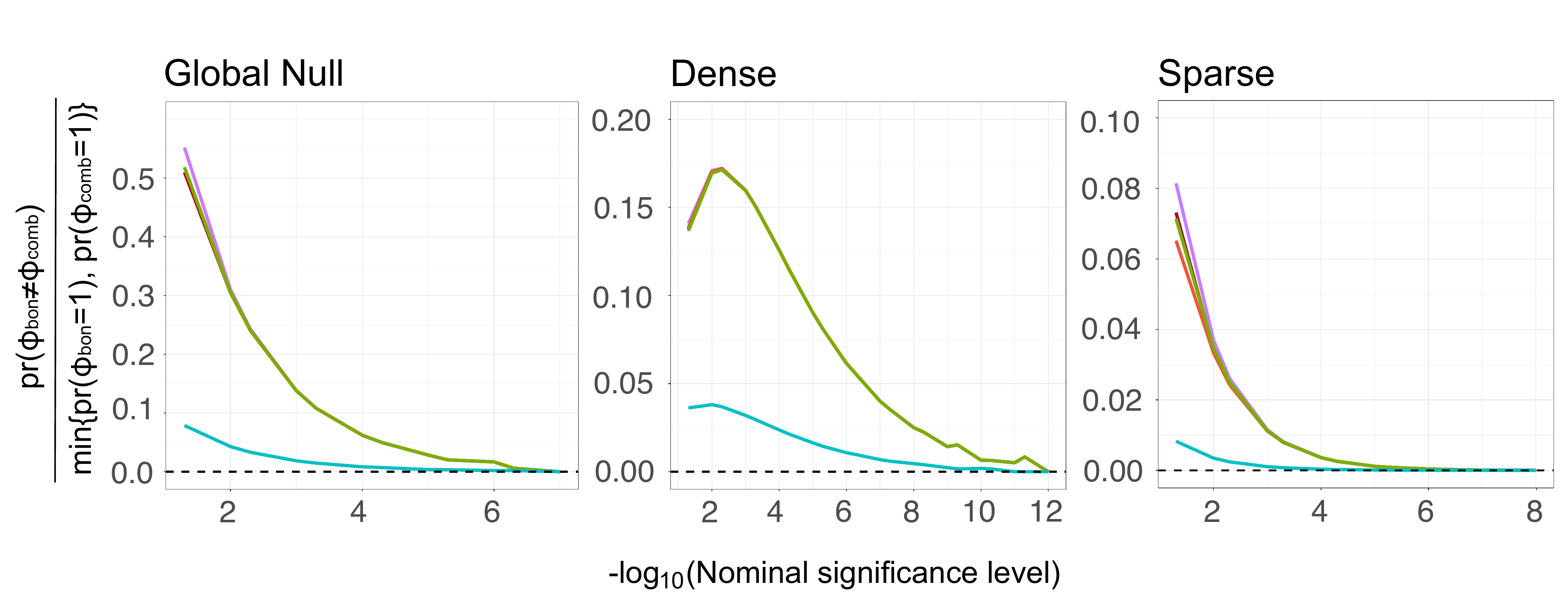}
\caption{The difference between the combination test and Bonferroni test diminishes as the significance level converges to 0. The left plot simulates the ratio in \Cref{thm:same_as_bonferroni} with fixed $\rho_{ij}=0.5$ under the global null. Right plots simulate the same ratio under global alternative with dense and sparse signals. The combination tests are with different distributions: Levy (turquoise), Cauchy (red), Fr\'echet $\gamma=1$ (green), Pareto $\gamma=1$ (purple), left-truncated $t_1$ with truncation threshold $p_0=0.9$ (dark orange). The number of repeated simulations is $10^8$.}
\label{fig:test_diff}
\end{figure}

To further investigate the asymptotic equivalence between the combination tests and the Bonferroni test, we examine how the size of their non-overlapping rejection regions evolves as $\alpha$ approaches $0$. Using the same settings as earlier in this section with $n=5$ and $\rho=0.5$, we fix the signal level $\mu=2$ to ensure the power difference between the two tests is not negligible.  We consider three mean vectors: $\vec{\mu}={\vec{0}_5}$ (global null),$({\vec{0}_4},2)$ (sparse signal), ${\vec{2}_5}$ (dense signal), allowing us to compare their performance under different scenarios. 
As shown in \cref{fig:test_diff}, the difference, quantified by the probability ratio between the overlapping rejection region and individual rejection regions, converges to zero as $\alpha$ decreases, being consistent with the asymptotic equivalence established in \Cref{thm:same_as_bonferroni}.

Since the Bonferroni test is known to suffer under strong dependence, we also compare the combination tests against the adjusted Bonferroni method, minP. Specifically, we calibrate the cutoff for $\min(p_1,\ldots,p_n)$ using Monte Carlo sampling from the true data-generating model to ensure the actual type I error matches the nominal level $\alpha$ (\Cref{tab:corrected_error}). We replicate the simulation settings from \cref{fig:power_gain}, replacing Bonferroni with minP as the baseline. As shown in \cref{fig:power_gain_with_correction}, combination tests outperform minP when signals are dense and test statistics are weakly correlated, consistent with findings in \citet{liu2020cauchy}. However, minP relies on knowledge of the dependence structure among p-values, limiting its practicality in many applications and making it computationally intensive.

\section{The combination test under asymptotic dependence}
\label{sec:dependence}
Although heavy-tailed combination tests are typically employed when p-values have unknown dependence, they do not guarantee control of the type-I error under arbitrary dependence structures, even asymptotically. One key assumption for ensuring asymptotic type-I error control in \Cref{sec:def_combination_test} is the requirement of quasi-asymptotic independence, which can be restrictive in practice. For instance, when the sample size is small, test statistics are likely to follow a t-distribution rather than a normal distribution. Additionally, even when the sample size is large, it can still be challenging to ensure that two dependent test statistics are pairwise normal.

The strength of asymptotic dependence between any two variables $(X_1, X_2)$ with the same marginal distribution $F$ can be quantified by the upper tail dependence coefficient \citep{joe1997multivariate}
\[
\lambda=\lim_{x\to+\infty}\Pr(X_1>x\mid X_2>x).
\label{eq:tail_dep_coef}
\]
As discussed earlier, if $X_1$ and $X_2$ are bivariate normal and are not perfectly correlated, they are quasi-asymptotic independent, and hence $\lambda = 0$. 
However, many dependent variables do not satisfy quasi-asymptotic independence. 
For instance, for bivariate t-distributed variables $(T_1, T_2)$ with degree of freedom $\nu$, variances $1$ and correlation $\rho$, their tail dependent coefficient \citep{demarta2005t} is
\[
\lambda_{\nu,\rho}=2t_{\nu+1}\left(-\left(\nu+1\times\frac{1-\rho}{1+\rho}\right)^{\frac{1}{2}}\right),
\label{eq:t_copula_tail_dep_coef}
\]
where $t_{\nu}(\cdot)$ is the cumulative distribution function of the t distribution. 
As a result,  $T_1$ and $T_2$ are never quasi-asymptotically independent, even when $\rho = 0$, due to shared covariance estimation.

To understand the sensitivity of the combination tests to violations of quasi-asymptotic independence, we generate test statistics $(T_1, \ldots, T_n)$ from a multivariate t distribution $t_{\nu}(\vec{0}_n, \Sigma_\rho)$, where $\Sigma_\rho$ is defined in \Cref{sec:effect_of_heaviness}. 
We choose an extreme degree of freedom $\nu = 2$ and set the correlation $\rho$ to $0$, $0.5$, $0.9$, and $0.99$, resulting in tail dependence indices ranging from $0.18$ to $0.91$. All base p-values are one-sided and derived from the test statistics.

\begin{table}[!ht]
\caption{Type-I error control of the combination tests when test statistics follow a multivariate t-distribution when $n = 5$. Values in parentheses
are the corresponding standard errors. For the Fr\'echet
and Pareto distributions, the tail index $\gamma = 1$. For truncated $t_1$, the truncation threshold $p_0=0.9$}
\begin{subtable}{\linewidth}
\centering
\caption{$\alpha=0.05$}
\scalebox{0.75}{
\begin{tabular}{@{}ll|rrrrrrr@{}}
\hline
$\rho$ & $\lambda_{2,\rho}$ & \multicolumn{1}{c}{Cauchy} & \multicolumn{1}{c}{Pareto} & \multicolumn{1}{c}{Truncated $t_1$} & \multicolumn{1}{c}{Fr\'echet} & \multicolumn{1}{c}{Levy} & \multicolumn{1}{c}{Bonferroni} & \multicolumn{1}{c}{Fisher} \\ \hline
& & \multicolumn{1}{c}{} & \multicolumn{1}{c}{} 
& \multicolumn{1}{c}{}  & \multicolumn{1}{c}{}  
& \multicolumn{1}{c}{}  & \multicolumn{1}{c}{}           
& \multicolumn{1}{c}{}  \\
0
& 0.18 & 2.90E-02 & 5.30E-02 & 4.73E-02 & 5.17E-02 & 3.89E-02 & 3.56E-02 & 6.26E-02\\
& & (1.68E-04) & (2.24E-04) & (2.21E-04) & (1.93E-04)  & (2.12E-04)  & (1.85E-04) & (2.42E-04)\\
& & & & & & & & \\
0.5 & 0.39 & 4.48E-02 & 5.24E-02 & 5.01E-02 & 5.13E-02 & 3.19E-02 & 2.65E-02 & 1.16E-01 \\
& & (2.07E-04) & (2.23E-04) & (2.18E-04) & (2.21E-02) & (1.76E-04) & (1.61E-04) & (3.20E-04)\\
& & & & & & & & \\
0.9 & 0.72 & 5.00E-02 & 5.09E-02 & 5.06E-02 & 4.99E-02 & 2.50E-02 & 1.67E-02 & 1.51E-01 \\
& & (2.18E-04) & (2.20E-04) & (2.19E-04) & (2.18E-04) & (1.56E-04) & (1.28E-04) & (3.58E-04)\\
& & & & & & & & \\
0.99 & 0.91 & 5.02E-02 & 5.03E-02 & 5.02E-02 & 4.92E-02 & 2.27E-02 & 1.19E-02 & 1.59E-01\\
& & (2.18E-04) & (2.18E-04) & (2.18E-04) & (2.16E-04) & (1.49E-04) & (1.09E-04) & (3.66E-04)\\ \hline
\end{tabular}
}
\end{subtable}
\hfill
\begin{subtable}{\linewidth}
\centering
\caption{$\alpha=5\times10^{-4}$}
\scalebox{0.75}{
\begin{tabular}{@{}ll|rrrrrrr@{}}
\hline
$\rho$ & $\lambda_{2,\rho}$ & Cauchy & Pareto & Truncated $t_1$ & Fr\'echet   & Levy & Bonferroni & Fisher \\ \hline
& & & & & & & & \\
0 & 0.18 & 2.48E-04 & 4.57E-04 & 4.57E-04 & 4.57E-04 & 3.49E-04 & 3.18E-04 & 2.17E-02 \\
& & (1.57E-05) & (2.14E-05) & (2.14E-05) & (2.14E-05) & (1.88E-05) & (1.78E-05) & (1.46E-04) \\
& & & & & & & & \\
0.5 & 0.39 & 3.94E-04 & 4.65E-04 & 4.65E-04 & 4.65E-04 & 3.08E-04 & 2.67E-04 & 2.63E-02 \\
& & (1.98E-05) & (2.16E-05) & (2.16E-05) & (2.16E-05) & (1.75E-05) & (1.63E-05) & (1.60E-04) \\
& & & & & & & & \\
0.9 & 0.72 & 5.20E-04 & 5.28E-04 & 5.28E-04 & 5.28E-04 & 2.37E-04 & 1.65E-04 & 3.82E-02   \\
& & (2.28E-05) & (2.30E-05) & (2.30E-05) & (2.30E-05) & (1.54E-05) & (1.28E-05) & (1.92E-04) \\
& & & & & & & & \\
0.99 & 0.91 & 5.24E-04 & 5.24E-04 & 5.24E-04 & 5.24E-04  & 2.22E-04 & 1.16E-04 & 4.25E-02  \\
& & (2.29E-05) & (2.29E-05) & (2.29E-05) & (2.29E-05) & (1.49E-05) & (1.08E-05) & (2.02E-04) \\ \hline
\end{tabular}
}
\end{subtable}
\label{tab:t_copula_5}
\end{table}

\Cref{tab:t_copula_5} and \ref{tab:t_copula_100} compare the empirical type-I errors of different combination tests at the significance level $\alpha=0.05$ and $5\times10^{-4}$, and for $n = 5$ and $n = 100$. Surprisingly, the results indicate that type-I errors remain well-controlled regardless of the tail dependence coefficient, demonstrating the robustness of the combination tests to violations of the pairwise normal assumption for the test statistics. 

Furthermore, \Cref{tab:t_copula_5} and \ref{tab:t_copula_100} suggest that the Bonferroni test tends to be exceedingly conservative when the dependence coefficient $\lambda > 0$, especially when both $n$ and $\lambda$ are large. In contrast, the combination tests based on heavy-tailed distributions with $\gamma=1$ consistently maintain a type-I error rate close to the specified significance level. Thus, we hypothesize that when test statistics are quasi-asymptotically dependent, the combination tests with a tail index $\gamma \leqslant 1$ are still asymptotically valid when $\alpha \to 0$, but they will not be asymptotically equivalent to the Bonferroni test. While the Bonferroni test can exhibit excessive conservatism, the combination tests with $\gamma = 1$ display neither conservatism nor inflation in their type-I error rates. 
For example, as discussed in \Cref{cor:weak_validity}, in situations where test statistics are perfectly correlated with $\rho = 1$, resulting in a tail dependence coefficient of $\lambda = 1$, the combination tests with $\gamma=1$ maintain an asymptotic type-I error of $\alpha$, whereas the true type-I error of the Bonferroni test is only $\alpha/n$. 

We further investigate the power gain of the combination test over the Bonferroni test when test statistics follow a multivariate t-distribution. Compared to the power comparison in \Cref{sec:power_with_bonferroni}, we replace the distribution of the test statistics from a multivariate normal distribution to a multivariate t distribution with $\nu = 2$, while keeping all other settings the same. \Cref{fig:power_gain_mvt} displays the maximum power gain of each combination test over the Bonferroni test as the power of both tests grows from $0$ to $1$ as signal strength increases. Compared to the subtle power improvement we observed for multivariate normally distributed test statistics in \cref{fig:power_gain}, the maximum power difference for multivariate t-distributed test statistics can be as large as $1$ even when signals are sparse. The power difference does not diminish even when the significance level decreases from 0.05 to $5\times10^{-4}$. 

These findings indicate a potential power advantage of the combination tests over the Bonferroni test, even in the asymptotic regime where $\alpha \to 0$, when test statistics are pairwise asymptotically dependent. 
Our empirical results indicate that, unlike in the case of asymptotic independence, combination tests can remain asymptotically valid while achieving a nontrivial power improvement over the Bonferroni test under asymptotic dependence. This highlights the potential of asymptotic dependence as a valuable framework for advancing both the theoretical and practical understanding of combination tests.

\begin{figure}[t]
     \centering
     \begin{subfigure}[h]{0.83\textwidth}
         \centering
         \includegraphics[width=\textwidth]{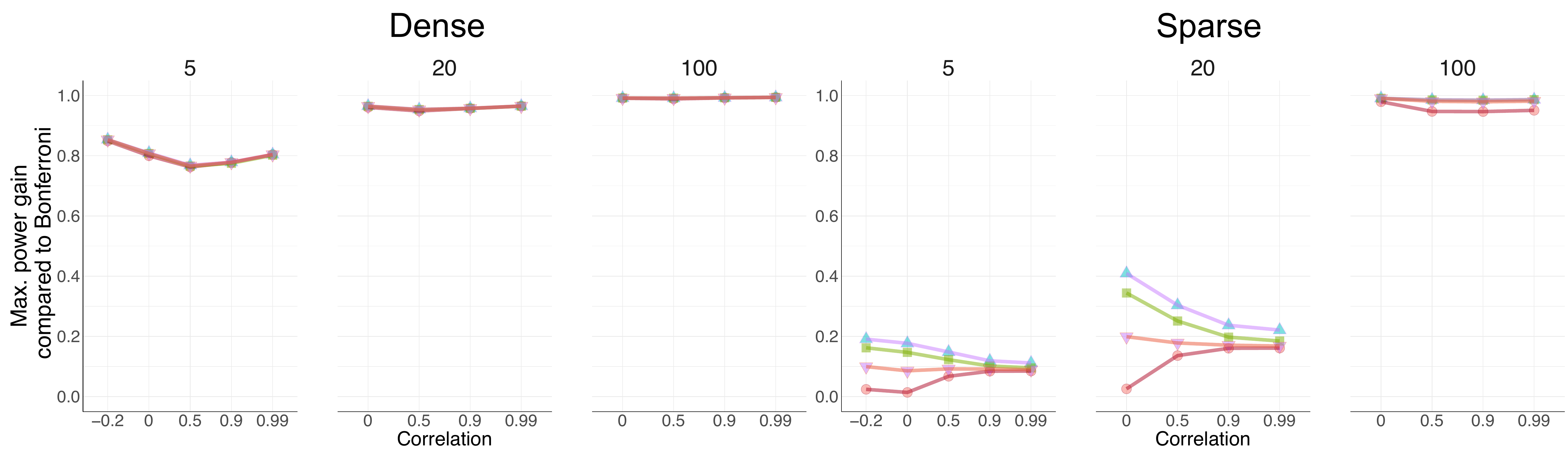}
         \caption{$\alpha=0.05$}
         \label{fig:power_0.05_mvt}
     \end{subfigure}
     \hfill
     \begin{subfigure}[h]{0.83\textwidth}
         \centering
         \includegraphics[width=\textwidth]{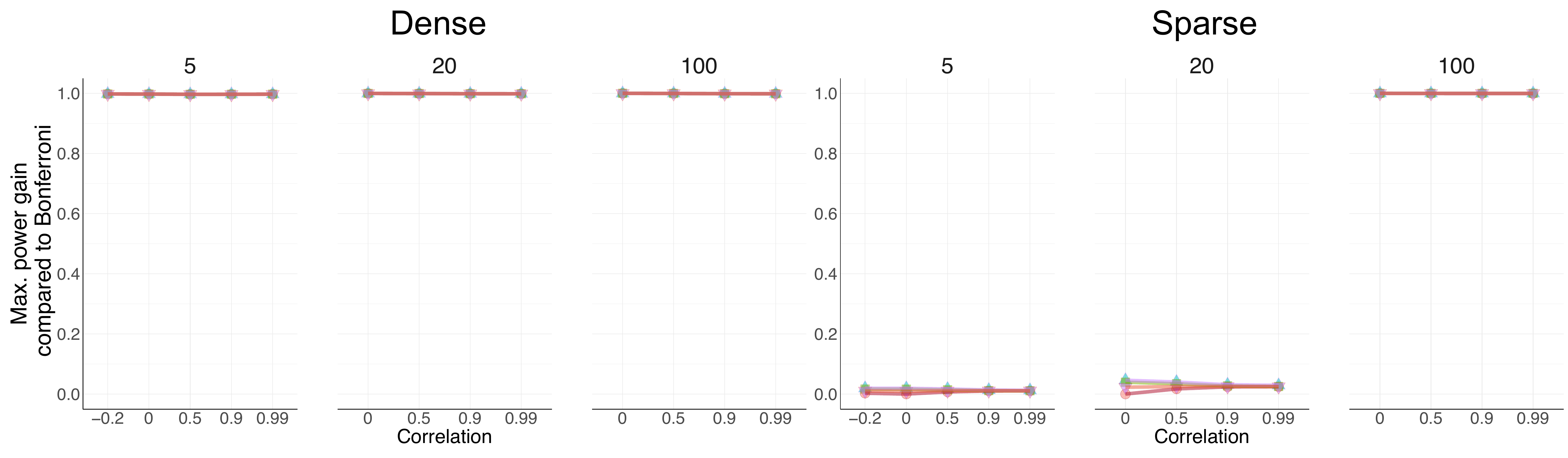}
         \caption{$\alpha=5\times10^{-4}$}
         \label{fig:power_5e-4_mvt}
     \end{subfigure}
\caption{Power comparison with the Bonferroni test when the asymptotic independence is violated of the combination test with different distributions: Cauchy (red with round dot), Fr\'echet $\gamma=1$ (green with square dot), Pareto $\gamma=1$ (purple with triangular dot), left-truncated $t_1$ with truncation threshold $p_0=0.9$ (dark orange with inverted-triangle dot). Left plots correspond to dense signals, and right plots correspond to sparse signals. The maximum power gain is defined as the maximum of the empirical power difference between the proposed test and the Bonferroni test over all possible values of $\mu$.}
\label{fig:power_gain_mvt}
\end{figure}

\section{Real Data Examples}
\label{sec:experiments}

\subsection{Circadian rhythm detection} 

Circadian rhythms, which are oscillations of behavior, physiology, and metabolism, are observed in almost all living organisms \citep{pittendrigh1960circadian}. Recent advances in omics technologies, such as microarray and next-generation sequencing, provide powerful platforms for identifying circadian genes that encode molecular clocks crucial for health and diseases \citep{rijo2019genomics}. In this case study, we focus on a gene expression dataset obtained from mouse liver samples, collected every hour across 48 different circadian time points, denoted as CT points, ranging from CT18 to CT65, under complete darkness conditions \citep{hughes2009harmonics}. At each time point, the expression levels of approximately 13,000 mouse genes were profiled by microarray. 
The objective of this case study is to identify genes that exhibit significant oscillatory behavior by aggregating results across all measured time points. 

One of the most widely used methods is JTK\_CYCLE \citep{hughes2010jtk_cycle}. JTK\_CYCLE determines whether a gene exhibits significant cyclic behavior by performing a Kendall's tau test. It compares the observed gene expression measurements across 48 time points to expected patterns with specific phases and periods using a rank-based correlation test. This process involves testing 216 combinations of phase and period, resulting in 216 correlated base p-values for each gene.  By default, JTK\_CYCLE combines these p-values using the Bonferroni test, though this approach has been shown to lack power in benchmarking studies \citep{mei2021genome}. 

In place of the Bonferroni test, we use the heavy-tailed combination tests to aggregate the $216$ correlated p-values for each gene. For comparison, we also include Fisher's method. To assess the performance of different tests, we utilize a set of the 60 positive control, i.e., cyclic genes, and 61 negative control, i.e., non-cyclic genes, from \cite{wu2014evaluation} as ground truth. \Cref{fig:circadian} displays the box plots of the combined p-values for the positive and negative controls. Compared to the Bonferroni method, the combined p-values from heavy-tailed combination tests have higher detection power of the true signals, while avoiding false positives in negative controls compared to Fisher's method. 

\begin{figure}[t]
    \centering
    \includegraphics[width=0.7\textwidth]{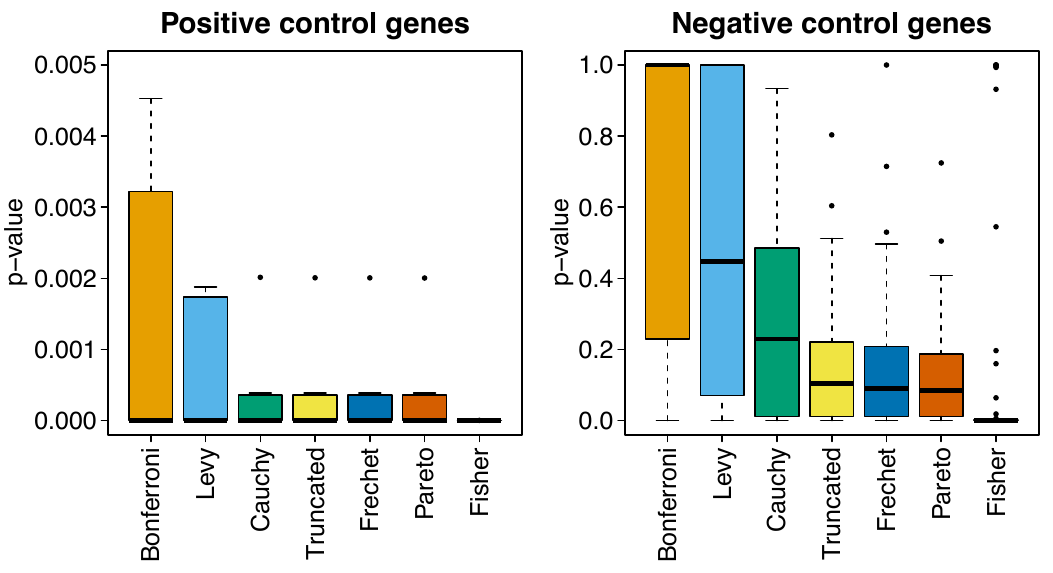} 
    \caption{P-values of positive control and negative control genes from circadian rhythm detection.
   The left plot shows box plots of the combined p-values of 60 positive controls and the right plot shows box plots of the combined p-values of 61 negative controls. ``Truncated'' refers to using the $t_1$ distribution with truncation threshold $p_0 = 0.9$. For Fr\'echet and Pareto distributions, the tail index is set to $\gamma = 1$.}
\label{fig:circadian}
\end{figure}

\subsection{SNP-based gene level association testing in GWAS}

In the second real data analysis, similar to \cite{liu2019acat}, we combine correlated p-values to identify genes that are significantly associated with diseases in genome-wide association studies, referred to as GWAS for brevity. 
A gene of interest may contain multiple single-nucleotide polymorphisms, referred to as SNPs, each tested individually against the trait, e.g., disease status, using a simple regression framework, resulting in SNP-level p-values. Then, p-values from the SNPs within the same gene region are further combined via a gene-level test. SNPs that are close to each other on the genome are highly correlated due to linkage disequilibrium, leading to highly correlated SNP-level p-values for the same gene. Several methods have been developed for gene-level association testing, such as EPIC \citep{wang2022epic} and MAGMA \citep{de2015magma}, which account for SNP-SNP correlations within the same gene. However, these methods can be computationally intensive. For example, 
deriving gene-level test statistics in these methods often requires inverting large covariance matrices.

\begin{figure}[t!]
\centering
\includegraphics[width=0.9\textwidth]{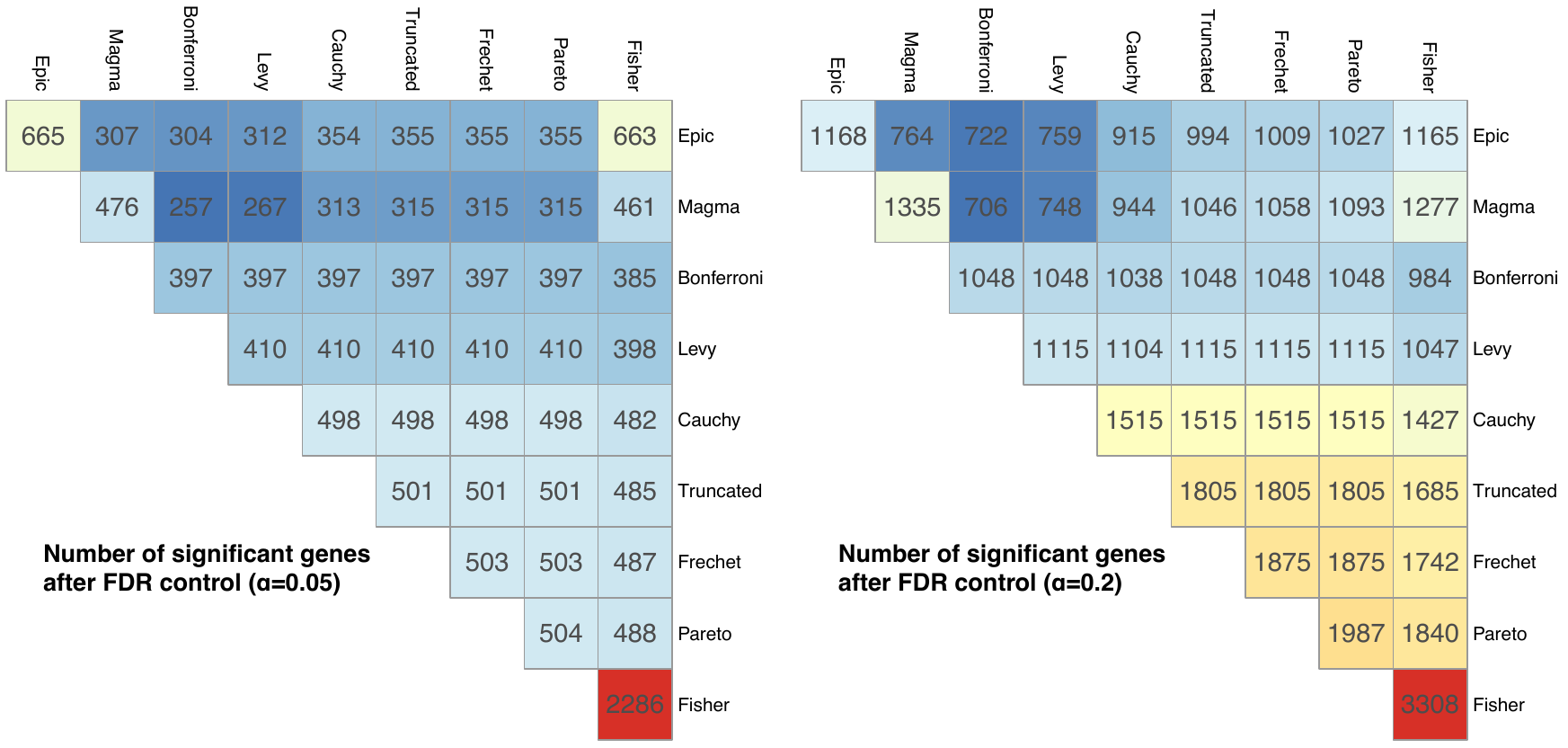}  
\caption{Number of significant genes for gene-level association testing combining SNP-level p-values when considering all genes. Diagonal values indicate the number of significant genes identified by each method; upper-triangular values indicate the number of overlapping discoveries between each pair of methods. Background colors correspond to the logarithms of the numbers. ``Truncated'' refers to the truncated $t_1$ distribution with truncation threshold $p_0 = 0.9$. For Fr\'echet and Pareto distributions, the tail index is set to $\gamma = 1$.}
\label{fig:gwas_sig_gene_all}
\end{figure}

In this analysis, we apply heavy-tailed combination tests to test for each gene's association with schizophrenia, referred to as SCZ \citep{ripke2013genome}. To adjust for multiple testing errors, we apply the Benjamini-Hochberg procedure \citep{benjamini1995controlling} on the gene-level combined p-values to control the false discovery rate, referred to as FDR for simplicity. \Cref{fig:gwas_sig_gene_all} shows the number of overlapping genes rejected by each method compared when FDR is controlled at $0.05$ and $0.2$. As illustrated, the number of genes detected by the combination tests is comparable to or even higher than those identified by Epic and Magma. Notably, the combination tests are highly computationally efficient, completing analyses almost instantly compared to domain-specific methods that require modeling the correlation structure. 

Compared to the Bonferroni test, the combination tests identify  $25\%$ more significant genes, even at a low nominal FDR level of $\alpha = 0.05$. \Cref{fig:snps_number} summarizes the number of SNPs for each gene, showing that most genes have fewer than 100 SNPs, suggesting that the significant power gain is not due to combining an excessively large number of p-values, which could lead to inflation of type-I errors. \Cref{fig:gwas_sig_gene_less_snps} displays that even when focusing solely on genes with $50$ or fewer SNPs, the combination tests still identify substantially more genes than the Bonferroni test. Compared to the simulation results, the substantial power gain in this real data analysis likely results from the violations of quasi-asymptotic independence of the SNP-level p-values. 

To evaluate whether the additional genes detected by the heavy-tailed
combination tests are biologically meaningful, we analyze the set of $939$ genes detected at the FDR level $\alpha = 0.2$ by the Cauchy, truncated $t_1$, Fr\'echet, or Pareto combination tests but not by the Bonferroni test. We conduct a gene-set enrichment analysis using DAVID \citep{sherman2022david}. Results are shown in \cref{tab:heavy_tail_class_with_dep_condition}. The top two significantly enriched gene ontology terms are ``regulation of ion transmembrane transport'' and ``chemical synaptic transmission'', both of which have been reported and confirmed by independent studies \citep{favalli2012role, liu2022impact}. 
These findings underscore the enhanced statistical power of transformation tests compared to the Bonferroni test in practical genetic applications.

\section{Discussion}

In this section, we examine the extensions and limitations of our results and discuss related literature. The asymptotic validity of the heavy-tailed combination tests can be generalized to cases where p-values are only valid, i.e., they satisfy $\Pr(p \leqslant \alpha)\leqslant \alpha$, as long as the p-values are pairwise quasi-asymptotically independent. Though the transformed test statistics $X_i$ derived from these valid p-values may lack a regularly varying tailed distribution, the combination tests should maintain control over type-I errors. Intuitively, this is because we can always construct uniformly distributed variables that are stochastically smaller than these valid p-values. 

In the context of multiple testing, the combination tests can be applied within a closed testing procedure to identify individual non-null hypotheses. In Supplementary \Cref{sec:closed_testing}, we provide a shortcut algorithm for applying closed testing with combination tests. \cite{goeman2019simultaneous} demonstrated that, as $n\to+\infty$, the closed testing procedure using harmonic mean p-values is significantly more powerful than the one based on Bonferroni corrections. However, for finite $n$ and when the family-wise error rate approaches zero, the equivalence between combination tests and the Bonferroni test may extend to their respective closed testing procedures.

To balance validity and power, we recommend using a truncated $t_1$ distribution with truncation threshold $p_0 = 0.9$, based on the empirical results. This definition differs slightly from the truncated Cauchy distribution proposed by \citet{fang2021heavy}, which assigns a point mass at the truncation threshold rather than rescaling the distribution. Notably, the half-Cauchy distribution in \citet{long2023cauchy} is a special case of our definition with $p_0=0.5$. While our focus is on establishing the asymptotic validity of combination tests using the truncated $t_1$
distribution under an unknown dependence structure, both \citet{fang2021heavy} and \citet{long2023cauchy} have also provided adjustments that ensure exact validity when p-values are independent.

While our results establish the asymptotic validity of the heavy-tailed combination tests under quasi-asymptotically independent test statistics, the combination tests can exhibit noticeable inflation in type-I error rates under arbitrary dependence and finite $\alpha$. Exact control over type-I errors may be achieved with additional adjustments. For the harmonic mean p-values, \citet{vovk2020combining} demonstrated that it is valid under arbitrary dependence when scaled by a factor $a_n=(y_n+n)^2/(ny_n+n)$ where $y_n$ is the unique solution to the equation $y_n^2=n\left\{(y_n+1)\log(y_n+1)-y_n\right\}$. This factor asymptotically approaches $\log n$ when $n$ increases and the test can be further improved through randomization techniques \citep{gasparin2024combining}. Other studies, such as  \citet{wilson2019harmonic} and subsequent works \citep{doi:10.1073/pnas.1900671116,doi:10.1073/pnas.1902157116,doi:10.1073/pnas.1909339116} have provided empirically calibrated thresholds for harmonic mean p-value. Additionally, \citet{chen2024sub} establish an adjustment of the harmonic mean p-value to guarantee its validity when the individual p-values follow a Clayton copula.

Numerous alternative methods for combining dependent p-values exist, each with distinct trade-offs.  Some approaches, such as those by \citet{goeman2004global} and \citet{edelmann2020global}, model specific dependence structures, which can be powerful but require strong model assumptions and can be computationally intensive. Other methods, like those by \citet{hommel1983tests} and \citet{vovk2020combining}, guarantee type-I error control under arbitrary dependence. However, as discussed in earlier studies \citep{fang2021heavy,chen2023trade}, combination methods with proven validity guarantees under arbitrary dependence may have limited power in practical applications. 

\section*{Data and Code Availability}
The R package facilitating the implementation of heavy-tailed combination tests is accessible at \url{https://github.com/gl-ybnbxb/heavytailcombtest}. The code to reproduce figures and tables is at: \url{https://github.com/gl-ybnbxb/combination-test-reproduce-code}. Time-series circadian gene expression data of mouse liver is downloaded from the Gene Expression Omnibus (GEO) database with accession number GSE11923. GWAS summary statistics of schizophrenia (SCZ) is downloaded from the Psychiatric Genomics Consortium at \url{https://pgc.unc.edu/for-researchers/download-results/}.

\section*{Acknowledgement}
   This work was supported by the National Science Foundation under grants DMS-2113646 and DMS-2238656, and by the National Institute of Health under grant R35  GM138342. We thank Nancy R. Zhang and Jian Ding for their helpful discussions. We acknowledge the University of Chicago’s Research Computing Center and Texas A\&M University's High Performance Research Computing for their support of this work.

\section*{Supplementary material}
The Supplementary Material includes more discussion about the test, theoretical proofs, and additional experimental results.

\bibliographystyle{abbrvnat}
\bibliography{ref.bib}

\begin{thebibliography}{51}
\providecommand{\natexlab}[1]{#1}
\providecommand{\url}[1]{\texttt{#1}}
\expandafter\ifx\csname urlstyle\endcsname\relax
  \providecommand{\doi}[1]{doi: #1}\else
  \providecommand{\doi}{doi: \begingroup \urlstyle{rm}\Url}\fi

\bibitem[Albrecher et~al.(2006)Albrecher, Asmussen, and Kortschak]{albrecher2006tail}
H.~Albrecher, S.~Asmussen, and D.~Kortschak.
\newblock Tail asymptotics for the sum of two heavy-tailed dependent risks.
\newblock \emph{Extremes}, 9\penalty0 (2):\penalty0 107--130, 2006.

\bibitem[Asmussen et~al.(2011)Asmussen, Blanchet, Juneja, and Rojas-Nandayapa]{asmussen2011efficient}
S.~Asmussen, J.~Blanchet, S.~Juneja, and L.~Rojas-Nandayapa.
\newblock Efficient simulation of tail probabilities of sums of correlated lognormals.
\newblock \emph{Annals of Operations Research}, 189:\penalty0 5--23, 2011.

\bibitem[Benjamini and Hochberg(1995)]{benjamini1995controlling}
Y.~Benjamini and Y.~Hochberg.
\newblock Controlling the false discovery rate: a practical and powerful approach to multiple testing.
\newblock \emph{Journal of the Royal statistical society: series B (Methodological)}, 57\penalty0 (1):\penalty0 289--300, 1995.

\bibitem[Botev and L'Ecuyer(2017)]{botev2017accurate}
Z.~Botev and P.~L'Ecuyer.
\newblock Accurate computation of the right tail of the sum of dependent log-normal variates.
\newblock In \emph{2017 Winter Simulation Conference (WSC)}, pages 1880--1890. IEEE, 2017.

\bibitem[Cai et~al.(2022)Cai, Lei, and Roeder]{cai2022model}
Z.~Cai, J.~Lei, and K.~Roeder.
\newblock Model-free prediction test with application to genomics data.
\newblock \emph{Proceedings of the National Academy of Sciences}, 119\penalty0 (34):\penalty0 e2205518119, 2022.

\bibitem[Chen and Yuen(2009)]{chen2009sums}
Y.~Chen and K.~C. Yuen.
\newblock Sums of pairwise quasi-asymptotically independent random variables with consistent variation.
\newblock \emph{Stochastic Models}, 25\penalty0 (1):\penalty0 76--89, 2009.

\bibitem[Chen et~al.(2023)Chen, Liu, Tan, and Wang]{chen2023trade}
Y.~Chen, P.~Liu, K.~S. Tan, and R.~Wang.
\newblock Trade-off between validity and efficiency of merging p-values under arbitrary dependence.
\newblock \emph{Statistica Sinica}, 33\penalty0 (2):\penalty0 851--872, 2023.

\bibitem[Chen et~al.(2024)Chen, Wang, Wang, and Zhu]{chen2024sub}
Y.~Chen, R.~Wang, Y.~Wang, and W.~Zhu.
\newblock Sub-uniformity of harmonic mean p-values.
\newblock \emph{arXiv preprint arXiv:2405.01368}, 2024.

\bibitem[Cline(1983)]{cline1983infinite}
D.~Cline.
\newblock Infinite series of random variables with regularly varying tails.
\newblock \emph{Tech. Rpt., Inst. Appl. Math. Statist., Univ. British Columbia}, 1983.

\bibitem[de~Leeuw et~al.(2015)de~Leeuw, Mooij, Heskes, and Posthuma]{de2015magma}
C.~A. de~Leeuw, J.~M. Mooij, T.~Heskes, and D.~Posthuma.
\newblock {MAGMA}: generalized gene-set analysis of {GWAS} data.
\newblock \emph{PLoS Computational Biology}, 11\penalty0 (4):\penalty0 1--19, 04 2015.
\newblock URL \url{https://doi.org/10.1371/journal.pcbi.1004219}.

\bibitem[Demarta and McNeil(2005)]{demarta2005t}
S.~Demarta and A.~J. McNeil.
\newblock The t copula and related copulas.
\newblock \emph{International statistical review}, 73\penalty0 (1):\penalty0 111--129, 2005.

\bibitem[Edelmann et~al.(2020)Edelmann, Saadati, Putter, and Goeman]{edelmann2020global}
D.~Edelmann, M.~Saadati, H.~Putter, and J.~Goeman.
\newblock A global test for competing risks survival analysis.
\newblock \emph{Statistical Methods in Medical Research}, 29\penalty0 (12):\penalty0 3666--3683, 2020.

\bibitem[Embrechts et~al.(2013)Embrechts, Kl{\"u}ppelberg, and Mikosch]{embrechts2013modelling}
P.~Embrechts, C.~Kl{\"u}ppelberg, and T.~Mikosch.
\newblock \emph{Modelling extremal events: for insurance and finance}, volume~33.
\newblock Springer Science \& Business Media, 2013.

\bibitem[Fang et~al.(2023)Fang, Chang, Park, and Tseng]{fang2021heavy}
Y.~Fang, C.~Chang, Y.~Park, and G.~C. Tseng.
\newblock Heavy-tailed distribution for combining dependent $ p $-values with asymptotic robustness.
\newblock \emph{Statistica Sinica}, 33:\penalty0 1115--1142, 2023.

\bibitem[Favalli et~al.(2012)Favalli, Li, Belmonte-de Abreu, Wong, and Daskalakis]{favalli2012role}
G.~Favalli, J.~Li, P.~Belmonte-de Abreu, A.~H. Wong, and Z.~J. Daskalakis.
\newblock The role of {BDNF} in the pathophysiology and treatment of schizophrenia.
\newblock \emph{Journal of psychiatric research}, 46\penalty0 (1):\penalty0 1--11, 2012.

\bibitem[Gasparin et~al.(2024)Gasparin, Wang, and Ramdas]{gasparin2024combining}
M.~Gasparin, R.~Wang, and A.~Ramdas.
\newblock Combining exchangeable p-values.
\newblock \emph{arXiv preprint arXiv:2404.03484}, 2024.

\bibitem[Geluk and Ng(2006)]{geluk2006tail}
J.~Geluk and K.~W. Ng.
\newblock Tail behavior of negatively associated heavy-tailed sums.
\newblock \emph{Journal of applied probability}, 43\penalty0 (2):\penalty0 587--593, 2006.

\bibitem[Geluk and Tang(2009)]{geluk2009asymptotic}
J.~Geluk and Q.~Tang.
\newblock Asymptotic tail probabilities of sums of dependent subexponential random variables.
\newblock \emph{Journal of Theoretical Probability}, 22\penalty0 (4):\penalty0 871--882, 2009.

\bibitem[Genovese et~al.(2006)Genovese, Roeder, and Wasserman]{genovese2006false}
C.~R. Genovese, K.~Roeder, and L.~Wasserman.
\newblock False discovery control with p-value weighting.
\newblock \emph{Biometrika}, 93\penalty0 (3):\penalty0 509--524, 2006.

\bibitem[Goeman et~al.(2004)Goeman, Van De~Geer, De~Kort, and Van~Houwelingen]{goeman2004global}
J.~J. Goeman, S.~A. Van De~Geer, F.~De~Kort, and H.~C. Van~Houwelingen.
\newblock A global test for groups of genes: testing association with a clinical outcome.
\newblock \emph{Bioinformatics}, 20\penalty0 (1):\penalty0 93--99, 2004.

\bibitem[Goeman et~al.(2019{\natexlab{a}})Goeman, Meijer, Krebs, and Solari]{goeman2019simultaneous}
J.~J. Goeman, R.~J. Meijer, T.~J. Krebs, and A.~Solari.
\newblock Simultaneous control of all false discovery proportions in large-scale multiple hypothesis testing.
\newblock \emph{Biometrika}, 106\penalty0 (4):\penalty0 841--856, 2019{\natexlab{a}}.

\bibitem[Goeman et~al.(2019{\natexlab{b}})Goeman, Rosenblatt, and Nichols]{doi:10.1073/pnas.1909339116}
J.~J. Goeman, J.~D. Rosenblatt, and T.~E. Nichols.
\newblock The harmonic mean p-value: Strong versus weak control, and the assumption of independence.
\newblock \emph{Proceedings of the National Academy of Sciences}, 116\penalty0 (47):\penalty0 23382--23383, 2019{\natexlab{b}}.
\newblock \doi{10.1073/pnas.1909339116}.
\newblock URL \url{https://www.pnas.org/doi/abs/10.1073/pnas.1909339116}.

\bibitem[Held(2019)]{doi:10.1073/pnas.1900671116}
L.~Held.
\newblock On the {B}ayesian interpretation of the harmonic mean p-value.
\newblock \emph{Proceedings of the National Academy of Sciences}, 116\penalty0 (13):\penalty0 5855--5856, 2019.
\newblock \doi{10.1073/pnas.1900671116}.
\newblock URL \url{https://www.pnas.org/doi/abs/10.1073/pnas.1900671116}.

\bibitem[Hommel(1983)]{hommel1983tests}
G.~Hommel.
\newblock Tests of the overall hypothesis for arbitrary dependence structures.
\newblock \emph{Biometrical Journal}, 25\penalty0 (5):\penalty0 423--430, 1983.

\bibitem[Hughes et~al.(2009)Hughes, DiTacchio, Hayes, Vollmers, Pulivarthy, Baggs, Panda, and Hogenesch]{hughes2009harmonics}
M.~E. Hughes, L.~DiTacchio, K.~R. Hayes, C.~Vollmers, S.~Pulivarthy, J.~E. Baggs, S.~Panda, and J.~B. Hogenesch.
\newblock Harmonics of circadian gene transcription in mammals.
\newblock \emph{PLoS genetics}, 5\penalty0 (4):\penalty0 e1000442, 2009.

\bibitem[Hughes et~al.(2010)Hughes, Hogenesch, and Kornacker]{hughes2010jtk_cycle}
M.~E. Hughes, J.~B. Hogenesch, and K.~Kornacker.
\newblock {JTK}\_{CYCLE}: an efficient nonparametric algorithm for detecting rhythmic components in genome-scale data sets.
\newblock \emph{Journal of biological rhythms}, 25\penalty0 (5):\penalty0 372--380, 2010.

\bibitem[Joe(1997)]{joe1997multivariate}
H.~Joe.
\newblock \emph{Multivariate models and multivariate dependence concepts}.
\newblock CRC press, 1997.

\bibitem[Ko and Tang(2008)]{ko2008sums}
B.~Ko and Q.~Tang.
\newblock Sums of dependent nonnegative random variables with subexponential tails.
\newblock \emph{Journal of Applied Probability}, 45\penalty0 (1):\penalty0 85--94, 2008.

\bibitem[Kortschak and Albrecher(2009)]{kortschak2009asymptotic}
D.~Kortschak and H.~Albrecher.
\newblock Asymptotic results for the sum of dependent non-identically distributed random variables.
\newblock \emph{Methodology and Computing in Applied Probability}, 11\penalty0 (3):\penalty0 279--306, 2009.

\bibitem[Liu et~al.(2022)Liu, Zinski, Mishra, Noh, Park, Qin, Olorife, Park, Abani, Park, et~al.]{liu2022impact}
D.~Liu, A.~Zinski, A.~Mishra, H.~Noh, G.-H. Park, Y.~Qin, O.~Olorife, J.~M. Park, C.~P. Abani, J.~S. Park, et~al.
\newblock Impact of schizophrenia {GWAS} loci converge onto distinct pathways in cortical interneurons vs glutamatergic neurons during development.
\newblock \emph{Molecular Psychiatry}, 27\penalty0 (10):\penalty0 4218--4233, 2022.

\bibitem[Liu and Xie(2020)]{liu2020cauchy}
Y.~Liu and J.~Xie.
\newblock Cauchy combination test: a powerful test with analytic p-value calculation under arbitrary dependency structures.
\newblock \emph{Journal of the American Statistical Association}, 115\penalty0 (529):\penalty0 393--402, 2020.

\bibitem[Liu et~al.(2019)Liu, Chen, Li, Morrison, Boerwinkle, and Lin]{liu2019acat}
Y.~Liu, S.~Chen, Z.~Li, A.~C. Morrison, E.~Boerwinkle, and X.~Lin.
\newblock {ACAT}: a fast and powerful p value combination method for rare-variant analysis in sequencing studies.
\newblock \emph{The American Journal of Human Genetics}, 104\penalty0 (3):\penalty0 410--421, 2019.

\bibitem[Long et~al.(2023)Long, Li, Zhang, and Li]{long2023cauchy}
M.~Long, Z.~Li, W.~Zhang, and Q.~Li.
\newblock The {C}auchy combination test under arbitrary dependence structures.
\newblock \emph{The American Statistician}, 77\penalty0 (2):\penalty0 134--142, 2023.

\bibitem[Marcus et~al.(1976)Marcus, Eric, and Gabriel]{marcus1976closed}
R.~Marcus, P.~Eric, and K.~R. Gabriel.
\newblock On closed testing procedures with special reference to ordered analysis of variance.
\newblock \emph{Biometrika}, 63\penalty0 (3):\penalty0 655--660, 1976.

\bibitem[Mei et~al.(2021)Mei, Jiang, Chen, Chen, Sancar, and Jiang]{mei2021genome}
W.~Mei, Z.~Jiang, Y.~Chen, L.~Chen, A.~Sancar, and Y.~Jiang.
\newblock Genome-wide circadian rhythm detection methods: systematic evaluations and practical guidelines.
\newblock \emph{Briefings in Bioinformatics}, 22\penalty0 (3):\penalty0 bbaa135, 2021.

\bibitem[Pittendrigh(1960)]{pittendrigh1960circadian}
C.~S. Pittendrigh.
\newblock Circadian rhythms and the circadian organization of living systems.
\newblock In \emph{Cold Spring Harbor symposia on quantitative biology}, volume~25, pages 159--184. Cold Spring Harbor Laboratory Press, 1960.

\bibitem[Qin et~al.(2020)Qin, Fan, Zheng, Wan, Mei, Wu, Sun, Brown, Zhang, Meyer, and Liu]{qin2020lisa}
Q.~Qin, J.~Fan, R.~Zheng, C.~Wan, S.~Mei, Q.~Wu, H.~Sun, M.~Brown, J.~Zhang, C.~A. Meyer, and X.~S. Liu.
\newblock Lisa: inferring transcriptional regulators through integrative modeling of public chromatin accessibility and {ChIP}-seq data.
\newblock \emph{Genome biology}, 21\penalty0 (1):\penalty0 1--14, 2020.

\bibitem[Rijo-Ferreira and Takahashi(2019)]{rijo2019genomics}
F.~Rijo-Ferreira and J.~S. Takahashi.
\newblock Genomics of circadian rhythms in health and disease.
\newblock \emph{Genome medicine}, 11:\penalty0 1--16, 2019.

\bibitem[Ripke et~al.(2013)Ripke, O'dushlaine, Chambert, Moran, K{\"a}hler, Akterin, Bergen, Collins, Crowley, Fromer, et~al.]{ripke2013genome}
S.~Ripke, C.~O'dushlaine, K.~Chambert, J.~L. Moran, A.~K. K{\"a}hler, S.~Akterin, S.~E. Bergen, A.~L. Collins, J.~J. Crowley, M.~Fromer, et~al.
\newblock Genome-wide association analysis identifies 13 new risk loci for schizophrenia.
\newblock \emph{Nature genetics}, 45\penalty0 (10):\penalty0 1150--1159, 2013.

\bibitem[Romme et~al.(2017)Romme, de~Reus, Ophoff, Kahn, and van~den Heuvel]{romme2017connectome}
I.~A. Romme, M.~A. de~Reus, R.~A. Ophoff, R.~S. Kahn, and M.~P. van~den Heuvel.
\newblock Connectome disconnectivity and cortical gene expression in patients with schizophrenia.
\newblock \emph{Biological psychiatry}, 81\penalty0 (6):\penalty0 495--502, 2017.

\bibitem[Rosenbaum(2012)]{Rosenbaumtestingtwice}
P.~R. Rosenbaum.
\newblock Testing one hypothesis twice in observational studies.
\newblock \emph{Biometrika}, 99:\penalty0 763--774, 2012.

\bibitem[Sherman et~al.(2022)Sherman, Hao, Qiu, Jiao, Baseler, Lane, Imamichi, and Chang]{sherman2022david}
B.~T. Sherman, M.~Hao, J.~Qiu, X.~Jiao, M.~W. Baseler, H.~C. Lane, T.~Imamichi, and W.~Chang.
\newblock {DAVID}: a web server for functional enrichment analysis and functional annotation of gene lists (2021 update).
\newblock \emph{Nucleic acids research}, 50\penalty0 (W1):\penalty0 W216--W221, 2022.

\bibitem[Sun et~al.(2020)Sun, Zhu, and Zhou]{sun2020statistical}
S.~Sun, J.~Zhu, and X.~Zhou.
\newblock Statistical analysis of spatial expression patterns for spatially resolved transcriptomic studies.
\newblock \emph{Nature methods}, 17\penalty0 (2):\penalty0 193--200, 2020.

\bibitem[Tang(2008)]{tang2008insensitivity}
Q.~Tang.
\newblock Insensitivity to negative dependence of asymptotic tail probabilities of sums and maxima of sums.
\newblock \emph{Stochastic Analysis and Applications}, 26\penalty0 (3):\penalty0 435--450, 2008.

\bibitem[Teugels et~al.(1987)Teugels, Bingham, and Goldie]{teugels1987regular}
J.~Teugels, N.~Bingham, and C.~Goldie.
\newblock \emph{Regular variations}.
\newblock Cambridge University Press, 1987.

\bibitem[Vovk and Wang(2020)]{vovk2020combining}
V.~Vovk and R.~Wang.
\newblock Combining p-values via averaging.
\newblock \emph{Biometrika}, 107\penalty0 (4):\penalty0 791--808, 2020.

\bibitem[Wang et~al.(2022)Wang, Lin, and Jiang]{wang2022epic}
R.~Wang, D.-Y. Lin, and Y.~Jiang.
\newblock {EPIC}: Inferring relevant cell types for complex traits by integrating genome-wide association studies and single-cell {RNA} sequencing.
\newblock \emph{PLoS genetics}, 18\penalty0 (6):\penalty0 e1010251, 2022.

\bibitem[Wilson(2019{\natexlab{a}})]{doi:10.1073/pnas.1902157116}
D.~J. Wilson.
\newblock Reply to {H}eld: {W}hen is a harmonic mean p-value a {B}ayes factor?
\newblock \emph{Proceedings of the National Academy of Sciences}, 116\penalty0 (13):\penalty0 5857--5858, 2019{\natexlab{a}}.
\newblock \doi{10.1073/pnas.1902157116}.
\newblock URL \url{https://www.pnas.org/doi/abs/10.1073/pnas.1902157116}.

\bibitem[Wilson(2019{\natexlab{b}})]{wilson2019harmonic}
D.~J. Wilson.
\newblock The harmonic mean p-value for combining dependent tests.
\newblock \emph{Proceedings of the National Academy of Sciences}, 116\penalty0 (4):\penalty0 1195--1200, 2019{\natexlab{b}}.

\bibitem[Wu et~al.(2014)Wu, Zhu, Yu, Zhou, Huang, and Zhang]{wu2014evaluation}
G.~Wu, J.~Zhu, J.~Yu, L.~Zhou, J.~Z. Huang, and Z.~Zhang.
\newblock Evaluation of five methods for genome-wide circadian gene identification.
\newblock \emph{Journal of biological rhythms}, 29\penalty0 (4):\penalty0 231--242, 2014.

\bibitem[Wu et~al.(2016)Wu, Anafi, Hughes, Kornacker, and Hogenesch]{wu2016metacycle}
G.~Wu, R.~C. Anafi, M.~E. Hughes, K.~Kornacker, and J.~B. Hogenesch.
\newblock Meta{C}ycle: an integrated {R} package to evaluate periodicity in large scale data.
\newblock \emph{Bioinformatics}, 32\penalty0 (21):\penalty0 3351--3353, 2016.

\end{thebibliography}

\newpage

\renewcommand{\thesection}{S\arabic{section}}   
\renewcommand{\thetable}{S\arabic{table}}   
\renewcommand{\thefigure}{S\arabic{figure}}

\setcounter{section}{0}
\setcounter{figure}{0}
\setcounter{table}{0}

\renewcommand{\thetheorem}{S\arabic{theorem}}
\renewcommand{\theassumption}{S\arabic{assumption}}
\renewcommand{\thecorollary}{S\arabic{corollary}}
\renewcommand{\theremark}{S\arabic{remark}}
\renewcommand{\thelemma}{S\arabic{lemma}}
\renewcommand{\theequation}{S\arabic{equation}}
\renewcommand{\theproposition}{S\arabic{proposition}}
\renewcommand{\thedefinition}{S\arabic{definition}}

\setcounter{equation}{0}
\setcounter{theorem}{0}
\setcounter{assumption}{0}
\setcounter{corollary}{0}
\setcounter{remark}{0}
\setcounter{lemma}{0}
\setcounter{proposition}{0}
\setcounter{definition}{0}

\bigskip
\begin{center}
{\large\bf SUPPLEMENTARY MATERIAL}
\end{center}

\section{Type-I error of the combination test with negatively correlated p-values}
\label{sec:negative_p_values}
We investigate the type-I error control of the combination tests when base p-values are negatively correlated. As shown in  \Cref{prop:two_sided_p_vals_pos_dep}, two-sided p-values with pairwise Gaussian test statistics are always pairwise non-negatively correlated, thus we generate negatively correlated one-sided p-values using the same experimental setting as in \Cref{subsec:validity_simulations} but with $\rho < 0$. To make $\Sigma_\rho$ positive definite, we require $\rho > -1/(n-1)$. We focus on $n = 2$ so that $\rho$ can take any negative values greater than $-1$. We consider three values of $\rho$: $-0.5,-0.9,-0.99$, and compare the type-I error of different combination tests.

\Cref{tab:comp_neg_dep} presents the empirical type-I errors of various methods, where each scenario is replicated $5\times10^4$ times in our experiments. Among all combination tests, only the Cauchy combination test is conservative, particularly when the p-values have strong negative correlations. This conservativeness arises from the fact that the support of the Cauchy distribution is $\mathbb{R}$, and the transformed test statistics $X_i$ can cancel each other when p-values are negatively correlated. In contrast, if we truncated the Cauchy distribution to be left bounded, the test is no longer conservative even with a modest truncation threshold $p_0 = 0.9$. 

As a confirmation of the asymptotic validity result, we further let $\alpha$ drop to $5 \times 10^{-8}$ and as shown in \Cref{tab:error_ratio}, the ratio between the empirical type-I error and $\alpha$ for the Cauchy combination test does slowly increases to $1$, which is consistent with its asymptotic validity. However, the Cauchy combination test is conservative for any moderately small $\alpha$.

\begin{table}[ht]
\centering
\caption{Empirical type-I errors of different heavy-tailed combination tests when $n = 2$ and p-values are negatively correlated. Values inside the parentheses are standard errors. The significance level is $0.05$. The Fr\'echet and Pareto distributions are with tail index $\gamma = 1$. The left-truncated t distribution with $\gamma =1$ has the truncation threshold $p_0 = 0.9$}
\scalebox{0.75}{
\begin{tabular}{@{}l|rrrrrrr@{}}
\hline
$\rho$    & Cauchy    & Pareto    & Truncated $t_1$ & Fr\'echet   & Levy      & Bonferroni & Fisher \\\hline
& & & & & & & \\
-0.5 & 0.039 & 0.054 & 0.049 & 0.053 & 0.052 & 0.052 & 0.027 \\
& ($8.69\times10^{-4}$) & ($1.01\times10^{-3}$) & ($9.62\times10^{-4}$)  & ($1.00\times10^{-3}$) & ($9.92\times10^{-4}$) & ($9.92\times10^{-4}$)  & ($7.94\times10^{-4}$) \\
& & & & & & & \\
-0.9 & 0.021 & 0.053 & 0.045 & 0.052 & 0.051 & 0.051 & 0.020 \\
& ($6.42\times10^{-4}$) & ($9.99\times10^{-4}$) & ($9.28\times10^{-4}$)  & ($9.90\times10^{-4}$) & ($9.88\times10^{-4}$) & ($9.88\times10^{-4}$)  & ($6.21\times10^{-4}$) \\
& & & & & & & \\
-0.99 & 0.008 & 0.054 & 0.045 & 0.053 & 0.052 & 0.052 & 0.019 \\
& ($3.95\times10^{-4}$) & ($1.01\times10^{-3}$) & ($9.28\times10^{-4}$)  & ($9.98\times10^{-4}$) & ($9.96\times10^{-4}$) & ($9.96\times10^{-4}$)  & ($6.03\times10^{-4}$) \\ \hline
\end{tabular}
}
\label{tab:comp_neg_dep}
\end{table}

\begin{table}[h]
\centering
\caption{The empirical type-I error and the 95\% confidence interval of the ratio empirical error/error bound derived from the 95\% Wilson binomial confidence interval when the number of base hypotheses is 2. Base p-values are one-sided p-values converted from $Z$ statistics distributed from bivariate normal with correlation $-0.9$}
\scalebox{0.85}{\begin{tabular}{ccccc}
\hline
\multicolumn{5}{c}{95\% confidence interval of $\frac{\text{empirical type-I error}}{\alpha}$}\\
\multicolumn{1}{c}{$\alpha$} & Cauchy & Pareto & Fr\'echet & Bonferroni\\ \hline
\multicolumn{1}{c|}{$5\times10^{-2}$} & $0.413\pm.000$ & $1.054\pm.000$ & $1.030\pm.000$ & $1.023\pm.000$\\
\multicolumn{1}{c|}{$5\times10^{-3}$} & $0.509\pm.000$ & $1.003\pm.000$ & $1.001\pm.000$ & $1.000\pm.000$\\
\multicolumn{1}{c|}{$5\times10^{-4}$} & $0.592\pm.000$ & $1.001\pm.001$ & $1.001\pm.001$ & $1.001\pm.001$\\
\multicolumn{1}{c|}{$5\times10^{-5}$} & $0.658\pm.002$ & $1.001\pm.002$ & $1.001\pm.002$ & $1.001\pm.002$\\
\multicolumn{1}{c|}{$5\times10^{-6}$} & $0.717\pm.007$ & $1.006\pm.008$ & $1.006\pm.008$ & $1.006\pm.008$\\
\multicolumn{1}{c|}{$5\times10^{-7}$} & $0.758\pm.021$ & $1.002\pm.024$ & $1.002\pm.024$ & $1.002\pm.024$\\
\multicolumn{1}{c|}{$5\times10^{-8}$} & $0.788\pm.068$ & $0.984\pm.076$ & $0.984\pm.076$ & $0.984\pm.076$\\ 
\hline
\end{tabular}}
\label{tab:error_ratio}
\end{table}

\newpage
\section{Closed Testing of Combination Test}
\label{sec:closed_testing}
\subsection{Describing the closed testing procedures}
The closed testing procedure, introduced by \citet{marcus1976closed}, is a multiple testing method designed to control the family-wise error rate (FWER). The definition of the closed testing procedure for a global test $\psi$ is given as follows
\begin{definition}[Closed Testing Procedure of $\psi$]
\label{def:closed_testing}
Suppose $H_1,H_2,\cdots,H_n$ are null hypotheses. The closed testing procedure rejects $H_i$ if all set $I$s containing $i$ can be rejected by $\psi$ on $I$. That is, the decision function of $H_i$ is
\begin{equation*}
\phi_i=1_{\left\{\min_{i\in I}\psi_I=1\right\}}
\end{equation*}
\end{definition}

Now let us formalize the closed testing procedure of the heavy-tailed combination test step by step. The standard heavy-tailed combination test based on the heavy-tailed distribution $F$ for a set $I$ has the test statistics
$$S_{I}=\sum_{i\in I}H(P_i),$$
where $H(P)=Q_F(1-P)$ and $Q_F$ is $F$'s quantile function. The corresponding p-value is 
$$P_I=|I|\bar{F}(S_I)$$
with $\bar{F}(x)=1-F(x)$.
According to the closure principle, when the threshold for family-wise error rate is $\alpha$, the decision function for the hypothesis $H_i$ is
\begin{equation*}
\phi_i=\min_{i\in I}1_{\{P_I\leqslant \alpha\}}=1_{\max_{i\in I}P_I\leqslant \alpha}=1_{\max_{1\leqslant k\leqslant n}\max_{i\in I,|I|=k}P_I\leqslant \alpha}.   
\end{equation*}
Therefore, the p-value for each hypothesis $H_i$ is 
\begin{equation}
\label{eq:closed_testing_indivdual_p_value}
P_i^*:=\max_{i\in I}P_I=\max_k P_{i,k}^*,
\end{equation}
where $P_{i,k}^*=\max_{i\in I,|I|=k}P_I$. For a fixed $k$, $P_{i,k}^*$ can be further rewritten as
\begin{equation*}
P_{i,k}^*=\max_{i\in I,|I|=k}P_I=k\bar{F}\left(\min_{i\in I,|I|=k}S_I\right)=k\bar{F}\left(\min_{i\in I,|I|=k}\sum_{i\in I}H(P_i)\right).
\end{equation*}
Since $H(\cdot)$ is a decreasing function, to reach the minimum, we should consider those set $I$'s with largest p-values to construct $P_{i,k}^*$. Specifically, when $k=1$, the only set in consideration is $\{i\}$ and 
\begin{equation*}
    P_{i,k}^*=P_i.
\end{equation*}
When $k\geqslant2$, there are two cases. If $P_i$ are one of the largest $k$ p-values,
\begin{equation*}
    \min_{i\in I,|I|=k}\sum_{j\in I}H(P_j) = \sum_{j=n-k+1}^nH\left(P_{(j)}\right).
\end{equation*}
Otherwise, we should combine $P_i$ with the largest $k-1$ p-values, i.e.,
\begin{equation*}
    \min_{i\in I,|I|=k}\sum_{j\in I}H(P_j) = H(P_i)+\sum_{j=n-k+2}^nH\left(P_{(j)}\right).
\end{equation*}
Summarize all three scenarios, 
\begin{equation}
\label{eq:more_detailed_individual_p_value}
 P_{i,k}^*=\max_{i\in I,|I|=k}P_I= 
 \begin{cases}
k\bar{F}\left(\max\big\{H\left(P_i\right),H\left(P_{(n-k+1)}\right)\}+\sum_{j=n-k+2}^n H\left(P_{(j)}\right)\right) & k\geqslant 2\\
P_i & k=1
\end{cases}
\end{equation}
\Cref{eq:closed_testing_indivdual_p_value,eq:more_detailed_individual_p_value} indicates that there are at most $n^2$ $ P_{i,k}^*$ to compute for the closed testing procedure once p-values are ordered. Since the summation and monotonicity of $H$ creates hierarchy for subset $I$'s, there is no need to consider all $2^n$ subsets.

\subsection{A Shortcut Algorithm for fixed $\alpha$}
For a given family-wise error rate threshold $\alpha$, we can develop a shortcut algorithm to further reduce computation. Without loss of generality, we assume observed p-values $p_1\leqslant\dots\leqslant p_n$. 

If an individual null $H_{i}$ is rejected,
\begin{align}
\label{eq:reject_i_p}
& p_i^*\leqslant\alpha \Leftrightarrow \max_k p_{i,k}^*\leqslant\alpha\\
\label{eq:reject_i_X}
\Leftrightarrow&
\begin{cases}
\quad H\left(p_i\right)\geqslant H(\alpha)\ \mathrm{for\ }k=1\\
\quad H\left(p_i\right)+\sum_{j=n-k+2}^n H\left(p_j\right)\geqslant H(\alpha/k)\ \mathrm{for\ }k=2,\cdots,n-i+1\\
\quad \sum_{j=n-k+1}^n H\left(p_j\right)\geqslant H(\alpha/k)\ \mathrm{for\ }k=n-i+2,\cdots,n\\
\end{cases},
\end{align}
where $H(p)=Q_F(1-p)$. 
Accordingly, we define threshold $c_k$'s as follows
\begin{equation*}
c_k=
\begin{cases}
H(\frac{\alpha}{k})-\sum_{j=n-k+2}^n H(p_j) & k\ge 2\\
H(\alpha) & k=1
\end{cases}
.
\end{equation*}
Then \eqref{eq:reject_i_X} can be further rewritten as
\begin{equation}
\label{eq:reject_i_final}
H(p_i)\ge \max\left(c_1,\cdots,c_{n-i+1}\right),\ H(p_{i-1})\ge c_{n-i+2},\cdots,H(p_1)\ge c_n.
\end{equation}
That is, the individual null $H_{i}$ is rejected if and only if \eqref{eq:reject_i_final} holds. Furthermore, we observe that if nulls $H_{1},\cdots,H_{i-1}$ are rejected, 
\begin{equation*}
    H(p_{i-1})\ge \max\left(c_1,\cdots,c_{n-i+2}\right),\cdots,H(p_1)\ge \max\left(c_1,\cdots,c_{n}\right),
\end{equation*}
and it natually holds that 
\begin{equation*}
 H(p_{i-1})\ge c_{n-i+2},\cdots,H(p_1)\ge c_n. 
\end{equation*}
Hence, the closed testing procedure of the heavy-tailed combination test can be formalized as a step-down procedure described as follows

\begin{algorithm}[ht]
\caption{Shortcut for the Closed Testing Procedure}
\label{alg:cap}
\KwIn{$p_1 \leq p_2 \leq \cdots \leq p_n$, threshold $\alpha$}
\For{$i \gets 1$ \KwTo $n$}{
    $x_i \gets H(p_i)$\;
}
$c_1 \gets H(\alpha)$\;
\For{$k \gets 2$ \KwTo $n$}{
    $c_k \gets H\left(\frac{\alpha}{k}\right) - \sum_{j=n-k+2}^n x_j$\;
}
$J \gets \mathrm{argmin}_i \{ x_i < \max(c_1, \ldots, c_{n-i+1}) \}$\;
\For{$i \gets 1$ \KwTo $n$}{
    Decision function $\phi_i \gets 1_{\{i < J\}}$\;
}
\KwOut{$\phi_1, \phi_2, \ldots, \phi_n$}
\end{algorithm}

\section{Proofs of theoretical results}
\label{sec:a_new_proof}
\subsection{Notations}
For the sake of simplicity, we use $\Pr_0(\cdot)$ to represent the probability measure under the global null. Besides, we use $A\overset{P}{=}B$ to stand for $\Pr(A=B)=1$. We use $\Phi(\cdot)$ and $\phi(\cdot)$ to denote the cumulative distribution function and density of the standard normal distribution. Without loss of generality, we assume all weights $\omega_i>0$ in the following proofs. Since when one $\omega_i=0$, both $\omega_iX_i$ and its tail probability are 0, and hence we can ignore the term $i$ in the multiple testing procedure. For all theorems and lemmas, we assume $F$ is the cumulative function of a distribution in $\mathscr{R}$ and corresponding quantile function is $Q_F$. The proof of all lemmas is in \Cref{sec:proof_of_lemmas}.

\subsection{Proof of \Cref{cor:tail_prob_max_vs_sum}}
\label{sec:pf_sum_max_same_tail}

\begin{proof}
The tail probability of the maximum of $X_i$'s is
\begin{equation}
\label{eq:max_smaller}
\begin{split}
\Pr\left(\max_{i=1,\ldots,n}X_i>x\right)=&\ \Pr\left(\cup_{i=1}^n\{X_i>x\}\right)\\
\leqslant& \ \sum_{i=1}^n\Pr\left(X_i>x\right)\\
=& \ \sum_{i=1}^n\overline{F_{i}}(x).
\end{split}
\end{equation}
We also have
\begin{equation}
\label{eq:max_larger}
\begin{split}
\Pr\left(\max_{i=1,\ldots,n}X_i>x\right)=&\ \Pr\left(\cup_{i=1}^n\{X_i>x\}\right)\\
\geqslant&\ \sum_{i=1}^n\Pr(X_i>x)-\sum_{i\neq j}\Pr(X_i>x,X_j>x)\\
=&\ \sum_{i=1}^n\overline{F_{i}}(x)-o(\max_{i=1,\ldots,n}\overline{F_{i}}(x)) \ ,
\end{split}
\end{equation}
where the first inequality utilize the Boole's inequality, and the last equation follows from the definition of quasi-asymptotic independence between $X_i$ and $X_j$.
Due to \Cref{eq:max_smaller,eq:max_larger}, by the Squeeze Theorem, we have 
\begin{equation*}
\lim _{x \rightarrow \infty} \frac{\Pr\left(\max_{i=1,\ldots,n}X_i>x\right)}{\sum_{i=1}^n\overline{F_{i}}(x)}=1 \ .
\end{equation*}

\end{proof}

\subsection{Proof of \Cref{thm:asymp_type_I_error_control}}
\label{sec:proof_error_control}

\begin{lemma}
\label{lemma:normal_to_regularly_varying}
Suppose random variable $Z\sim N(\mu,1)$. Then, both $X_1=Q_F\left(\Phi(Z)\right)$ and $X_2=Q_F\left(2\Phi(|Z|)-1\right)$ have distributions in class $\mathscr R$. Moreover, if $\mu=0$, the distributions of both $X_1$ and $X_2$ follow $F$.
\end{lemma}

\begin{restatable}{lemma}{NormalQuasiIndep}
\label{lemma:quasiindep}
Suppose that for all $i< j$, random variable $(T_i,T_j)$ has the bivariate normal distribution with finite means, marginal variance 1, and correlation $\rho_{ij}=\mathrm{Corr}(T_i,T_j)\in[-\rho_0,\rho_0]$. Then, for any fixed $\omega_i>0$, $i=1,\ldots,n$, we have that $\omega_i X_i$ where $X_i=Q_F(2\Phi(|T_i|)-1)$ are pairwise quasi-asymptotically independent random variables with any choice of $\rho_{ij}\in[-\rho_0,\rho_0]$:
\begin{equation}
\label{eq:uniform_asymp_indep}
\begin{split}
    &\lim_{x \rightarrow +\infty} \sup_{\rho_{ij}\in[-\rho_0,\rho_0]}\frac{\Pr\left(\omega_iX_{i}^+>x,\omega_jX_{j}^+>x\right)}{\Pr(\omega_iX_i>x)+\Pr(\omega_jX_j>x)}=0,\\
    &\lim_{x \rightarrow +\infty} \sup_{\rho_{ij}\in[-\rho_0,\rho_0]}\frac{\Pr\left(\omega_iX_{i}^+>x,\omega_jX_{j}^->x\right)}{\Pr(\omega_iX_i>x)+\Pr(\omega_jX_j>x)}=0,\\
    &\lim_{x \rightarrow +\infty} \sup_{\rho_{ij}\in[-\rho_0,\rho_0]}\frac{\Pr\left(\omega_iX_{i}^->x,\omega_jX_{j}^+>x\right)}{\Pr(\omega_iX_i>x)+\Pr(\omega_jX_j>x)}=0. 
\end{split}
\end{equation}
Moreover, the same results also hold for $X_i=Q_F(\Phi(T_i))$ if further assume $\bar{F}(x)\geqslant F(-x)$ for sufficiently large $x$.
\end{restatable}

\begin{lemma}
\label{lemma:tail_prob_max_vs_sum_uniform_version}
With the same assumptions as \Cref{lemma:quasiindep}, the following holds:
\begin{equation*}
    \lim _{x \rightarrow +\infty} \sup_{\rho_{ij}\in[-\rho_0,\rho_0]}\frac{\Pr\left(S_{n, \vec{\omega}}>x\right)}{\sum_{i=1}^n\Pr(\omega_iX_i>x)}
    =\lim _{x \rightarrow +\infty} \sup_{\rho_{ij}\in[-\rho_0,\rho_0]}\frac{\Pr\left(\max_{i=1,\ldots,n}\omega_iX_i>x\right)}{\sum_{i=1}^n\Pr(\omega_iX_i>x)}=1,
\end{equation*}
where $S_{n, \vec{\omega}}=\sum_{i=1}^n\omega_iX_i$.
\end{lemma}

Now we prove the theorem.
\begin{proof}[Proof of \Cref{thm:asymp_type_I_error_control}]
First, by \Cref{lemma:normal_to_regularly_varying}, under the global null, the cumulative distribution function of $\omega_iX_i$ is
\begin{equation*}
\Pr_0(\omega_iX_i\leqslant x)=\Pr_0(X_i\leqslant x/\omega_i)=F(x/\omega_i)   
\end{equation*}
Therefore, $\omega_iX_i$ belongs to $\mathscr{R}$. Denote $\gamma$ the tail index of $F$. Then, by \Cref{lemma:tail_prob_max_vs_sum_uniform_version}, the right tail probability of the distribution of $S_{n, \vec{\omega}}=\sum_{i=1}^n\omega_iX_i$ under the global null has the following property:
\begin{equation*}
    1=\lim _{x \rightarrow +\infty}\sup_{\rho_{ij}\in[-\rho_0,\rho_0]}\frac{\Pr_0\left(S_{n, \vec{\omega}}>x\right)}{\sum_{i=1}^n\bar{F}(x/\omega_i)}
    \overset{(\square)}{=}\lim _{x \rightarrow +\infty}\sup_{\rho_{ij}\in[-\rho_0,\rho_0]}\frac{\Pr_0\left(S_{n, \vec{\omega}}>x\right)}{\sum_{i=1}^n\omega_i^\gamma\bar{F}(x)},   
\end{equation*}
where $(\square)$ uses the fact that the tail index is $\gamma$. 

In the following, we will only prove the asymptotic validity of the
weighted version of the combination test, i.e., \Cref{def:weighted_combination_test}. Since the standard and average version of the combination test, \Cref{def:standard_combination_test,def:average_combination_test}, are the special cases of the weighted one.
\begin{align*}
&\lim_{\alpha\to0^+}\sup_{\rho_{ij}\in[-\rho_0,\rho_0]}\frac{\Pr_0\left(\phi_{\text{wgt.}}^{F,\omega}=1\right)}{\alpha} 
= \lim_{\alpha\to0^+}\sup_{\rho_{ij}\in[-\rho_0,\rho_0]}\frac{\Pr_0\left(\sum_{i=1}^n\omega_iX_i>Q_F\left(1-\alpha/\sum_{i=1}^n\omega_i^\gamma\right)\right)}{\alpha}\\
=&\lim_{\alpha\to0^+}\sup_{\rho_{ij}\in[-\rho_0,\rho_0]}\frac{\Pr_0\left(\sum_{i=1}^n\omega_iX_i>Q_F\left(1-\alpha/\sum_{i=1}^n\omega_i^\gamma\right)\right)}{\sum_{i=1}^n\omega_i^\gamma\bar{F}\left(Q_F\left(1-\alpha/\sum_{i=1}^n\omega_i^\gamma\right)\right)}
=1.
\end{align*}

Accordingly, we prove when $\alpha\to0$, the type-I error of all three versions of the combination test can be controlled at nominated level $\alpha$.
\end{proof}

\subsection{Proof of \Cref{cor:weak_validity}}
\label{sec:pf_validity_with_perfect_correlation}
\begin{proof}
We check the asymptotic validity for two-sided p-values and one-sided p-values separately.

\paragraph{Part I. Two-sided p-values.} For the two-sided p-values, both $\rho_{ij}=1$ and $-1$ will lead to equal p-values and hence equal transformed statistics $X_i$ and $X_j$. We might as well assume there are $n_0$ out of $n$ base test statistics $T_i$'s that are perfectly correlated and are $T_1,\dots, T_{n_0}$. Then, we have
\begin{equation*}
  \weightedsumnoation=\weightedsum\overset{P}{=}\sum_{i\leqslant n_0}\omega_iX_1+\sum_{i>n_0}\omega_iX_i,  
\end{equation*}
and the tail probability of $\weightedsumnoation$ should be estimated as 
\begin{equation*}
    \bar{F}(x/\sum_{i\leqslant n_0}\omega_i)+\sum_{i>n_0}\bar{F}(x/\omega_i).
\end{equation*}
This tail probability can be further estimated as 
\begin{equation*}
    \left\{\left(\sum_{i\leqslant n_0}\omega_i\right)^\gamma+\sum_{i>n_0}\omega_i^\gamma\right\}\bar{F}(x).
\end{equation*}
Hence, with \Cref{thm:asymp_type_I_error_control}, the actual rejection threshold of the combination test assuring the asymptotic validity should be $Q_F\left(1-\tfrac{\alpha}{\left(\sum_{i\leqslant n_0}\omega_i\right)^\gamma+\sum_{i>n_0}\omega_i^\gamma}\right)$. To ensure the asymptotic validity, this threshold must be smaller than $Q_F\left({1-\tfrac{\alpha}{\kappaweight}}\right)$, which is the actual threshold used in the test. In other words,
\begin{equation*}
    \left(\sum_{i\leqslant n_0}\omega_i\right)^\gamma+\sum_{i>n_0}\omega_i^\gamma\leqslant\kappaweight\Rightarrow \sum_{i\leqslant n_0}\omega_i\leqslant\left(\sum_{i\leqslant n_0}\omega_i^\gamma\right)^\frac{1}{\gamma}.
\end{equation*}
Solving the inequality, we get $\gamma\leqslant1$.

Therefore, for two-sided p-values, under the assumptions of \Cref{thm:asymp_type_I_error_control} while allowing perfect correlations, the asymptotic validity of the combination test defined as \Cref{def:weighted_combination_test} is assured when the tail index $\gamma\leqslant1$.

\paragraph{Part II. One-sided p-values.} Without loss of generality, we assume that there are $n_1$ out of $n$ base test statistics $T_1,\dots,T_{n_1}$ are correlated with $T_1$ with a correlation 1, and $n_2$ out of $n$ base test statistics $T_{n_1+1},\dots,T_{n_1+n_2}$ are correlated with $T_1$ with a correlation $-1$. 

We first check the relationships between p-values with the correlations $\rho=\pm1$ respectively:
\begin{enumerate}[label=(\roman*).]
    \item When $\rho=1$, $T_1\overset{P}{=}T_i$ and hence 
$$P_1=1-\Phi(T_1)\overset{P}{=}1-\Phi(T_i)=P_i.$$
    \item When $\rho=-1$, $T_1\overset{P}{=}-T_i$ and hence 
$$P_1=1-\Phi(T_1)\overset{P}{=}1-\Phi(-T_i)=\Phi(T_i)=1-P_i.$$
\end{enumerate}

Then, the summation $\weightedsumnoation$ can be rewritten as
$$\weightedsumnoation=\weightedsum\overset{P}{=}\sum_{i\leqslant n_1}\omega_iX_1+\sum_{n_1<i\leqslant n_1+n_2}\omega_iX_{n_1+1}+\sum_{i>n_1+n_2}\omega_iX_i.$$
We now prove the asymptotic validity of the test when either condition of $F$ in \Cref{cor:weak_validity} holds: 

(i). $F$ is bounded below.
Denote $\kappa_1=\sum_{i\leqslant n_1}\omega_i$ and $\kappa_2=\sum_{n_1<i\leqslant n_1+n_2}\omega_i$. Without loss of generality, we assume all weights $\omega_i>0$ and hence both $\kappa_1$ and $\kappa_2$ are positive.
We first check the definition of quasi-asymptotic independence, \Cref{def:quasi_indep}, between $\kappa_1X_1$ and $\kappa_2X_{n_1+1}$:
\begin{equation*}
\begin{split}
&\lim_{x\to+\infty} \frac{\Pr\left(\kappa_1X^+_1>x,\kappa_2X^+_{n_1+1}>x\right)}{\Pr(X_1>x/\kappa_1)+\Pr(X_{n_1+1}>x/\kappa_2)}
=\lim_{x\to+\infty} \frac{\Pr\left(X_1>x/\kappa_1,X_{n_1+1}>x/\kappa_2\right)}{({\kappa_1}^\gamma+{\kappa_2}^\gamma)\bar{F}(x)}\\
=&\lim_{x\to+\infty} \frac{\Pr\left(P_1>F(x/\kappa_1),1-P_1>F(x/\kappa_2)\right)}{({\kappa_1}^\gamma+{\kappa_2}^\gamma)\bar{F}(x)}
=\lim_{x\to+\infty} \frac{\Pr\left(P_1>F(x/\kappa_1),P_1<1-F(x/\kappa_2)\right)}{({\kappa_1}^\gamma+{\kappa_2}^\gamma)\bar{F}(x)}\\
=&\ 0
\end{split}
\end{equation*}
The last equation is because for sufficiently large $x$, $F(x/\kappa_1)>1-F(x/\kappa_2)$ and hence $P_1$ cannot be both larger than $F(x/\kappa_1)$ and smaller than $1-F(x/\kappa_2)$, which makes the probability $\Pr\left(P_1>F(x/\kappa_1),P_1<1-F(x/\kappa_2)\right)=0$.
To finish the proof of quasi-asymptotic independence, we also check the pair $(\kappa_1X_1^+,\kappa_2X_{n_1+1}^-)$:
\begin{equation*}
\label{eq:quasi_indep_perfect_neg_corr}
\begin{split}
&\lim_{x\to+\infty} \frac{\Pr\left(\kappa_1X^+_1>x,\kappa_2X^-_{n_1+1}>x\right)}{\Pr(X_1>x/\kappa_1)+\Pr(X_{n_1+1}>x/\kappa_2)}
=\lim_{x\to+\infty} \frac{\Pr\left(X_1>x/\kappa_1,X_{n_1+1}<-x/\kappa_2\right)}{({\kappa_1}^\gamma+{\kappa_2}^\gamma)\bar{F}(x)}=0,
\end{split}
\end{equation*}
where the last equation is because when the heavy-tailed distribution $F$ is bounded below, $\Pr(X_{n_1+1}<-x/\kappa_2)=0$ for sufficient large $x$. Accordingly, the quasi-asymptotic independence holds and the convergence is uniform for unknown nuisance parameters $\rho_{ij}$. In this case, the rejection threshold should be estimated as
$Q_F\left({1-\tfrac{\alpha}{{\kappa_1}^\gamma+{\kappa_2}^\gamma+\sum_{i>n_1+n_2}\omega_i^\gamma}}\right)$ and it needs to be smaller than $Q_F\left({1-\frac{\alpha}{\sum_{i=1}^n\omega_i^\gamma}}\right)$ to ensure the validity of the combination test. This condition is equivalent to ${\kappa_1}^\gamma+{\kappa_2}^\gamma+\sum_{i>n_1+n_2}\omega_i^\gamma\leqslant \sum_{i=1}^n\omega_i^\gamma\Rightarrow{\kappa_1}^\gamma+{\kappa_2}^\gamma\leqslant\sum_{i\leqslant n_1+n_2}\omega_i^\gamma$. This is guaranteed since $\gamma\leqslant1$. 

(ii). $\bar{F}(x)=F(-x)$ for all $x\in\mathbb R$.
Since $\bar{F}(x)=F(-x)$ for all $x\in\mathbb R$, $X_1$ and $X_{n_1+1}$ further have the relationship that
$$X_1=Q_F(1-P_1)\overset{P}{=}Q_F(P_{n_1+1})=-Q_F(1-P_{n_1+1})=-X_{n_1+1}.$$
So, $\weightedsumnoation$ can be further simplified as 
$$\weightedsumnoation\overset{P}{=}(\kappa_1-\kappa_2)X_1+\sum_{i>n_1+n_2}\omega_iX_i.$$
Following a similar analysis for two-sided p-values, we get the condition for the asymptotic validity of the combination test:
\begin{equation}
\label{eq:one_sided_symmetric_condition}
|\kappa_1-\kappa_2|^\gamma+\sum_{i>n_1+n_2}\omega_i^\gamma\leqslant \sum_{i=1}^{n_1+n_2}\omega_i^\gamma.
\end{equation}
Since $\gamma\leqslant1$, 
$$|\kappa_1-\kappa_2|^\gamma\leqslant\left(\sum_{i\leqslant n_1+n_2}{\omega_i}\right)^\gamma\leqslant \sum_{i\leqslant n_1+n_2}\omega_i^\gamma.$$
\Cref{eq:one_sided_symmetric_condition} holds and hence the asymptotic validity is established.

Combining 1 and 2, we finish the proof.
\end{proof}

\subsection{Proof of \Cref{cor:weak_validity2}}
\label{sec:pf_validity_with_perfect_correlation_2}
\begin{proof}
We check the asymptotic validity for two-sided p-values and one-sided p-values separately.

For both two-sided and one-sided p-values, $\rho_{ij}=1$ will lead to equal p-values and hence equal transformed statistics $X_i$ and $X_j$. Then, we have
$$\weightedsumnoation=\weightedsum\overset{P}{=}\left(\sum_{i=1}^n\omega_i\right)X_1.$$
The tail probability of $\weightedsumnoation$ should be estimated as $\bar{F}(x/\sum_{i}^n\omega_i)$ and can be further estimated as 
$$\left(\sum_{i=1}^n\omega_i\right)^\gamma\bar{F}(x).$$
Hence, 
$$\lim_{\alpha\to0^+}\frac{\Pr_0\left(\phi_{\text{comb}}^F=1\right)}{\alpha}=\left(\sum_{i=1}^n\omega_i\right)^\gamma\lim_{\alpha\to0^+}\frac{\bar{F}\left(Q_F\left(1-\alpha/\sum_{i=1}^n\omega_i^\gamma\right)\right)}{\alpha}=\frac{\left(\sum_{i=1}^n\omega_i\right)^\gamma}{\sum_{i=1}^n\omega_i^\gamma}.$$
\end{proof}

\subsection{Proof of \Cref{thm:same_as_bonferroni}}
\label{sec:pf_same_as_bon}

\begin{lemma}
\label{lemma:neg_cor_pair_prob}
Let $X$ and $Y$ be random variables that are jointly normally distributed with a positive correlation $\rho$ such that $\rho\leqslant\rho_0< 1$, and with marginal variances equal to 1. Let weights $\omega_X, \omega_Y > 0$. Define $c_0 = \tfrac{3}{2} - \tfrac{1}{1+\rho_0}$ and $\delta_\alpha = (n-1) \tfrac{Q_F\left(1-\alpha^{c_0}\right)}{Q_F\left(1-\alpha\right)}$. Then, as $\alpha \to 0$, the following properties hold for $\delta_\alpha$:
\begin{enumerate}[label=(\roman*)]
\item $\delta_\alpha\to0$,
\item $\delta_\alpha Q_F\left({1-{\alpha}}\right)\to\infty$,
\item $\bar{F}\left((1+\delta_\alpha)Q_F\left(1-\alpha\right)\right)/\bar{F}\left(Q_F\left({1-{\alpha}}\right)\right)\to1$,
\item $\sup_{\rho\in[-\rho_0,\rho_0]}\Pr\left(Y\geqslant h\left(\frac{1}{\omega_Y}\frac{\delta_\alpha}{n-1}Q_F\left({1-{\alpha}}\right)\right)\mid X \geqslant h\left(\frac{1}{\omega_X}(1+\delta_\alpha)Q_F\left({1-{\alpha}}\right)\right)\right)\to0$,
\end{enumerate}
where $h(\cdot)=\Phi^{-1}(F(\cdot))$.
\end{lemma}

\begin{restatable}{lemma}{CompBonPosDepPVals}
\label{lemma:diff_combination_bon_prob}
Suppose test statistics $\{T_i\}_{i=1}^n$ satisfy that for any pair $1 \leqslant i < j \leqslant n$, $(T_i, T_j)$ follows a bivariate normal distribution with unit marginal variance and correlation $\rho_{ij} \in [-\rho_0, \rho_0]$.
Define the transformed test statistics $\omega_i X_i = \omega_i Q_F(1 - P_i)$ and the weighted sum $S_{n, \vec{\omega}} = \sum_{i=1}^n \omega_i X_i$ with $\omega_i > 0$. Let the p-values be two-sided, given by $P_i = 2(1 - \Phi(|T_i|))$. Then, for any $i=1,2,\cdots,n$, the following holds:
\begin{equation}
\label{eq:diff_combination_bon_prob}
\lim_{\alpha\to0^+}\sup_{\rho_{ij}\in[-\rho_0,\rho_0]}\frac{\Pr\left(\omega_i X_i > Q_F(1 - \alpha),~S_{n, \vec{\omega}} \leqslant Q_F(1 - \alpha)\right)}{\sum_{i=1}^n \Pr\left(\omega_i X_i > Q_F(1 - \alpha)\right)}=0~.
\end{equation}
Additionally, if the distribution $F$ satisfies $\bar{F}(x) \geqslant F(-x)$ for all $x \in \mathbb{R}$, then the same conclusion holds for one-sided p-values, $P_i = 1 - \Phi(T_i)$.
\end{restatable}

\begin{proof}[Proof of \Cref{thm:same_as_bonferroni}]
Without loss of generality, we assume $\sum_{i=1}^n\omega_i^\gamma=1$. Denote $S_{n, \vec{\omega}}=\sum_{i=1}^n\omega_iX_i$.
With respect to the definition of quantile function, $Q_F(t)>x\Leftrightarrow t>F(x)$ for all $x\in\mathbb R$. Then
\begin{equation*}
    \begin{aligned}
        \max_{i=1,\ldots,n}\omega_iX_i> Q_F\left({1-{\alpha}}\right) &\Leftrightarrow \bigcup_{i=1}^n \left\{w_iX_i>Q_F\left(1-\alpha\right)\right\}\\ &\Leftrightarrow \bigcup_{i=1}^n \left\{P_i<\bar F \left(\frac{1}{\omega_i}Q_F\left(1-\alpha\right)\right)\right\}\\
        &\Leftrightarrow \bigcup_{i=1}^n\left\{\frac{P_i}{\omega_i^\gamma}<\alpha\right\}\Leftrightarrow \min_{i=1,\ldots,n}\frac{P_i}{\omega_i^\gamma}<{\alpha}
    \end{aligned}
\end{equation*}

where the right-hand-side is exactly the weighted Bonferroni test $\phi_{\text{bon.}}^{{\tilde\omega}}$. 

We can rewrite the ratio of the probability of tests' difference and the probability of the tests as
\begin{align*}
&\lim_{\alpha\to0^+}\sup_{\rho\in[-\rho_0,\rho_0]}\frac{\Pr\left(\phi_{\text{wgt.}}^{F,\omega}\neq\phi_{\text{bon.}}^{{\tilde\omega}}\right)}{\min\left\{\Pr\left(\phi_{\text{wgt.}}^{F,\omega}=1\right),\Pr\left(\phi_{\text{bon.}}^{{\tilde\omega}}=1\right)\right\}}\\
=~&\lim_{\alpha\to0^+}\sup_{\rho\in[-\rho_0,\rho_0]}\left[\frac{\Pr\left(S_{n, \vec{\omega}}> Q_F\left({1-{\alpha}}\right),\max_{i=1,\ldots,n}\omega_iX_i\leqslant Q_F\left({1-{\alpha}}\right)\right)}{\min\left\{\Pr\left(S_{n, \vec{\omega}}> Q_F\left({1-{\alpha}}\right)\right),\Pr\left(\max_{i=1,\ldots,n}\omega_iX_i>Q_F\left({1-{\alpha}}\right)\right)\right\}}\right.\\
& \left.+\frac{\Pr\left(S_{n, \vec{\omega}}\leqslant Q_F\left(1-\alpha\right),\max_{i=1,\ldots,n}\omega_iX_i> Q_F\left(1-\alpha\right)\right)}{\min\left\{\Pr\left(S_{n, \vec{\omega}}> Q_F\left({1-{\alpha}}\right)\right),\Pr\left(\max_{i=1,\ldots,n}\omega_iX_i>Q_F\left({1-{\alpha}}\right)\right)\right\}}\right]\\ 
\leqslant~& \lim_{\alpha\to0^+}\sup_{\rho\in[-\rho_0,\rho_0]}\frac{\Pr\left(S_{n, \vec{\omega}}> Q_F\left(1-\alpha\right)\right)+\Pr\left(\max_{i=1,\ldots,n}\omega_iX_i> Q_F\left(1-\alpha\right)\right)}{\sum_{i=1}^n\Pr\left(\omega_iX_i> Q_F\left({1-{\alpha}}\right)\right)}\\
&-2\lim_{\alpha\to0^+}\inf_{\rho\in[-\rho_0,\rho_0]}\frac{\Pr\left(S_{n, \vec{\omega}}> Q_F\left(1-\alpha\right),\max_{i=1,\ldots,n}\omega_iX_i> Q_F\left(1-\alpha\right)\right)}{\sum_{i=1}^n\Pr\left(\omega_iX_i> Q_F\left({1-{\alpha}}\right)\right)}\\
 =~& 2-2\lim_{\alpha\to0^+}\inf_{\rho\in[-\rho_0,\rho_0]}\frac{\Pr\left(S_{n, \vec{\omega}}> Q_F\left(1-\alpha\right),\max_{i=1,\ldots,n}\omega_iX_i> Q_F\left(1-\alpha\right)\right)}{\sum_{i=1}^n\Pr\left(\omega_iX_i> Q_F\left({1-{\alpha}}\right)\right)} \ ,
\end{align*}
where the last two equations are based on \Cref{lemma:tail_prob_max_vs_sum_uniform_version}. Then, to prove asymptotic equivalence of two tests, it suffices to confirm
\begin{equation}
\label{eq:same_as_bon_equal_condition}
\lim_{\alpha\to0^+}\inf_{\rho\in[-\rho_0,\rho_0]}\frac{\Pr\left(S_{n, \vec{\omega}}> Q_F\left(1-\alpha\right),\max_{i=1,\ldots,n}\omega_iX_i> Q_F\left(1-\alpha\right)\right)}{\sum_{i=1}^n\Pr\left(\omega_iX_i> Q_F\left({1-{\alpha}}\right)\right)}=1 \ .
\end{equation}

Let $A_{i, \alpha}=\left\{\omega_iX_i > Q_F\left(1-\alpha\right), \sum_{k=1}^n \omega_kX_k > Q_F\left(1-\alpha\right)\right\}$ and $ A_\alpha=\bigcup_{i=1}^n A_{i, \alpha}$. 
Then, the probability in the numerator of \Cref{eq:same_as_bon_equal_condition} is just $\Pr(A_\alpha)$. By the Boole's and Bonferroni's inequalities, we have
\begin{equation*}
\sum_{i=1}^n \Pr\left(A_{i, \alpha}\right)-\sum_{1 \leqslant i<j \leqslant n} \Pr\left(A_{i, \alpha} \cap A_{j, \alpha}\right) \leqslant \Pr\left(A_\alpha\right) \leqslant \sum_{i=1}^n \Pr\left(A_{i, \alpha}\right)   
\end{equation*}
Since the bivariate normality condition guarantees the quasi-asymptotic independence between $\omega_iX_i$ and $\omega_jX_j$ based on \Cref{lemma:quasiindep}, for all $1\leqslant i<j\leqslant n$, we have
\begin{equation*}
\begin{split}
&\lim_{\alpha\to0^+}\sup_{\rho\in[-\rho_0,\rho_0]}\frac{\Pr\left(A_{i, \alpha} \cap A_{j, \alpha}\right)}{\sum_{i=1}^n\Pr\left(\omega_iX_i > Q_F\left(1-\alpha\right)\right)}\\
\leqslant~&\lim_{\alpha\to0^+}\sup_{\rho\in[-\rho_0,\rho_0]}\frac{\Pr\left(\omega_iX_i > Q_F\left(1-\alpha\right), \omega_jX_j > Q_F\left(1-\alpha\right)\right)}{\sum_{i=1}^n\Pr\left(\omega_iX_i > Q_F\left(1-\alpha\right)\right)}\\
\leqslant~&\lim_{\alpha\to0^+}\sup_{\rho\in[-\rho_0,\rho_0]}\frac{\Pr\left(\omega_iX_i > Q_F\left(1-\alpha\right), \omega_jX_j > Q_F\left(1-\alpha\right)\right)}{\Pr\left(\omega_iX_i > Q_F\left(1-\alpha\right)+\Pr\left(\omega_jX_j > Q_F\left(1-\alpha\right)\right)\right)}=0~.
\end{split}
\end{equation*}
Therefore,
\begin{equation*}
\lim_{\alpha\to0^+}\sup_{\rho\in[-\rho_0,\rho_0]}\frac{\sum_{1 \leqslant i<j\leqslant n} \Pr\left(A_{i, \alpha} \cap A_{j, \alpha}\right)}{\sum_{i=1}^n\Pr\left(\omega_iX_i > Q_F\left(1-\alpha\right)\right)}=0~.
\end{equation*}
Then, by the Squeeze Theorem, we know that
\begin{equation*}
\lim_{\alpha\to0^+}\inf_{\rho\in[-\rho_0,\rho_0]}\frac{\Pr\left(A_\alpha\right)}{\sum_{i=1}^n\Pr\left(\omega_iX_i > Q_F\left(1-\alpha\right)\right)}
=\lim_{\alpha\to0^+}\inf_{\rho\in[-\rho_0,\rho_0]}\frac{\sum_{i=1}^n\Pr\left(A_{i, \alpha}\right)}{\sum_{i=1}^n\Pr\left(\omega_iX_i > Q_F\left(1-\alpha\right)\right)}.
\end{equation*}

Since by \Cref{lemma:diff_combination_bon_prob} we have 
\begin{equation*}
\begin{split}
\lim_{\alpha\to0^+}\sup_{\rho\in[-\rho_0,\rho_0]}\frac{\Pr\left(\omega_iX_i > Q_F\left(1-\alpha\right), \sum_{k=1}^n \omega_kX_k \leqslant Q_F\left(1-\alpha\right)\right)}{\sum_{i=1}^n\Pr\left(\omega_iX_i> Q_F\left({1-{\alpha}}\right)\right)}=0~.
\end{split}
\end{equation*}
Plugging in
\begin{equation*}
\begin{split}
\Pr\left(A_{i,\alpha}\right)
=\Pr\left(\omega_iX_i > Q_F\left(1-\alpha\right)\right)
-\Pr\left(\omega_iX_i > Q_F\left(1-\alpha\right), \sum_{k=1}^n \omega_kX_k \leqslant Q_F\left(1-\alpha\right)\right)~,
\end{split}
\end{equation*}
we have
{\small
\begin{equation*}
\begin{split}
1\ge~&\lim_{\alpha\to0^+}\inf_{\rho\in[-\rho_0,\rho_0]}\frac{\Pr\left(S_{n, \vec{\omega}}> Q_F\left(1-\alpha\right),\max_{i=1,\ldots,n}\omega_iX_i> Q_F\left(1-\alpha\right)\right)}{\sum_{i=1}^n\Pr\left(\omega_iX_i > Q_F\left(1-\alpha\right)\right)}\\
=~&\lim_{\alpha\to0^+}\inf_{\rho\in[-\rho_0,\rho_0]}\frac{\sum_{i=1}^n\Pr\left(A_{i, \alpha}\right)}{\sum_{i=1}^n\Pr\left(\omega_iX_i > Q_F\left(1-\alpha\right)\right)}\\
=&\lim_{\alpha\to0^+}\inf_{\rho\in[-\rho_0,\rho_0]}\frac{\sum_{i=1}^n\Pr\left(\omega_iX_i > Q_F\left(1-\alpha\right)\right)-\sum_{i=1}^n\Pr\left(\omega_iX_i > Q_F\left(1-\alpha\right), \sum_{k=1}^n \omega_kX_k \leqslant Q_F\left(1-\alpha\right)\right)}{\sum_{i=1}^n\Pr\left(\omega_iX_i > Q_F\left(1-\alpha\right)\right)}\\
\geqslant&1-\lim_{\alpha\to0^+}\sup_{\rho\in[-\rho_0,\rho_0]}\frac{\sum_{i=1}^n\Pr\left(\omega_iX_i > Q_F\left(1-\alpha\right), \sum_{k=1}^n \omega_kX_k \leqslant Q_F\left(1-\alpha\right)\right)}{\sum_{i=1}^n\Pr\left(\omega_iX_i > Q_F\left(1-\alpha\right)\right)}
=1~,
\end{split}
\end{equation*}
}
and hence
\begin{equation*}
\lim_{\alpha\to0^+}\inf_{\rho\in[-\rho_0,\rho_0]}\frac{\Pr\left(S_{n, \vec{\omega}}> Q_F\left(1-\alpha\right),\max_{i=1,\ldots,n}\omega_iX_i> Q_F\left(1-\alpha\right)\right)}{\sum_{i=1}^n\Pr\left(\omega_iX_i > Q_F\left(1-\alpha\right)\right)}=1~.
\end{equation*}
This guarantees the asymptotic equivalence
\begin{equation*}
    \lim_{\alpha\to0^+}\sup_{\rho\in[-\rho_0,\rho_0]}\frac{\Pr\left(\phi_{\text{wgt.}}^{F,\omega}\neq\phi_{\text{bon.}}^{{\tilde\omega}}\right)}{\min\left\{\Pr\left(\phi_{\text{wgt.}}^{F,\omega}=1\right),\Pr\left(\phi_{\text{bon.}}^{{\tilde\omega}}=1\right)\right\}}=0~.
\end{equation*}
\end{proof}

\subsection{Proof of Lemmas}
\label{sec:proof_of_lemmas}
\begin{proof}[Proof of \Cref{lemma:normal_to_regularly_varying}]
We prove the results for $X_1=Q_F\left(\Phi(Z)\right)$ and $X_2=Q_F\left(2\Phi(|Z|)-1\right)$ separately in part I and II.

\paragraph{Part I.} The tail probability of $X_1$ is
\begin{equation*}
\Pr\left(X_1>x\right)=\Pr\left(\Phi(Z)>F(x)\right)=\Pr\left(Z>\Phi^{-1}(F(x))\right)=1-\Phi\left(\Phi^{-1}\left(F(x)\right)-\mu\right).
\end{equation*}
Then check the definition of the regularly varying tailed distribution:
\begin{equation*}
\begin{split}
&\lim_{x\to+\infty}\frac{\Pr\left(X_1>xy\right)}{\Pr\left(X_1>x\right)}
=\lim_{x\to+\infty}\frac{1-\Phi\left(\Phi^{-1}\left(F(xy)\right)-\mu\right)}{1-\Phi\left(\Phi^{-1}\left(F(x)\right)-\mu\right)}\\
=&\lim_{x\to+\infty}\frac{\Phi^{-1}\left(F(x)\right)-\mu}{\Phi^{-1}\left(F(xy)\right)-\mu}\times\frac{\phi\left(\Phi^{-1}\left(F(xy)\right)-\mu\right)}{\phi\left(\Phi^{-1}\left(F(x)\right)-\mu\right)}\\
=&\lim_{x\to+\infty}\frac{\Phi^{-1}\left(F(x)\right)}{\Phi^{-1}\left(F(xy)\right)}\times\frac{\phi\left(\Phi^{-1}(F(xy))\right)}{\phi\left(\Phi^{-1}(F(x))\right)}\times\exp\left[\mu\left\{\Phi^{-1}\left(F(xy)\right)-\Phi^{-1}\left(F(x)\right)\right\}\right]\\
=&\lim_{x\to+\infty}\frac{1-F(xy)}{1-F(x)}\times\lim_{x\to+\infty}\exp\left[\mu\left\{\Phi^{-1}\left(F(xy)\right)-\Phi^{-1}\left(F(x)\right)\right\}\right]\\
=&y^{-\gamma}\times\lim_{x\to+\infty}\frac{\exp\left[\mu\left\{-2\log\left(\bar{F}(xy)\right)-\log\log\left(1/\bar{F}(xy)\right)-\log(4\pi)+o(1)\right\}^{\frac{1}{2}}\right]}{\exp\left[\mu\left\{-2\log\left(\bar{F}(x)\right)-\log\log\left(1/\bar{F}(x)\right)-\log(4\pi)+o(1)\right\}^{\frac{1}{2}}\right]}=y^{-\gamma}\ ,
\end{split}
\end{equation*}
where the second and fourth equation are due to $\lim_{x\to+\infty}\frac{1-\Phi(x)}{\phi(x)/x}=1$, and the second to last equation is because $\lim_{x\to0}\frac{\Phi^{-1}(1-x)}{\sqrt{-2\log x-\log\log(1/x)-\log(4\pi)+o(1)}}=1$.
Hence, the distribution of $X_1$ is still in class $\mathscr{R}$.

In particular, when $\mu=0$,
\begin{equation*}
\Pr\left(X_1\leqslant x\right)=\Phi\left(\Phi^{-1}\left(F(x)\right)\right)=F(x).
\end{equation*}
That is, $X_1\sim F$.

\paragraph{Part II.} The tail probability of $X_2$ is
\begin{equation}
\label{eq:tail_two_sided}
\begin{split}
&\Pr\left(X_2>x\right)=\Pr\left(2\Phi(|Z|)-1>F(x)\right)=\Pr\left(|Z|>\Phi^{-1}\left(\frac{F(x)+1}{2}\right)\right)\\
=&1-\Phi\left(\Phi^{-1}\left(\frac{F(x)+1}{2}\right)-\mu\right)+\Phi\left(\Phi^{-1}\left(-\frac{F(x)+1}{2}\right)-\mu\right)\\
=&1-\Phi\left(\Phi^{-1}\left(\frac{F(x)+1}{2}\right)-\mu\right)+1-\Phi\left(\Phi^{-1}\left(\frac{F(x)+1}{2}\right)+\mu\right).
\end{split}
\end{equation}
For any $\mu$,
\begin{equation*}
\begin{split}
&\lim_{x\to+\infty}\frac{1-\Phi\left(\Phi^{-1}\left(\frac{F(xy)+1}{2}\right)-\mu\right)}{1-\Phi\left(\Phi^{-1}\left(\frac{F(x)+1}{2}\right)-\mu\right)}
=\lim_{x\to+\infty}\frac{\Phi^{-1}\left(\frac{F(x)+1}{2}\right)-\mu}{\Phi^{-1}\left(\frac{F(xy)+1}{2}\right)-\mu}\times\frac{\phi\left(\Phi^{-1}\left(\frac{F(xy)+1}{2}\right)-\mu\right)}{\phi\left(\Phi^{-1}\left(\frac{F(x)+1}{2}\right)-\mu\right)}\\
=&\lim_{x\to+\infty}\frac{\Phi^{-1}\left(\frac{F(x)+1}{2}\right)}{\Phi^{-1}\left(\frac{F(xy)+1}{2}\right)}\times\frac{\phi\left(\Phi^{-1}\left(\frac{F(xy)+1}{2}\right)\right)}{\phi\left(\Phi^{-1}\left(\frac{F(x)+1}{2}\right)\right)}\times\exp\left\{\mu\left(\Phi^{-1}\left(\frac{F(xy)+1}{2}\right)-\Phi^{-1}\left(\frac{F(x)+1}{2}\right)\right)\right\}\\
=&\lim_{x\to+\infty}\frac{1-F(xy)}{1-F(x)}\times\lim_{x\to+\infty}\exp\left\{\mu\left(\Phi^{-1}\left(\frac{F(xy)+1}{2}\right)-\Phi^{-1}\left(\frac{F(x)+1}{2}\right)\right)\right\}\\
=&y^{-\gamma}\times\lim_{x\to+\infty}\frac{\exp\left\{\mu\sqrt{-2\log\left(\bar{F}(xy)/2\right)-\log\log\left(2/\bar{F}(xy)\right)-\log(4\pi)+o(1)}\right\}}{\exp\left\{\mu\sqrt{-2\log\left(\bar{F}(x)/2\right)-\log\log\left(2/\bar{F}(x)\right)-\log(4\pi)+o(1)}\right\}}=y^{-\gamma}\ ,
\end{split}
\end{equation*} 
where the second and fourth equation are due to $\lim_{x\to+\infty}\frac{1-\Phi(x)}{\phi(x)/x}=1$, and the second to last equation is because $\Phi^{-1}(1-x)=\sqrt{-2\log x-\log\log(1/x)-\log(4\pi)+o(1)}$.

Without loss of generality, assume $\mu>0$, then
\begin{equation*}
\lim_{y\to\infty}\frac{1-\Phi^{-1}(y+\mu)}{1-\Phi^{-1}(y-\mu)}=\lim_{y\to\infty}\frac{y-\mu}{y+\mu}\frac{\phi(y+\mu)}{\phi(y-\mu)}=\lim_{y\to\infty}\exp\left(-2\mu y\right)=0\ .
\end{equation*}
Consequently, plug in \eqref{eq:tail_two_sided},
\begin{equation*}
\begin{split}
&\lim_{x\to+\infty}\frac{\Pr\left(X_2>xy\right)}{\Pr\left(X_2>x\right)}\\
=&\lim_{x\to+\infty}\frac{1-\Phi\left(\Phi^{-1}\left(\frac{F(xy)+1}{2}\right)-\mu\right)+1-\Phi\left(\Phi^{-1}\left(\frac{F(xy)+1}{2}\right)+\mu\right)}{1-\Phi\left(\Phi^{-1}\left(\frac{F(x)+1}{2}\right)-\mu\right)+1-\Phi\left(\Phi^{-1}\left(\frac{F(x)+1}{2}\right)+\mu\right)}\\
=&\lim_{x\to+\infty}\frac{1-\Phi\left(\Phi^{-1}\left(\frac{F(xy)+1}{2}\right)-\mu\right)+1-\Phi\left(\Phi^{-1}\left(\frac{F(xy)+1}{2}\right)+\mu\right)}{1-\Phi\left(\Phi^{-1}\left(\frac{F(x)+1}{2}\right)-\mu\right)}\\
=&\lim_{x\to+\infty}\frac{1-\Phi\left(\Phi^{-1}\left(\frac{F(xy)+1}{2}\right)-\mu\right)}{1-\Phi\left(\Phi^{-1}\left(\frac{F(x)+1}{2}\right)-\mu\right)}\\
&\quad+\lim_{x\to+\infty}\frac{1-\Phi\left(\Phi^{-1}\left(\frac{F(x)+1}{2}\right)+\mu\right)}{1-\Phi\left(\Phi^{-1}\left(\frac{F(x)+1}{2}\right)-\mu\right)}\times\lim_{x\to+\infty}\frac{1-\Phi\left(\Phi^{-1}\left(\frac{F(xy)+1}{2}\right)+\mu\right)}{1-\Phi\left(\Phi^{-1}\left(\frac{F(x)+1}{2}\right)+\mu\right)}\\
=&y^{-\gamma}+0\times y^{-\gamma}=y^{-\gamma}\ ,
\end{split}
\end{equation*}

Again, when $\mu=0$,
\begin{equation*}
\begin{split}
&\Pr\left(X_2\leqslant x\right)=\Pr\left(2\Phi(|Z|)-1\leqslant F(x)\right)=\Pr\left(|Z|\leqslant\Phi^{-1}\left(\frac{F(x)+1}{2}\right)\right)\\
=&\Phi\left(\Phi^{-1}\left(\frac{F(x)+1}{2}\right)\right)-\Phi\left(\Phi^{-1}\left(-\frac{F(x)+1}{2}\right)\right)=F(x).
\end{split}
\end{equation*}

Summarize two parts. The lemma follows.
\end{proof}

\begin{proof}[Proof of \Cref{lemma:quasiindep}]
We first prove that for any bivariate normal random vector $(X,Y)$ with an unknown correlation $\rho$ and a common marginal variance 1, it holds that
\begin{equation}
\label{eq:normal_asym_indep}
\lim _{t\rightarrow +\infty}\sup_{\rho\in[-\rho_0,\rho_0]}\Pr\left(Y>t\mid X>t\right)=0,~
\lim _{t\rightarrow +\infty}\sup_{\rho\in[-\rho_0,\rho_0]}\Pr\left(X>t\mid Y>t\right)=0~. 
\end{equation}
Without loss of generality, we only need to prove the first equation. Suppose the mean vector is $(\mu_1,\mu_2)$ where $\max_i|\mu_i|<\infty$, and the correction is $\rho$ where $|\rho|\leqslant\rho_0<1$. Denote $\phi_Y(\cdot)$ and $\phi_{XY}(\cdot,\cdot)$ as the densities of $Y$ and $(X, Y)$. \Cref{eq:normal_asym_indep} can be rewritten as
\begin{equation*}
\begin{split}
&\lim _{t\rightarrow +\infty}\sup_{\rho\in[-\rho_0,\rho_0]}\frac{\Pr\left(X>t,Y>t\right)}{\Pr\left(X>t\right)}
=\lim _{t\rightarrow +\infty}\sup_{\rho\in[-\rho_0,\rho_0]}\frac{\Pr\left(X>t,\tfrac{Y-\mu_2-\rho(X-\mu_1)}{\sqrt{1-\rho^2}}>\tfrac{t-\mu_2-\rho(X-\mu_1)}{\sqrt{1-\rho^2}}\right)}{\Pr\left(X>t\right)}\\
=&\lim _{t\rightarrow +\infty}\sup_{\rho\in[-\rho_0,\rho_0]} \frac{E\left[1_{\{X>t\}}\bar{\Phi}\left(\tfrac{t-\mu_2-\rho(X-\mu_1)}{\sqrt{1-\rho^2}}\right)\right]}{\Pr\left(X>t\right)}
=\lim _{t\rightarrow +\infty}\sup_{\rho\in[-\rho_0,\rho_0]} \frac{E\left[1_{\{X>t\}}\Phi\left(\tfrac{\rho(X-\mu_1)-(t-\mu_2)}{\sqrt{1-\rho^2}}\right)\right]}{\Pr\left(X>t\right)}\\
\leqslant&\lim _{t\rightarrow +\infty}\frac{E\left[1_{\{X>t\}}\Phi\left(\tfrac{\rho_0(X-\mu_1)-(t-\mu_2)}{\sqrt{1-\rho_0^2}}\right)\right]}{\Pr\left(X>t\right)}
\le\lim _{t\rightarrow +\infty}\Phi\left(\frac{\rho_0(t-\mu_1)-(t-\mu_2)}{\sqrt{1-\rho_0^2}}\right)=0.
\end{split}
\end{equation*}
The first and second equation is based on the property of a bivariate normal distribution. That is, there exists a standard normal random variable $Z$, independent of $X$, such that $Y-\mu_2=\rho(X-\mu_1)+\sqrt{1-\rho^2}Z$. The third equation utilizes the fact $1-\Phi(x)=\Phi(-x)$ for all $x\in\mathbb{R}$. The first inequality is derived as follows:
\begin{align*}
    \text{Define}~f(\rho)=\frac{\rho(x-\mu_1)-(t-\mu_2)}{\sqrt{1-\rho^2}}~\Rightarrow~f'(\rho)=\frac{x-\rho(t-\mu_2)}{(1-\rho^2)^{\tfrac{2}{3}}}
\end{align*}
Since $t$ goes to $\infty$, we consider $t\geqslant\max(\mu_1,\mu_2)$. Moreover, we consider $t$ such that $\frac{t-\mu_1}{t-\mu_2}>\rho_0$. This is possible since $\frac{t-\mu_1}{t-\mu_2}\to1$ as $t\to\infty$. Then
\begin{equation*}
    f'(\rho)=\frac{x-\rho(t-\mu_2)}{(1-\rho^2)^{\tfrac{2}{3}}}>\frac{(t-\mu_1)-\rho(t-\mu_2)}{(1-\rho^2)^{\tfrac{2}{3}}}\geqslant\frac{(t-\mu_1)-\rho_0(t-\mu_2)}{(1-\rho^2)^{\tfrac{2}{3}}}>0~.
\end{equation*}
Hence, $f(\rho)$ is increasing and $f(\rho)\leqslant f(\rho_0)$.

Now we prove the pairwise asymptotic independence of transformed statistics for two-sided and one-sided p-values separately by checking the definition:

\paragraph{Part I. Transformed statistics from two-sided p-values $X_i = Q_F(2\Phi(|T_i|)-1)$.} Without loss of generalization, we assume $\omega_i\geqslant\omega_j>0$. We start with $\left(\omega_iX_{i}^{+}, \omega_jX_{j}^{+}\right)$:
\begin{equation}
\label{eq:quasi_indep_pos_neg_limit_part_I}
\begin{aligned}
&\lim _{x \rightarrow +\infty}\sup_{\rho_{ij}\in[-\rho_0,\rho_0]}\frac{\Pr\left(\omega_iX_{i}^+>x,\omega_jX_{j}^+>x\right)}{\Pr\left(\omega_iX_i>x\right)+\Pr\left(\omega_jX_j>x\right)}
\leqslant \lim _{x \rightarrow +\infty}\sup_{\rho_{ij}\in[-\rho_0,\rho_0]}\frac{\Pr\left(X_i>\tfrac{x}{\omega_i},X_j>\tfrac{x}{\omega_j}\right)}{\Pr\left(X_i>\tfrac{x}{\omega_i}\right)}\\
\leqslant&\lim _{x \rightarrow +\infty}\sup_{\rho_{ij}\in[-\rho_0,\rho_0]} \frac{\Pr\left(X_i>\tfrac{x}{\omega_i},X_j>\tfrac{x}{\omega_i}\right)}{\Pr\left(X_i>\tfrac{x}{\omega_i}\right)}\\
=&\lim _{x \rightarrow +\infty}\sup_{\rho_{ij}\in[-\rho_0,\rho_0]}\Pr\left(2\Phi(|T_i|)-1>F\left(\tfrac{x}{\omega_i}\right)~\big|~2\Phi(|T_i|)-1>F\left(\tfrac{x}{\omega_i}\right)\right),
\end{aligned}
\end{equation}
where the second inequality is due to $\omega_i\geqslant\omega_j>0$, and the equation is based on the fact that $Q_F(t)> x\Leftrightarrow t>F(x)$ for any $x\in\mathbb{R}$ according to the definition of the quantile function. Define $t=\Phi^{-1}\left(\frac{F(x/\omega_i)+1}{2}\right)$, \Cref{eq:quasi_indep_pos_neg_limit_part_I} can be rewritten as
{\small
\begin{equation}
\label{eq:quasi_indep_pos_neg_limit_part_II}
\begin{aligned}
&\lim _{x \rightarrow +\infty} \sup_{\rho_{ij}\in[-\rho_0,\rho_0]}\frac{\Pr\left(\omega_iX_{i}^+>x,\omega_jX_{j}^+>x\right)}{\Pr\left(\omega_iX_i>x\right)+\Pr\left(\omega_jX_j>x\right)} 
\leqslant\lim _{t \rightarrow +\infty}\sup_{\rho_{ij}\in[-\rho_0,\rho_0]}\Pr\left(|T_i|>t\ |\ |T_j|>t\right)\\
=&\lim _{t\rightarrow +\infty}\sup_{\rho_{ij}\in[-\rho_0,\rho_0]}\frac{\Pr\left(T_i>t,T_j>t\right)+ \Pr\left(-T_i>t,T_j>t\right)+\Pr\left(T_i>t,-T_j>t\right)+ \Pr\left(-T_i>t,-T_j>t\right)}{\Pr\left(T_j>t\right)+\Pr\left(-T_j>t\right)}\\
\leqslant&\lim _{t \rightarrow +\infty}\sup_{\rho_{ij}\in[-\rho_0,\rho_0]}\Pr\left(T_i>t\mid T_j>t\right)+\Pr\left(-T_i>t\mid T_j>t\right)+\Pr\left(T_i>t\mid -T_j>t\right)+ \Pr\left(-T_i>t\mid -T_j>t\right)
=0\ .
\end{aligned}
\end{equation}
}
where the last equation utilizes \eqref{eq:normal_asym_indep} together with the fact that $(-T_i,T_j),(T_i,-T_j),(-T_i,-T_j)$ are also bivariate-normally distributed given the normality of $(T_i,T_j)$.

Next, we will check $\left(\omega_iX_{i}^{+}, \omega_jX_{j}^{-}\right)$ and $\left(\omega_iX_{i}^{-}, \omega_jX_{j}^{+}\right)$. Without loss of generality, we only consider $\left(\omega_iX_{i}^{+}, \omega_jX_{j}^{-}\right)$, and $\left(\omega_iX_{i}^{-}, \omega_jX_{j}^{+}\right)$ follows exactly the same proof.
\begin{equation}
\label{eq:quasi_indep_pos_neg_limit_step2}
\begin{split}
&\lim_{x \rightarrow +\infty}\sup_{\rho_{ij}\in[-\rho_0,\rho_0]}\frac{\Pr\left(\omega_iX_{i}^+>x,\omega_jX_{j}^->x\right)}{\Pr(\omega_iX_i>x)+\Pr(\omega_jX_j>x)} 
\leqslant\lim_{x \rightarrow +\infty}\sup_{\rho_{ij}\in[-\rho_0,\rho_0]}\frac{\Pr\left(X_{i}>\tfrac{x}{\omega_i},X_{j}<-\tfrac{x}{\omega_i}\right)}{\Pr(X_i>\tfrac{x}{\omega_i})}\\
\leqslant&\lim _{x \rightarrow +\infty}\sup_{\rho_{ij}\in[-\rho_0,\rho_0]}\Pr\left(X_{j}<-\tfrac{x}{\omega_i}\mid  X_{i}>\tfrac{x}{\omega_i}\right)
\leqslant\lim _{x \rightarrow +\infty}\sup_{\rho_{ij}\in[-\rho_0,\rho_0]}\Pr\left(X_{j}\leqslant-\tfrac{x}{\omega_i}\mid  X_{i}>\tfrac{x}{\omega_i}\right)\\
=& \lim _{x \rightarrow +\infty}\sup_{\rho_{ij}\in[-\rho_0,\rho_0]}\Pr\left(2\Phi(|T_{j}|)-1\leqslant F\left(-\tfrac{x}{\omega_i}\right)\mid  2\Phi(|T_{j}|)-1>F\left(\tfrac{x}{\omega_i}\right)\right)
\end{split} 
\end{equation}
where the second inequality is due to $\omega_i\geqslant\omega_j$, and the last equation is based on the fact that $Q_F(t)\leqslant x\Leftrightarrow t\leqslant F(x)$ and $Q_F(t)> x\Leftrightarrow t>F(x)$ for any $x\in\mathbb{R}$ according to the definition of the quantile function. Consider the change of variable: $t_1(x)=\Phi^{-1}\left(\frac{F\left(-x/\omega_i\right)+1}{2}\right)$ and $t_2(x)=\Phi^{-1}\left(\frac{F\left(x/\omega_i\right)+1}{2}\right)$. Then, 
\begin{equation}
\label{eq:quasi_indep_pos_neg_limit_step3}
\eqref{eq:quasi_indep_pos_neg_limit_step2}= \lim _{\substack{t_1 \to0\\ t_2 \to+\infty}}\sup_{\rho_{ij}\in[-\rho_0,\rho_0]}\Pr\left(|T_j|<t_1\mid|T_i|>t_2\right), 
\end{equation}
if the latter limit exists.
We further set $Z=\frac{(T_j-\mu_j)-\rho(T_i-\mu_i)}{\sqrt{1-\rho^2}}\sim\mathbb{N}\left(0,1\right)$ ($\rho=\rho_{ij}=\rho_{ji}$), and $Z$ and $T_i$ are independent by construction. Therefore, 
\begin{equation}
\label{eq:quasi_indep_pos_neg_limit_step4}
\begin{split}
\Pr\left(|T_j|<t_1\mid |T_i|>t_2\right) & = \Pr\left(-\frac{t_1+\rho T_i+\mu_j-\rho\mu_i}{\sqrt{1-\rho^2}}<Z<\frac{t_1-\rho T_i-\mu_j+\rho\mu_i}{\sqrt{1-\rho^2}}\ \bigg|\ |T_i|>t_2\right)\\
& = \frac{E\left[E\left(1_{\left\{-\frac{t_1+\rho T_i+\mu_j-\rho\mu_i}{\sqrt{1-\rho^2}}<Z<\frac{t_1-\rho T_i-\mu_j+\rho\mu_i}{\sqrt{1-\rho^2}}\right\}}1_{\{|T_i|>t_2\}}\ \bigg|\ T_i\right)\right]}{\Pr\left(|T_i|>t_2\right)}\\
& = \frac{E\left[\left(\Phi(\frac{t_1-\rho T_i-\mu_j+\rho\mu_i}{\sqrt{1-\rho^2}})-\Phi(\frac{-t_1-\rho T_i-\mu_j+\rho\mu_i}{\sqrt{1-\rho^2}})\right)1_{\{|T_i|>t_2\}}\right]}{\Pr\left(|T_i|>t_2\right)}\\
& \leqslant \max_{t}\phi(t)\times\frac{2t_1}{\sqrt{1-\rho^2}}=\sqrt{\frac{2}{\pi}}\frac{t_1}{\sqrt{1-\rho^2}}
\leqslant\sqrt{\frac{2}{\pi}}\frac{t_1}{\sqrt{1-\rho_0^2}}~,
\end{split}
\end{equation}
where the inequality applies the mean value theorem. Plug in \eqref{eq:quasi_indep_pos_neg_limit_step4}, \eqref{eq:quasi_indep_pos_neg_limit_step3} can be extended as
\begin{equation*}
\eqref{eq:quasi_indep_pos_neg_limit_step2}
= \lim _{\substack{t_1 \to0\\ t_2 \to+\infty}}\sup_{\rho_{ij}\in[-\rho_0,\rho_0]}\Pr\left(|T_j|<t_1\mid|T_i|>t_2\right)
\leqslant\lim _{t_1 \to0}\sqrt{\frac{2}{\pi}}\frac{t_1}{\sqrt{1-\rho_0^2}}=0
\end{equation*}
Accordingly,
\begin{equation}
\label{eq:quasi_indep_pos_neg_limit_final2}
\lim_{x \rightarrow +\infty}\sup_{\rho_{ij}\in[-\rho_0,\rho_0]}\frac{\Pr\left(\omega_iX_{i}^+>x,\omega_jX_{j}^->x\right)}{\Pr(\omega_iX_i>x)+\Pr(\omega_jX_j>x)} = 0
\end{equation}
Following exactly the same derivation, we also have
\begin{equation}
\label{eq:quasi_indep_pos_neg_limit_final3}
\lim_{x \rightarrow +\infty}\sup_{\rho_{ij}\in[-\rho_0,\rho_0]}\frac{\Pr\left(\omega_iX_{i}^->x,\omega_jX_{j}^+>x\right)}{\Pr(\omega_iX_i>x)+\Pr(\omega_jX_j>x)} = 0
\end{equation}

The first part is done by combining \eqref{eq:quasi_indep_pos_neg_limit_part_II}, \eqref{eq:quasi_indep_pos_neg_limit_final2} and \eqref{eq:quasi_indep_pos_neg_limit_final3}.

\paragraph{Part II. Transformed statistics from one-sided p-values $X_i = Q_F(\Phi(T_i))$.}
We start by checking $(\omega_iX_i^+,\omega_iX_j^+)$.
\begin{equation}
\label{eq:asymp_indep_one_sided_p}
\begin{split}
&\lim _{x \rightarrow +\infty}\sup_{\rho_{ij}\in[-\rho_0,\rho_0]}\frac{\Pr\left(\omega_iX_{i}^+>x,\omega_jX_{j}^+>x\right)}{\Pr\left(\omega_iX_i>x\right)+\Pr\left(\omega_jX_j>x\right)}
\leqslant\lim _{x \rightarrow +\infty}\sup_{\rho_{ij}\in[-\rho_0,\rho_0]}\frac{\Pr\left(X_i>\tfrac{x}{\omega_i},X_j>\tfrac{x}{\omega_j}\right)}{\Pr\left(X_i>\tfrac{x}{\omega_i}\right)} \\
\leqslant& \lim _{x \rightarrow +\infty}\sup_{\rho_{ij}\in[-\rho_0,\rho_0]} \frac{\Pr\left(X_i>\tfrac{x}{\omega_i},X_j>\tfrac{x}{\omega_i}\right)}{\Pr\left(X_i>\tfrac{x}{\omega_i}\right)}
=\lim _{x \rightarrow +\infty}\sup_{\rho_{ij}\in[-\rho_0,\rho_0]}\Pr\left(\Phi(T_{i})>F\left(\tfrac{x}{\omega_i}\right)~\big|~\Phi(T_{j})>F\left(\tfrac{x}{\omega_i}\right)\right)\\
=&\lim_{x\to+\infty}\sup_{\rho_{ij}\in[-\rho_0,\rho_0]}\Pr\left(T_i>\Phi^{-1}\left(F\left(\tfrac{x}{\omega_i}\right)\right)~\big|~ T_j>\Phi^{-1}\left(F\left(\tfrac{x}{\omega_i}\right)\right)\right)\\
=&\lim_{t\to+\infty}\sup_{\rho_{ij}\in[-\rho_0,\rho_0]}\Pr\left(T_i>t\mid T_j>t\right)=0,
\end{split}
\end{equation}
where the first equality is again based on the fact that $Q_F(t)> x\Leftrightarrow t>F(x)$ for any $x\in\mathbb{R}$ according to the definition of the quantile function, and the last two equations use $t=\Phi^{-1}\left(F\left(x/{\omega_i}\right)\right)$ and \eqref{eq:normal_asym_indep} respectively.

Then we check $(\omega_iX_i^+,\omega_jX_j^-)$ and $(\omega_iX_i^-,\omega_jX_j^+)$.
Since $F$ satisfies $F(-x)\leqslant 1-F(x)$ for sufficiently large $x$. Then, denote $t=\Phi^{-1}(F\left(x/\omega_i\right))$ and we have
{\small
\begin{equation}
\label{eq:asymp_indep_one_sided_p_pos}
\begin{split}
&\lim_{x \rightarrow +\infty}\sup_{\rho_{ij}\in[-\rho_0,\rho_0]}\frac{\Pr\left(\omega_iX_{i}^+>x,\omega_jX_{j}^->x\right)}{\Pr(\omega_iX_i>x)+\Pr(\omega_jX_j>x)} \leqslant\lim_{x \rightarrow +\infty}\sup_{\rho_{ij}\in[-\rho_0,\rho_0]}\frac{\Pr\left(X_{i}>\tfrac{x}{\omega_i},X_{j}<-\tfrac{x}{\omega_j}\right)}{\Pr(X_i>\tfrac{x}{\omega_i})}\\
\leqslant&\lim_{x \rightarrow +\infty}\sup_{\rho_{ij}\in[-\rho_0,\rho_0]}\frac{\Pr\left(X_{i}>\tfrac{x}{\omega_i},X_{j}\leqslant-\tfrac{x}{\omega_i}\right)}{\Pr(X_i>\tfrac{x}{\omega_i})}
=\lim_{x\to +\infty}\sup_{\rho_{ij}\in[-\rho_0,\rho_0]}\frac{\Pr\left(\Phi(T_i)>F\left(\tfrac{x}{\omega_i}\right),\ \Phi(T_j)\leqslant F\left(-\tfrac{x}{\omega_i}\right)\right)}{\Pr\left(\Phi(T_i)>F\left(\tfrac{x}{\omega_i}\right)\right)}\\
\leqslant&\lim_{x\to +\infty}\sup_{\rho_{ij}\in[-\rho_0,\rho_0]}\frac{\Pr\left(\Phi(T_i)>F\left(\tfrac{x}{\omega_i}\right),\ \Phi(T_j)\leqslant 1-F\left(\tfrac{x}{\omega_i}\right)\right)}{\Pr\left(\Phi(T_i)>F\left(\tfrac{x}{\omega_i}\right)\right)}\\
=&\lim_{x\to+\infty}\sup_{\rho_{ij}\in[-\rho_0,\rho_0]}\frac{\Pr\left(\Phi(T_i)>F\left(\tfrac{x}{\omega_i}\right),\ \Phi(-T_j)\geqslant F\left(\tfrac{x}{\omega_i}\right)\right)}{\Pr\left(\Phi(T_i)>F\left(\tfrac{x}{\omega_i}\right)\right)}
=\lim _{t \to+\infty}\sup_{\rho_{ij}\in[-\rho_0,\rho_0]}\Pr\left(-T_j>t\mid T_i>t\right) = 0
\end{split}
\end{equation}
}
where the first equality is based on the definition of the quantile function, the third inequality is based on $F\left(-\tfrac{x}{\omega_i}\right)\leqslant1-F\left(\tfrac{x}{\omega_i}\right)$ for sufficiently large $x$, and the last equality utilizes \Cref{eq:normal_asym_indep} together with the fact that $(-T_j,T_i)$ is also bivariate-normally distributed given the normality of $(T_i,T_j)$. 

Similarly, we can get 
\begin{equation}
\label{eq:asymp_indep_one_sided_p_3}
\lim_{x \rightarrow +\infty}\sup_{\rho_{ij}\in[-\rho_0,\rho_0]} \frac{\Pr\left(\omega_iX_{i}^->x,\omega_jX_{j}^+>x\right)}{\Pr(\omega_iX_i>x)+\Pr(\omega_jX_j>x)}=0.
\end{equation}
The second part is done by combining \cref{eq:asymp_indep_one_sided_p,eq:asymp_indep_one_sided_p_pos,eq:asymp_indep_one_sided_p_3}.

The lemma follows by summarizing the results of part I and II.
\end{proof}

\begin{proof}[Proof of \Cref{lemma:tail_prob_max_vs_sum_uniform_version}]
First, by \Cref{lemma:normal_to_regularly_varying}, $X_i$ belongs to $\mathscr{R}$. Hence $\omega_iX_i$ also belongs to $\mathscr{R}$. Further, on basis of \Cref{lemma:quasiindep}, the transformed weighted statistics $\omega_iX_i$ are pairwise quasi-asymptotically independent with any choice of $\rho_{ij}\in[-\rho_0,\rho_0]$:
\begin{align*}
    &\lim_{x \rightarrow +\infty} \sup_{\rho_{ij}\in[-\rho_0,\rho_0]}\frac{\Pr\left(\omega_iX_{i}^+>x,\omega_jX_{j}^+>x\right)}{\Pr(\omega_iX_i>x)+\Pr(\omega_jX_j>x)}=0,\\
    &\lim_{x \rightarrow +\infty} \sup_{\rho_{ij}\in[-\rho_0,\rho_0]}\frac{\Pr\left(\omega_iX_{i}^+>x,\omega_jX_{j}^->x\right)}{\Pr(\omega_iX_i>x)+\Pr(\omega_jX_j>x)}=0,\\
    &\lim_{x \rightarrow +\infty} \sup_{\rho_{ij}\in[-\rho_0,\rho_0]}\frac{\Pr\left(\omega_iX_{i}^->x,\omega_jX_{j}^+>x\right)}{\Pr(\omega_iX_i>x)+\Pr(\omega_jX_j>x)}=0. 
\end{align*}

Then, by \Cref{thm:asymp_type_I_error_control}, the right tail probability of the distribution of $S_{n, \vec{\omega}}=\sum_{i=1}^n\omega_iX_i$ has the following property:
\begin{equation}
\label{eq:weighted_sum_tail_uniform_convergence_lower_bound}
    \lim _{x \rightarrow +\infty} \sup_{\rho_{ij}\in[-\rho_0,\rho_0]}\frac{\Pr\left(S_{n, \vec{\omega}}>x\right)}{\sum_{i=1}^n\Pr(\omega_iX_i>x)}
    \geqslant\sup_{\rho_{ij}\in[-\rho_0,\rho_0]}\lim _{x \rightarrow +\infty} \frac{\Pr\left(S_{n, \vec{\omega}}>x\right)}{\sum_{i=1}^n\Pr(\omega_iX_i>x)}=1.
\end{equation}
For arbitrary fixed $0<\epsilon<1$,
\begin{align*}
    &\Pr\left(S_{n, \vec{\omega}}>x\right)\le\Pr\big(\cup_{i=1}^n\left\{\omega_iX_i>(1-\epsilon)x\right\}\big)+\Pr\big(S_{n, \vec{\omega}}>x,\cap_{i=1}^n\left\{\omega_iX_i\le(1-\epsilon)x\right\}\big)\\
    &\le\sum_{i=1}^n\Pr(\omega_iX_i>(1-\epsilon)x)+\sum_{i=1}^n\Pr\big(\omega_iX_i>x/n,S_{n, \vec{\omega}}-\omega_iX_i>\epsilon x\big)\\
    &\le\sum_{i=1}^n\Pr(\omega_iX_i>(1-\epsilon)x)+\sum_{1\le i\neq j\le n}^n\Pr\left(\omega_iX_i>\frac{x}{n}\wedge\frac{\epsilon x}{n-1},\omega_jX_j>\frac{x}{n}\wedge\frac{\epsilon x}{n-1}\right).
\end{align*}
Hence, plugging in \eqref{eq:uniform_asymp_indep}, it holds that
\begin{equation}
\label{eq:weighted_sum_tail_uniform_convergence_upper_bound}
    \lim _{x \rightarrow +\infty} \sup_{\rho_{ij}\in[-\rho_0,\rho_0]}\frac{\Pr\left(S_{n, \vec{\omega}}>x\right)}{\sum_{i=1}^n\Pr(\omega_iX_i>x)}
    \le (1-\epsilon)^{-\gamma}+0=(1-\epsilon)^{-\gamma}~.
\end{equation} 

It follows from \eqref{eq:weighted_sum_tail_uniform_convergence_lower_bound} and \eqref{eq:weighted_sum_tail_uniform_convergence_upper_bound} with $\epsilon\to0$ that
\begin{equation*}
    \lim _{x \rightarrow +\infty} \sup_{\rho_{ij}\in[-\rho_0,\rho_0]}\frac{\Pr\left(S_{n, \vec{\omega}}>x\right)}{\sum_{i=1}^n\Pr(\omega_iX_i>x)}=1.
\end{equation*}

For the tail probability of the maximum, we follow a similar proof to that of \Cref{cor:tail_prob_max_vs_sum}:
\begin{equation*}
    \sum_{i=1}^n\Pr(\omega_iX_i>x)-\sum_{i\neq j}\Pr(\omega_iX_i>x,\omega_jX_j>x)
    \leqslant \Pr(\max_{i=1,\ldots,n}\omega_iX_i>x)
    \leqslant \sum_{i=1}^n\Pr(\omega_iX_i>x)
\end{equation*}
By \Cref{lemma:quasiindep},
\begin{equation*}
    \lim_{x\to+\infty}\sup_{\rho\in[-\rho_0,\rho_0]}\frac{\sum_{i\neq j}\Pr(\omega_iX_i>x,\omega_jX_j>x)}{\sum_{i=1}^n\Pr(\omega_iX_i>x)}=0~,
\end{equation*}
and hence
\begin{equation*}
    \lim_{x\to+\infty}\sup_{\rho\in[-\rho_0,\rho_0]}\frac{\Pr(\max_{i=1,\ldots,n}\omega_iX_i>x)}{\sum_{i=1}^n\Pr(\omega_iX_i>x)}=1~.
\end{equation*}
\end{proof}

\begin{proof}[Proof of \Cref{lemma:neg_cor_pair_prob}]
(i) We prove by contradiction. Suppose $\delta_\alpha$ does not converge to 0 as $\alpha\to0^+$, namely, there exists a constant $c>0$ such that for sufficiently small $\alpha$, $Q_F\left(1-\alpha^{c_0}\right)\geqslant cQ_F\left(1-\alpha\right)$.

On one hand,
\begin{equation*}
    \lim_{\alpha\to0^+}\frac{\bar{F}\left(Q_F\left(1-\alpha^{c_0}\right)\right)}{\bar{F}\left(Q_F\left(1-\alpha\right)\right)}=\lim_{\alpha\to0^+}\alpha^{c_0-1}=\infty \ .
\end{equation*}

On the other hand,
\begin{equation*}
    \lim_{\alpha\to0^+}\frac{\bar{F}\left(Q_F\left(1-\alpha^{c_0}\right)\right)}{\bar{F}\left(Q_F\left(1-\alpha\right)\right)}\leqslant\lim_{\alpha\to0^+}\frac{\bar{F}\left(cQ_F\left(1-\alpha\right)\right)}{\bar{F}\left(Q_F\left(1-\alpha\right)\right)}=c^{-\gamma}<\infty \ ,
\end{equation*}
where $\gamma$ is the tail index of $F$. This leads to contradiction and thus $\delta_\alpha\to0$ must hold.

(ii) It is straightforward to see the conclusion by noting that $\delta_\alpha Q_F\left(1-\alpha\right)=(n-1)Q_F\left(1-\alpha^{c_0}\right)$.

(iii) Since $\delta_\alpha\to0$ as $\alpha\to0^+$, for any $\epsilon>0$, there exists a $c_\epsilon>0$ such that for all $\alpha<c_\epsilon$, $\delta_\alpha<\epsilon$. Accordingly, for all $\alpha<c_\epsilon$, $\bar{F}\left((1+\delta_\alpha) Q_F\left(1-\alpha\right)\right)\geqslant\bar{F}\left((1+\epsilon) Q_F\left(1-\alpha\right)\right)$, and hence
\begin{equation*}
\label{eq:term_I_fraction_is_1}
1\geqslant\lim_{\alpha\to0^+}\frac{\bar{F}\left((1+\delta_\alpha) Q_F\left(1-\alpha\right)\right)}{\bar{F}\left(Q_F\left(1-\alpha\right)\right)}\geqslant\lim_{\alpha\to0^+}\frac{\bar{F}\left((1+\epsilon) Q_F\left(1-\alpha\right)\right)}{\bar{F}\left(Q_F\left(1-\alpha\right)\right)}=(1+\epsilon)^{-\gamma}
\end{equation*}
And let $\epsilon\to0$, we prove (iii).

(iv) As the first step of the proof, we will show that 
\begin{equation*}
\label{eq:two_h_ratio_limit}
\lim_{\alpha\to0^+}\frac{h\left(\frac{1}{\omega_X}(1+\delta_\alpha)Q_F\left(1-\alpha\right)\right)}{h\left(\frac{1}{\omega_Y}\frac{\delta_\alpha}{n-1}Q_F\left(1-\alpha\right)\right)}=\frac{1}{\sqrt{c_0}}>1 \ .
\end{equation*}

Based on the fact that $\lim_{x\to1}\frac{\Phi^{-1}(x)}{\sqrt{-2\log(1-x)}}=1$, we have
\begin{equation*}
\begin{aligned}
    &\lim_{\alpha\to0^+}\frac{h\left(\frac{1}{\omega_X}(1+\delta_\alpha)Q_F\left(1-\alpha\right)\right)}{h\left(\frac{1}{\omega_Y}\frac{\delta_\alpha}{n-1}Q_F\left(1-\alpha\right)\right)}=\lim_{\alpha\to0^+}\frac{h\left(\frac{1}{\omega_X}Q_F\left(1-\alpha\right)+\frac{1}{\omega_X}(n-1)Q_F\left(1-\alpha^{c_0}\right)\right)}{h\left(\frac{1}{\omega_Y}Q_F\left(1-\alpha^{c_0}\right)\right)}\\
=&\lim_{\alpha\to0^+}\sqrt{ \frac{\log\left(1-F\left(\frac{1}{\omega_X}Q_F\left(1-\alpha\right)+\frac{1}{\omega_X}(n-1)Q_F\left(1-\alpha^{c_0}\right)\right)\right)}{\log\left(1-F\left(\frac{1}{\omega_Y}Q_F\left(1-\alpha^{c_0}\right)\right)\right)} }
\end{aligned}
\end{equation*}

Note that $c_0=\frac{3}{2}-\frac{1}{1+\rho_0}<1$, and
\begin{align*}
&\lim_{\alpha\to0^+}{ \frac{\log\left(1-F\left(\frac{1}{\omega_X}Q_F\left(1-\alpha\right)+\frac{1}{\omega_X}(n-1)Q_F\left(1-\alpha^{c_0}\right)\right)\right)}{\log\left(1-F\left(\frac{1}{\omega_Y}Q_F\left(1-\alpha^{c_0}\right)\right)\right)}}\\
=&\lim_{\alpha\to0^+}\frac{\log\left(\bar{F}\left(\frac{1}{\omega_X}Q_F\left(1-\alpha\right)+\frac{1}{\omega_X}(n-1)Q_F\left(1-\alpha^{c_0}\right)\right)\right)}{\log({\omega_Y}^\gamma)+ c_0\log\left({\alpha}\right)}\\
= & \lim_{\alpha\to0^+}\frac{\log\left(\frac{\bar{F}\left(\frac{1}{\omega_X}Q_F\left(1-\alpha\right)+\frac{1}{\omega_X}(n-1)Q_F\left(1-\alpha^{c_0}\right)\right)}{\bar{F}(Q_F\left(1-\alpha\right))}\right)+\log({\alpha})}{\log({\omega_Y}^\gamma)+c_0\log\left({\alpha}\right)}
= \lim_{\alpha\to0^+}\frac{\log(\omega_X^\gamma)+\log({\alpha})}{\log({\omega_Y}^\gamma)+c_0\log\left({\alpha}\right)} = \frac{1}{c_0}
\end{align*}
where the third equality utilizes part (iii) and thus
\begin{equation*}
    \lim_{\alpha\to0}\frac{h\left(\frac{1}{\omega_X}(1+\delta_\alpha)Q_F\left(1-\alpha\right)\right)}{h\left(\frac{1}{\omega_Y}\frac{\delta_\alpha}{n-1}Q_F\left(1-\alpha\right)\right)} = \frac{1}{\sqrt{c_0}}>1.
\end{equation*}

Now we are ready to prove part (iv). Denote $\mu_1$, $\mu_2$ the mean of $X$ and $Y$. To simplify the notation, denote
\begin{equation*}
\begin{split}
&\tilde{h}_1(\alpha)=h\left(\frac{1}{\omega_X}\left(1+\delta_\alpha\right) Q_F\left(1-\alpha\right)\right)=h\left(\frac{1}{\omega_X}Q_F\left(1-\alpha\right)+\frac{1}{\omega_X}(n-1)Q_F\left(1-\alpha^{c_0}\right)\right),\\
&\tilde{h}_2(\alpha)=h\left(\frac{1}{\omega_Y}\frac{\delta_\alpha}{n-1}  Q_F\left(1-\alpha\right)\right)=h\left(\frac{1}{\omega_Y}Q_F\left(1-\alpha^{c_0}\right)\right). 
\end{split}
\end{equation*}

Then we have
\begin{align}\label{eq:iv-1}
\begin{split}
&\lim_{\alpha\to0^+}\sup_{\rho\in[-\rho_0,\rho_0]}\Pr\left(Y\geqslant\tilde{h}_2(\alpha)\mid X\geqslant\tilde{h}_1(\alpha)\right)\\
\leqslant&\lim_{\alpha\to0^+}\sup_{\rho\in[-\rho_0,\rho_0]}\Pr\left(Y\geqslant\tilde{h}_1(\alpha)\mid X\geqslant\tilde{h}_1(\alpha)\right)+\lim_{\alpha\to0^+}\sup_{\rho\in[-\rho_0,\rho_0]}\Pr\left(\tilde{h}_2(\alpha)\leqslant Y\leqslant\tilde{h}_1(\alpha)\mid X\geqslant\tilde{h}_1(\alpha)\right)\\
{=}&\lim_{\alpha\to0^+}\sup_{\rho\in[-\rho_0,\rho_0]}\Pr\left(\tilde{h}_2(\alpha)\leqslant Y\leqslant\tilde{h}_1(\alpha)\mid X\geqslant\tilde{h}_1(\alpha)\right)
\end{split}
\end{align}
where the last equality uses \Cref{eq:normal_asym_indep}.

Furthermore, it holds that
{\small
\begin{equation}\label{eq:iv-2}
\begin{split}
&\lim_{\alpha\to0^+}\sup_{\rho\in[-\rho_0,\rho_0]}\Pr\left(\tilde{h}_2(\alpha)\leqslant Y\leqslant\tilde{h}_1(\alpha)\mid X\geqslant\tilde{h}_1(\alpha)\right)\\
=&\lim_{\alpha\to0^+}\sup_{\rho\in[-\rho_0,\rho_0]}\frac{\Pr\left(\tilde{h}_2(\alpha)\leqslant Y\leqslant \tilde{h}_1(\alpha),\ X\geqslant\tilde{h}_1(\alpha)\right)}{\Pr\left(X\geqslant\tilde{h}_1(\alpha)\right)}\\
=&\lim_{\alpha\to0^+}\sup_{\rho\in[-\rho_0,\rho_0]}\frac{E\left[\Pr\left(X\geqslant\tilde{h}_1(\alpha)|\ Y\right)1_{\{\tilde{h}_2(\alpha)\leqslant Y\leqslant\tilde{h}_1(\alpha)\}}\right]}{\Pr\left(X\geqslant\tilde{h}_1(\alpha)\right)}\\
=&\lim_{\alpha\to0^+}\sup_{\rho\in[-\rho_0,\rho_0]}\int_{\tilde{h}_2(\alpha)}^{\tilde{h}_1(\alpha)}\frac{1-\Phi\left(\frac{\tilde{h}_1(\alpha)-\rho y-\mu_1+\rho\mu_2}{\sqrt{1-\rho^2}}\right)}{1-\Phi\left(\tilde{h}_1(\alpha)-\mu_1\right)}\phi(y-\mu_2)dy\\
\leqslant&\lim_{\alpha\to0^+}\sup_{\rho\in[-\rho_0,\rho_0]}\frac{1-\Phi\left(\sqrt{\frac{1-\rho}{1+\rho}}\tilde{h}_1(\alpha)-\frac{\mu_1-\rho\mu_2}{\sqrt{1-\rho^2}}\right)}{1-\Phi\left(\tilde{h}_1(\alpha)-\mu_1\right)}\left[\Phi(\tilde{h}_1(\alpha)-\mu_2)-\Phi(\tilde{h}_2(\alpha)-\mu_2)\right]\\
=&\underbrace{\lim_{\alpha\to0^+}\sup_{\rho\in[-\rho_0,\rho_0]}\frac{1-\Phi\left(\sqrt{\frac{1-\rho}{1+\rho}}\tilde{h}_1(\alpha)-\frac{\mu_1-\rho\mu_2}{\sqrt{1-\rho^2}}\right)}{1-\Phi\left(\tilde{h}_1(\alpha)-\mu_1\right)}\times\left(1-\Phi(\tilde{h}_2(\alpha)-\mu_2)\right)}_{\text{I}}\times\underbrace{\lim_{\alpha\to0^+}\left[1-\frac{1-\Phi(\tilde{h}_1(\alpha)-\mu_2)}{1-\Phi(\tilde{h}_2(\alpha)-\mu_2)}\right]}_{\text{II}}
\end{split}
\end{equation}
}
where the third equality follows from the bivariate normality of $X$ and $Y$.

Since $\lim_{x\to+\infty}\frac{1-\Phi(x)}{\phi(x)/x}=1$ where $\phi(x)$ is the density of the standard normal, 
\begin{equation}\label{eq:iv-3}
    \lim_{\alpha\to0}\frac{1-\Phi(\tilde{h}_1(\alpha)-\mu_2)}{1-\Phi(\tilde{h}_2(\alpha)-\mu_2)}=\lim_{\alpha\to0}\frac{\tilde{h}_2(\alpha)-\mu_2}{\tilde{h}_1(\alpha)-\mu_2}\exp\left\{-\frac{1}{2}\left(\tilde{h}_1(\alpha)^2-\tilde{h}_2(\alpha)^2\right)+\mu_2\left(\tilde{h}_1(\alpha)-\tilde{h}_2(\alpha)\right)\right\}=0 \ .
\end{equation}
Accordingly, term II goes to 1. And for term I, 
\begin{equation}\label{eq:iv-4}
\begin{split}
&\lim_{\alpha\to0^+}\sup_{\rho\in[-\rho_0,\rho_0]} \frac{1-\Phi\left(\sqrt{\frac{1-\rho}{1+\rho}}\tilde{h}_1(\alpha)-\frac{\mu_1-\rho\mu_2}{\sqrt{1-\rho^2}}\right)}{1-\Phi\left(\tilde{h}_1(\alpha)-\mu_1\right)}\times\left(1-\Phi(\tilde{h}_2(\alpha)-\mu_2)\right)\\
\leqslant&\lim_{\alpha\to0^+}\sup_{\rho\in[-\rho_0,\rho_0]} \frac{1-\Phi\left(\sqrt{\frac{1-\rho_0}{1+\rho_0}}\tilde{h}_1(\alpha)-\frac{\mu_1+|\mu_2|}{\sqrt{1-\rho_0^2}}\right)}{1-\Phi\left(\tilde{h}_1(\alpha)-\mu_1\right)}\times\left(1-\Phi(\tilde{h}_2(\alpha)-\mu_2)\right)\\
=&c_1\lim_{\alpha\to0^+}\frac{\tilde{h}_1(\alpha)-\mu_1}{\sqrt{\frac{1-\rho_0}{1+\rho_0}}\tilde{h}_1(\alpha)-\frac{\mu_1+|\mu_2|}{\sqrt{1-\rho_0^2}}}\times\frac{1}{\tilde{h}_2(\alpha)}\times\exp\left\{\frac{\rho_0}{1+\rho_0}\tilde{h}_1(\alpha)^2-\frac{\mu_1+|\mu_2|}{1+\rho_0}\tilde{h}_1(\alpha)-\frac{1}{2}\tilde{h}_2(\alpha)^2+\mu_2\tilde{h}_2(\alpha)\right\}\\
=&c_1\sqrt{\frac{1+\rho_0}{c_0(1-\rho_0)}}\lim_{\alpha\to0^+}\frac{1}{\tilde{h}_1(\alpha)}\times\exp\left\{\frac{\rho_0}{1+\rho_0}\tilde{h}_1(\alpha)^2-\frac{\mu_1+|\mu_2|}{1+\rho_0}\tilde{h}_1(\alpha)-\frac{1}{2}\tilde{h}_2(\alpha)^2+\mu_2\tilde{h}_2(\alpha)\right\}=0
\end{split}
\end{equation}
where
\begin{equation*}
c_1=\frac{1}{\sqrt{2\pi}}\exp\left\{-\frac{\rho_0^2{\mu_1}^2+2\mu_1|\mu_2|+(2-\rho_0^2){\mu_2}^2}{2(1-\rho^2)}\right\}
\end{equation*}
The first equality again uses $\lim_{x\to+\infty}\frac{1-\Phi(x)}{\phi(x)/x}=1$ and the last one is from 
\begin{equation*}
\frac{c_0}{2}-\frac{\rho_0}{1+\rho_0}=\frac{1}{2}\left(\frac{3}{2}-\frac{1}{1+\rho_0}-\frac{2\rho_0}{1+\rho_0}\right)=\frac{3}{4}-\frac{\rho_0+1/2}{1+\rho_0}=\frac{1}{2}\left(\frac{1}{1+\rho_0}-\frac{1}{2}\right)>0,
\end{equation*}
and hence
\begin{equation*}
\begin{split}
&\lim_{\alpha\to0^+}\frac{\rho_0}{1+\rho_0}\tilde{h}_1(\alpha)^2-\frac{\mu_1+|\mu_2|}{1+\rho_0}\tilde{h}_1(\alpha)-\frac{1}{2}\tilde{h}_2(\alpha)^2+\mu_2\tilde{h}_2(\alpha)-\log\tilde{h}_1(\alpha)\\
=&\lim_{\alpha\to0^+}\tilde{h}_1(\alpha)^2\times\lim_{\alpha\to0^+}\frac{\rho_0}{1+\rho_0}-\frac{\mu_1+|\mu_2|}{(1+\rho_0)\tilde{h}_1(\alpha)}-\frac{1}{2}\frac{\tilde{h}_2(\alpha)^2}{\tilde{h}_1(\alpha)^2}+\mu_2\frac{\tilde{h}_2(\alpha)}{\tilde{h}_1(\alpha)^2}-\frac{\log\tilde{h}_1(\alpha)}{\tilde{h}_1(\alpha)^2}\\
=&\lim_{\alpha\to0^+}\tilde{h}_1(\alpha)^2\times\lim_{\alpha\to0^+}\frac{\rho_0}{1+\rho_0}-0-\frac{1}{2}c_0+\mu_2\times c_0\times 0-0\\
=&\lim_{\alpha\to0^+}-\tilde{h}_1(\alpha)^2\times\left(\frac{c_0}{2}-\frac{\rho_0}{1+\rho_0}\right)=-\infty.
\end{split}
\end{equation*}

Combine \Cref{eq:iv-1,eq:iv-2,eq:iv-3,eq:iv-4}, we reach
\begin{equation*}
\lim_{\alpha\to0^+}\sup_{\rho\in[-\rho_0,\rho_0]} \Pr\left(Y\geqslant\tilde{h}_2(\alpha)\mid X\geqslant\tilde{h}_1(\alpha)\right)
=\lim_{\alpha\to0^+}\sup_{\rho\in[-\rho_0,\rho_0]} \Pr\left(\tilde{h}_2(\alpha)\leqslant Y\leqslant\tilde{h}_1(\alpha)\mid X\geqslant\tilde{h}_1(\alpha)\right) = 0
\end{equation*}
which finishes the proof.
\end{proof}

\begin{proof}[Proof of \Cref{lemma:diff_combination_bon_prob}]
Denote $c_0=\frac{3}{2}-\frac{1}{1+\rho_0}$ and $\delta_\alpha=(n-1)\frac{Q_F\left(1-\alpha^{c_0}\right)}{Q_F\left(1-\alpha\right)}$, which are the same choices in \Cref{lemma:neg_cor_pair_prob}. Thus, $\frac{|\rho_{ij}|}{1+|\rho_{ij}|}\le\frac{\rho_0}{1+\rho_0}<c_0<1$.

Denote 
\begin{equation*}
\begin{split}
&\text{I}=\Pr\left(Q_F\left(1-\alpha\right) < \omega_iX_i\leqslant \left(1+\delta_\alpha\right) Q_F\left(1-\alpha\right), \sum_{k=1}^n \omega_kX_k<Q_F\left(1-\alpha\right)\right),\\
&\text{II}=\Pr\left(\omega_iX_i >\left(1+\delta_\alpha\right) Q_F\left(1-\alpha\right), \sum_{k=1}^n \omega_kX_k<Q_F\left(1-\alpha\right)\right),
\end{split}
\end{equation*}
and thus
\begin{equation}\label{eq:prob_of_individual_Ai_two_terms}
\Pr\left(\omega_iX_i > Q_F\left(1-\alpha\right), \sum_{k=1}^n \omega_kX_k\leqslant Q_F\left(1-\alpha\right)\right) = \text{I} + \text{II}
\end{equation}
In the following, we prove both terms I and II satisfy
\begin{equation*}
    \lim_{\alpha\to0^+}\sup_{\rho_{ij}\in[-\rho_0,\rho_0]}\frac{\triangle}{\sum_{i=1}^n\Pr\left(\omega_iX_i>Q_F(1-\alpha)\right)} = 0
\end{equation*}
where $\triangle$ can be either I or II. Then, combining with \eqref{eq:prob_of_individual_Ai_two_terms} finishes the proof.

(a). Estimate $\text{I}$.
\begin{equation*}
\label{eq:prob_of_Ai_term_I}
\begin{split}
\text{I} &\leqslant \Pr\left(Q_F\left(1-\alpha\right) < \omega_i X_i\leqslant(1+\delta_\alpha) Q_F\left(1-\alpha\right)\right) \\
&=\Pr\left(\omega_iX_i>Q_F\left(1-\alpha\right)\right)-\Pr\left(\omega_iX_i>\left(1+\delta_\alpha\right) Q_F\left(1-\alpha\right)\right)\\
&= \Pr\left(\omega_iX_i>Q_F\left(1-\alpha\right)\right) \left(1-\frac{\Pr\left(\omega_iX_i>(1+\delta_\alpha) Q_F\left(1-\alpha\right)\right)}{\Pr\left(\omega_iX_i>Q_F\left(1-\alpha\right)\right)}\right)
\end{split}
\end{equation*}
Since $\omega_iX_i$ is still regularly-varying distributed, with the same choice of $\delta_\alpha$ as in \Cref{lemma:neg_cor_pair_prob}, we have
\begin{equation*}
\lim_{\alpha\to0^+}\sup_{\rho_{ij}\in[-\rho_0,\rho_0]}\frac{\text{I}}{\Pr\left(\omega_iX_i>Q_F\left(1-\alpha\right)\right)}
\leqslant 1-\lim_{\alpha\to0^+}\frac{\Pr\left(\omega_iX_i>(1+\delta_\alpha) Q_F\left(1-\alpha\right)\right)}{\Pr\left(\omega_iX_i>Q_F\left(1-\alpha\right)\right)}=0~.
\end{equation*}
Accordingly,
\begin{equation}
\label{eq:prob_of_Ai_term_I_o(alpha)}
\lim_{\alpha\to0^+}\sup_{\rho_{ij}\in[-\rho_0,\rho_0]}\frac{\text{I}}{\sum_{i=1}^n\Pr\left(\omega_iX_i>Q_F\left(1-\alpha\right)\right)}=0~.
\end{equation}

(b). Estimate $\text{II}$. By union bound, the following upper bound holds for $\text{II}$:
\begin{equation}
\label{eq:prob_of_Ai_term_II}
\begin{split}
\text{II}\leqslant & \Pr\left(\omega_iX_i>(1+\delta_\alpha)Q_F\left(1-\alpha\right), \bigcup_{j \neq i}\left\{\omega_jX_j\leqslant -\frac{\delta_\alpha}{n-1} Q_F\left(1-\alpha\right)\right\}\right)\\
\leqslant &\sum_{j \neq i} \Pr\left(\omega_iX_i >\left(1+\delta_\alpha\right) Q_F\left(1-\alpha\right), \omega_jX_j\leqslant-\frac{\delta_\alpha}{n-1} Q_F\left(1-\alpha\right)\right)~.
\end{split}
\end{equation}
To get
\begin{equation}
\label{eq:prob_of_Ai_term_II_o(alpha)}
\lim_{\alpha\to0^+}\sup_{\rho_{ij}\in[-\rho_0,\rho_0]}\frac{\text{II}}{\sum_{i=1}^n\Pr\left(\omega_iX_i>Q_F\left(1-\alpha\right)\right)}=0~,
\end{equation}
it suffices to prove that for any $i\neq j$,
\begin{equation}
\label{eq:prob_I2_ij_pair}
\lim_{\alpha\to0^+}\sup_{\rho_{ij}\in[-\rho_0,\rho_0]}\frac{\Pr\left(\omega_iX_i >\left(1+\delta_\alpha\right) Q_F\left(1-\alpha\right), \omega_jX_j\leqslant-\frac{\delta_\alpha}{n-1} Q_F\left(1-\alpha\right)\right)}{\sum_{i=1}^n\Pr\left(\omega_iX_i>Q_F\left(1-\alpha\right)\right)}=0~.
\end{equation}

Case 1: $X_i$ and $X_j$ are transformed from two-sided p-values.
Define $g(x)=\Phi^{-1}\left(\frac{F(x)+1}{2}\right)$. Then, based on the definition of the quantile function, we have the following equivalence:
\begin{equation*}
\label{eq:condition_X_i_to_Z_i}
\begin{split}
\omega_iX_i >(1+\delta_\alpha) Q_F\left(1-\alpha\right)&\Leftrightarrow\left|T_i\right| > g\left(\frac{1}{\omega_i}(1+\delta_\alpha) Q_F\left(1-\alpha\right)\right) \\
\omega_jX_j\leqslant-\frac{\delta_\alpha}{n-1} Q_F\left(1-\alpha\right)&\Leftrightarrow\left|T_j\right|\leqslant g\left(-\frac{1}{\omega_j}\frac{\delta_\alpha}{n-1} Q_F\left(1-\alpha\right)\right) \\
\end{split}
\end{equation*}

Denote $\mu_i$ and $\mu_j$ the mean of $T_i$ and $T_j$.
Due to the bivariate normality assumption and $\left|\rho_{ji}\right|\leqslant\rho_0<1$, we can write $T_j-\mu_j=\rho_{ji} (T_i-\mu_i)+\gamma_{ji} Z_{ji}$, where ${\rho_{ji}}^2+{\gamma_{ji}}^2=1$ and $\sqrt{1-{\rho_0}^2}\leqslant\gamma_{ji}\leqslant1$, and $Z_{ji}$ is independent of $T_i$ and distributed from a standard normal. Then,
\begin{equation}
\begin{split}
\label{eq:prob_I2_ij_pair_two_sided_rewrite}
& \Pr\left(\omega_iX_i >\left(1+\delta_\alpha\right) Q_F\left(1-\alpha\right), \omega_jX_j\leqslant-\frac{\delta_{\alpha}}{n-1} Q_F\left(1-\alpha\right)\right)\\
=&\Pr\left(\left|T_i\right| > g\left(\frac{1}{\omega_i}(1+\delta_\alpha) Q_F\left(1-\alpha\right)\right) , \left|T_j\right|\leqslant g\left(-\frac{1}{\omega_j}\frac{\delta_\alpha}{n-1} Q_F\left(1-\alpha\right)\right)\right)\\
=&\Pr\left(\left|T_i\right| > g\left(\frac{1}{\omega_i}(1+\delta_\alpha) Q_F\left(1-\alpha\right)\right) , \left|\mu_j+\rho_{ji}(T_i-\mu_i)+\gamma_{ji}Z_{ji}\right|\leqslant g\left(-\frac{1}{\omega_j}\frac{\delta_\alpha}{n-1} Q_F\left(1-\alpha\right)\right)\right)\\
=&E\Biggl [1_{\left(\left|T_i\right| > g\left(\frac{1}{\omega_i}(1+\delta_\alpha) Q_F\left(1-\alpha\right)\right)\right)}\Biggl(\Phi\left(\frac{g\left(-\frac{1}{\omega_j}\frac{\delta_\alpha}{n-1} Q_F\left(1-\alpha\right)\right)-\mu_j-\rho_{ji}(T_i-\mu_i)}{\gamma_{ji}}\right)-\\
&\Phi\left(\frac{-g\left(-\frac{1}{\omega_j}\frac{\delta_\alpha}{n-1} Q_F\left(1-\alpha\right)\right)-\mu_j-\rho_{ji}(T_i-\mu_i)}{\gamma_{ji}}\right)\Biggl)\Biggl]\\
\leqslant&\sqrt{\frac{2}{\pi}}\frac{g\left(-\frac{1}{\omega_j}\frac{\delta_\alpha}{n-1} Q_F\left(1-\alpha\right)\right)}{\gamma_{ji}}\Pr\left(\left|T_i\right| > g\left(\frac{1}{\omega_i}(1+\delta_\alpha) Q_F\left(1-\alpha\right)\right)\right)\\
\leqslant&\sqrt{\frac{2}{\pi}}\frac{g\left(-\frac{1}{\omega_j}\frac{\delta_\alpha}{n-1} Q_F\left(1-\alpha\right)\right)}{\sqrt{1-\rho_0^2}}\Pr\left(\omega_iX_i >\left(1+\delta_\alpha\right) Q_F\left(1-\alpha\right)\right)\ ,
\end{split}
\end{equation}
where the inequality applies the mean value theorem and the fact that the density of the standard normal is upper bounded by $\frac{1}{\sqrt{2\pi}}$.

Since 
\begin{equation*}
\begin{split}
\lim_{\alpha\to0^+}g\left(-\frac{1}{\omega_j}\frac{\delta_\alpha}{n-1}  Q_F\left(1-\alpha\right)\right)=0
\end{split}
\end{equation*}
($\delta_\alpha Q_F\left(1-\alpha\right)\to\infty$, see \Cref{lemma:neg_cor_pair_prob}), \Cref{eq:prob_I2_ij_pair} can be verified by the following inequalities:
\begin{equation*}
\label{eq:prob_I2_ij_pair_dep_case}
\begin{split}
& \lim _{\alpha \rightarrow 0^+}\sup_{\rho_{ij}\in[-\rho_0,\rho_0]} \frac{\Pr\left(X_i >\frac{1}{\omega_i}\left(1+\delta_\alpha\right) Q_F\left(1-\alpha\right), X_j\leqslant-\frac{1}{\omega_j}\frac{\delta_\alpha}{n-1}  Q_F\left(1-\alpha\right)\right)}{\sum_{i=1}^n\Pr\left(\omega_iX_i >Q_F\left(1-\alpha\right)\right)} \\
\leqslant~& \lim _{\alpha \rightarrow 0^+}\sup_{\rho_{ij}\in[-\rho_0,\rho_0]}\frac{\Pr\left(X_i >\frac{1}{\omega_i}\left(1+\delta_\alpha\right) Q_F\left(1-\alpha\right), X_j\leqslant-\frac{1}{\omega_j}\frac{\delta_\alpha}{n-1}  Q_F\left(1-\alpha\right)\right)}{\Pr\left(\omega_iX_i >Q_F\left(1-\alpha\right)\right)} \\
\leqslant~& \lim _{\alpha \rightarrow 0^+} \frac{\Pr\left(X_i >\frac{1}{\omega_i}\left(1+\delta_\alpha\right) Q_F\left(1-\alpha\right)\right)}{\Pr\left(\omega_iX_i >Q_F\left(1-\alpha\right)\right)}\times\sqrt{\frac{2}{\pi}}\frac{g\left(-\frac{1}{\omega_j}\frac{\delta_\alpha}{n-1} Q_F\left(1-\alpha\right)\right)}{\sqrt{1-{\rho_0}^2}}\\
=~&\lim _{\alpha \rightarrow 0^+}\sqrt{\frac{2}{\pi}}\frac{g\left(-\frac{1}{\omega_j}\frac{\delta_\alpha}{n-1} Q_F\left(1-\alpha\right)\right)}{\sqrt{1-{\rho_0}^2}}=0
\end{split}
\end{equation*}
where the second inequality utilizes \Cref{eq:prob_I2_ij_pair_two_sided_rewrite}.

Case 2: $X_i$ and $X_j$ are transformed from one-sided p-values.
Define $h(x)=\Phi^{-1}\left(F(x)\right)$. Then, the following equivalence holds based on the definition of the quantile function:
\begin{equation*}
\label{eq:condition_X_i_to_Z_i_one}
\begin{split}
\omega_iX_i >(1+\delta_\alpha) Q_F\left(1-\alpha\right)&\Leftrightarrow T_i > h\left(\frac{1}{\omega_i}(1+\delta_\alpha) Q_F\left(1-\alpha\right)\right) \\
\omega_jX_j\leqslant-\frac{\delta_\alpha}{n-1} Q_F\left(1-\alpha\right)&\Leftrightarrow T_j\leqslant h\left(-\frac{1}{\omega_j}\frac{\delta_\alpha}{n-1} Q_F\left(1-\alpha\right)\right)\leqslant-h\left(\frac{1}{\omega_j}\frac{\delta_\alpha}{n-1} Q_F\left(1-\alpha\right)\right),\\
\end{split}
\end{equation*}
where the last inequality follows from the assumption that $\bar{F}(x)\geqslant F(-x)$ for sufficiently large $x$.
Due to the bivariate normality assumption, we can write $T_j-\mu_j=\rho_{ji} (T_i-\mu_i)+\gamma_{ji} Z_{ji}$, where ${\rho_{ji}}^2+{\gamma_{ji}}^2=1$, $|\rho_{ji}|\leqslant\rho_0$, and $Z_{ji}$ is independent of $T_i$ and distributed from a standard normal.

Without loss of generality, we can assume $\gamma_{ji}>0$ and hence $\sqrt{1-{\rho_0}^2}\leqslant\gamma_{ji}\leqslant1$. 

If $0<\rho_{ji}\leqslant\rho_0$, when $\alpha$ is sufficiently small,
\begin{equation*}
\rho_{ji}h\left(\frac{1}{\omega_i}(1+\delta_\alpha) Q_F\left(1-\alpha\right)\right)-\rho_{ji}\mu_i+\mu_j>0.
\end{equation*}
Then,
\begin{equation*}
\label{eq:rewrite_Z_ji_one}
Z_{ji}=\frac{T_j-\mu_j-\rho_{ji} \left(T_i-\mu_i\right)}{\gamma_{ji}}
<-\frac{h\left(\frac{1}{\omega_j}\frac{\delta_\alpha}{n-1} Q_F\left(1-\alpha\right)\right)}{\gamma_{ji}}
\leqslant-h\left(\frac{1}{\omega_j}\frac{\delta_\alpha}{n-1} Q_F\left(1-\alpha\right)\right).
\end{equation*}
As $\alpha\to0^+$, since $h\left(\frac{1}{\omega_j}\frac{\delta_\alpha}{n-1}  Q_F\left(1-\alpha\right)\right)\to\infty$, $-h\left(\frac{1}{\omega_j}\frac{\delta_\alpha}{n-1} Q_F\left(1-\alpha\right)\right)$ goes to $-\infty$.
Then,
\begin{equation*}
\label{eq:prob_I2_ij_pair_dep_case_one}
\begin{split}
& \lim _{\alpha \rightarrow 0^+}\sup_{\rho_{ij}\in[-\rho_0,\rho_0]} \frac{\Pr\left(\omega_iX_i >\left(1+\delta_\alpha\right) Q_F\left(1-\alpha\right),\ \omega_jX_j\leqslant-\frac{\delta_\alpha}{n-1}  Q_F\left(1-\alpha\right)\right)}{\sum_{i=1}^n\Pr\left(\omega_iX_i >Q_F\left(1-\alpha\right)\right)} \\
\leqslant~& \lim _{\alpha \rightarrow 0^+}\sup_{\rho_{ij}\in[-\rho_0,\rho_0]} \frac{\Pr\left(\omega_iX_i >\left(1+\delta_\alpha\right) Q_F\left(1-\alpha\right),\ \omega_jX_j\leqslant-\frac{\delta_\alpha}{n-1}  Q_F\left(1-\alpha\right)\right)}{\Pr\left(\omega_iX_i >Q_F\left(1-\alpha\right)\right)} \\
\leqslant~& \lim _{\alpha \rightarrow 0^+}\sup_{\rho_{ij}\in[-\rho_0,\rho_0]}  \frac{\Pr\left(\omega_iX_i >\left(1+\delta_\alpha\right) Q_F\left(1-\alpha\right),\ Z_{ji}<-h\left(\frac{1}{\omega_j}\frac{\delta_\alpha}{n-1} Q_F\left(1-\alpha\right)\right)\right)}{\Pr\left(\omega_iX_i > Q_F\left(1-\alpha\right)\right)}\\
=~&\lim _{\alpha \rightarrow 0^+} \frac{\Pr\left(\omega_iX_i >\left(1+\delta_\alpha\right) Q_F\left(1-\alpha\right)\right)}{\Pr\left(\omega_iX_i >Q_F\left(1-\alpha\right)\right)} \Pr\left(Z_{ji}<-h\left(\frac{1}{\omega_j}\frac{\delta_\alpha}{n-1} Q_F\left(1-\alpha\right)\right)\right)\\
=&\lim _{\alpha \rightarrow 0^+} \Pr\left(Z_{ji}<-h\left(\frac{1}{\omega_j}\frac{\delta_\alpha}{n-1} Q_F\left(1-\alpha\right)\right)\right)=0\ .
\end{split}
\end{equation*}

If $-1<\rho_{ji}<0$, with $1 = \lim _{\alpha \rightarrow 0^+}\frac{\Pr\left(X_i >\left(1+\delta_\alpha\right) Q_F\left(1-\alpha\right)\right)}{\alpha}=\lim _{\alpha \rightarrow 0^+} \frac{\Pr\left(T_i > h\left(\left(1+\delta_\alpha\right) Q_F\left(1-\alpha\right)\right)\right)}{\alpha}$, we have

\begin{equation*}
\begin{split}
& \lim _{\alpha \rightarrow 0^+}\sup_{\rho_{ij}\in[-\rho_0,\rho_0]} \frac{\Pr\left(\omega_iX_i >\left(1+\delta_\alpha\right) Q_F\left(1-\alpha\right),\ \omega_jX_j\leqslant-\frac{\delta_\alpha}{n-1}  Q_F\left(1-\alpha\right)\right)}{\sum_{i=1}^n\Pr\left(\omega_iX_i >Q_F\left(1-\alpha\right)\right)} \\
\leqslant~& \lim _{\alpha \rightarrow 0^+}\sup_{\rho_{ij}\in[-\rho_0,\rho_0]} \frac{\Pr\left(\omega_iX_i >\left(1+\delta_\alpha\right) Q_F\left(1-\alpha\right),\ \omega_jX_j\leqslant-\frac{\delta_\alpha}{n-1}  Q_F\left(1-\alpha\right)\right)}{\Pr\left(\omega_iX_i >Q_F\left(1-\alpha\right)\right)} \\
=~& \lim _{\alpha \rightarrow 0^+}\sup_{\rho_{ij}\in[-\rho_0,\rho_0]} \frac{\Pr\left(T_i > h\left(\frac{1}{\omega_i}\left(1+\delta_\alpha\right) Q_F\left(1-\alpha\right)\right),-T_j\geqslant h\left(\frac{1}{\omega_j}\frac{\delta_\alpha}{n-1}  Q_F\left(1-\alpha\right)\right)\right)}{\Pr\left(T_i > h\left(\frac{1}{\omega_j}\left(1+\delta_\alpha\right) Q_F\left(1-\alpha\right)\right)\right)}\\ 
=~& \lim _{\alpha \rightarrow 0^+}\sup_{\rho_{ij}\in[-\rho_0,\rho_0]}\Pr\left(-T_j\geqslant h\left(\frac{1}{\omega_j}\frac{\delta_\alpha}{n-1}  Q_F\left(1-\alpha\right)\right)~\big|~ T_i > h\left(\frac{1}{\omega_i}\left(1+\delta_\alpha\right) Q_F\left(1-\alpha\right)\right)\right)= \ 0 \ ,
\end{split}
\end{equation*}

where the last equality is due to the part (iv) in \Cref{lemma:neg_cor_pair_prob} and $T_i$ and $-T_j$ are positively dependent and bivariate-normally distributed. Hence, \eqref{eq:prob_I2_ij_pair} also holds for this case.

Combining Case 1 and 2, \eqref{eq:prob_I2_ij_pair} holds, and accordingly, \eqref{eq:diff_combination_bon_prob} holds by aggregating \eqref{eq:prob_of_individual_Ai_two_terms},\eqref{eq:prob_of_Ai_term_I_o(alpha)}, and \eqref{eq:prob_of_Ai_term_II_o(alpha)}.
\end{proof}

\subsection{Other theoretical results}
\begin{restatable}{proposition}{LeftTruncationCDF}
\label{prop:left_truncated_cauchy_cdf}
The left-truncated t distribution belongs to the regularly varying tailed class $\mathscr{R}$. Furthermore, its tail index $\gamma$ equals the degree of freedom of original t distribution.
\end{restatable}

\begin{proof}
Denote $X$ a random variable distributed from the student t distribution with degree of freedom $\gamma$ and $F_{t, \gamma}(x)$ its cumulative distribution function. Denote $c$ the lower bound of $X$, then the left-truncated t distribution has a cumulative distribution function:
\begin{equation*}
F(x) = \Pr(X\leqslant x \mid X\geqslant c)=\frac{F_{t, \gamma}(x) - F_{t, \gamma}(c)}{1-F_{t, \gamma}(c)},\quad x\geqslant c.
\end{equation*}
By the definition of the regularly varying tailed class,
\begin{equation*}
\begin{split}
\lim_{x\to+\infty}\frac{\bar F(xy)}{\bar F(x)}=\lim_{x\to+\infty}\frac{1-\frac{F_{t, \gamma}(xy) - F_{t, \gamma}(c)}{1-F_{t, \gamma}(c)}}{1-\frac{F_{t, \gamma}(x) - F_{t, \gamma}(c)}{1-F_{t, \gamma}(c)}}=\lim_{x\to+\infty}\frac{\bar F_{t, \gamma}(xy)}{\bar F_{t, \gamma}(x)}=y^{-\gamma}.
\end{split}
\end{equation*}
Then the proposition follows.
\end{proof}

\begin{restatable}{proposition}{PositivePValues}
\label{prop:two_sided_p_vals_pos_dep}
Suppose $(X,Y)$ is distributed from a bivariate normal with mean $\mu=(0,0)$ and covariance matrix 
$$\Sigma=\left(
\begin{array}{cc}
   1  &  \rho\\
    \rho & 1
\end{array}\right).$$
Then, the following hold:
\begin{enumerate}[label=(\roman*)]
\item \textrm{cov}$(p_1(X),p_1(Y))$ has the same sign as $\rho$ with $p_1(\cdot)=1-\Phi(\cdot)$
\item \textrm{cov}$(p_2(X),p_2(Y))\geqslant0$ with $p_2(\cdot)=2\left(1-\Phi(|\cdot|)\right)$
\end{enumerate}
\end{restatable}

\begin{proof}
(i) We first rewrite the covariance of $p_1(X)$ and $p_1(Y)$ as follows:
\begin{equation}
\label{eq:cov_p_val_oneside}
\mathrm{cov}\left(p_1(X),p_1(Y)\right)=\mathrm{cov}\left(1-\Phi(X),1-\Phi(Y)\right)\\
=\mathrm{cov}\left(\Phi(X),\Phi(Y)\right)
\end{equation}

When $\rho=0$, $X$ and $Y$ are independent, and hence, $p_1(X)$ and $p_1(Y)$ are independent. Then $\textrm{cov}\left(p_1(X),p_1(Y)\right)=0$. 

When $\rho\neq0$, we rewrite $Y$ as :
$$Y=\rho X+\sqrt{1-\rho^2}Z,$$
where $Z$ is a standard normally distributed random variable independent of $X$. Define $\Lambda(X)=E\left(\Phi(\rho X+\sqrt{1-\rho^2}Z)\mid X\right)$. Then,
\begin{equation*}
\begin{split}
\mathrm{cov}\left(p_1(X),p_1(Y)\right)
=&\mathrm{cov}\left(\Phi(X),\Phi(Y)\right)=\mathrm{cov}\left(\Phi(X),\Phi(\rho X+\sqrt{1-\rho^2}Z)\right)\\
=&E\left(\Phi(X)\Phi(\rho X+\sqrt{1-\rho^2}Z)\right)-\frac{1}{4}\\
=&E\left[\Phi(X)\times E\left(\Phi(\rho X+\sqrt{1-\rho^2}Z)\mid X\right)\right]-\frac{1}{4}\\
=&E\left(\Phi(X)\Lambda(X)\right)-\frac{1}{4}=\mathrm{cov}\left(\Phi(X),\Lambda(X)\right)
\end{split}
\end{equation*}
Suppose $W$ is another random variable sampled independently from the identical distribution of $X$. Then the covariance between $p_1(X)$ and $p_1(Y)$ can be rewritten as:
\begin{equation*}
\mathrm{cov}\left(p_1(X),p_1(Y)\right)=\mathrm{cov}\left(\Phi(X),\Lambda(X)\right)=\frac{1}{2}E\left[\left(\Phi(X)-\Phi(W)\right)\times\left(\Lambda(X)-\Lambda(W)\right)\right].
\end{equation*}
When $\rho$ is positive, both $\Phi(\cdot)$ and $\Lambda(\cdot)$ are increasing. Then, $\left(\Phi(X)-\Phi(W)\right)\times\left(\Lambda(X)-\Lambda(W)\right)$ is always non-negative. Accordingly, the covariance is always positive. When $\rho$ is negative, $\Phi(\cdot)$ is increasing and $\Lambda(\cdot)$ is decreasing and hence $\left(\Phi(X)-\Phi(W)\right)\times\left(\Lambda(X)-\Lambda(W)\right)$ is always non-positive. As a result, the covariance is always negative. In a word, the covariance shares the same sign as $\rho$, which finishes the proof of (i).

(ii) Notice that the covariance of $p_2(X)$ and $p_2(Y)$ can be rewritten as
\begin{equation}
\label{eq:cov_p_val}
\begin{split}
\mathrm{cov}\left(p_2(X),p_2(Y)\right)=& \ \mathrm{cov}\left(2\left(1-\Phi(|X|)\right),2\left(1-\Phi(|Y|)\right)\right)\\
=& \ 4\mathrm{cov}\left(\Phi(|X|),\Phi(|Y|)\right).
\end{split}
\end{equation}
Hence, it suffices to consider the sign of $\mathrm{cov}(\Phi(|X|),\Phi(|Y|))$ (which is of the same sign with  $\mathrm{cov}(p_2(X),p_2(Y))$). By Hoeffding's covariance identity, this covariance can be rewritten as:
{\small
\begin{equation*}
\begin{split}
\mathrm{cov}(\Phi(|X|),\Phi(|Y|))=&\int_0^1\int_0^1\left(\Pr\left(\Phi(|X|)\leqslant u,\Phi(|Y|)\leqslant v\right)-\Pr\left(\Phi(|X|)\leqslant u\right)\Pr\left(\Phi(|Y|)\leqslant v\right)\right)dudv\\
=&\int_0^1\int_0^1\left(\Pr\left(|X|\leqslant \Phi^{-1}(u),|Y|\leqslant \Phi^{-1}(v)\right)-\Pr\left(|X|\leqslant \Phi^{-1}(u)\right)\Pr\left(|Y|\leqslant \Phi^{-1}(v)\right)\right)dudv.
\end{split}
\end{equation*}
}
{Since for any fixed $u$ and fixed $v$, sets $G=\{(x,y)\in\mathbb{R}^2:-\Phi^{-1}(u)\leqslant x\leqslant\Phi^{-1}(u)\}$ and $F=\{(x,y)\in\mathbb{R}^2:-\Phi^{-1}(v)\leqslant y\leqslant\Phi^{-1}(v)\}$ are convex and symmetric about the origin, on the basis of the Gaussian correlation inequality, it holds that
\begin{equation}
\label{eq:gaussian_correlation_ineq}
\mu\left(G\cup F\right)\geqslant\mu(G)\times\mu(F),   
\end{equation}
where $\mu$ is the probability measure defined by the bivariate normal distribution of $(X,Y)$. \Cref{eq:gaussian_correlation_ineq} is equivalent to 
$$\Pr\left(|X|\leqslant \Phi^{-1}(u),|Y|\leqslant \Phi^{-1}(v)\right)-\Pr\left(|X|\leqslant \Phi^{-1}(u)\right)\Pr\left(|Y|\leqslant \Phi^{-1}(v)\right)\geqslant0.$$
Hence, \Cref{eq:cov_p_val} is also non-negative and so is the covariance of $p_2(X)$ and $p_2(Y)$.}
\end{proof}

\newpage
\section{Supplementary figures and tables}
\label{sec:more_sim_results}

\begin{figure}[ht]
     \centering
     \begin{subfigure}[h]{\textwidth}
         \centering
         \includegraphics[width=\textwidth]{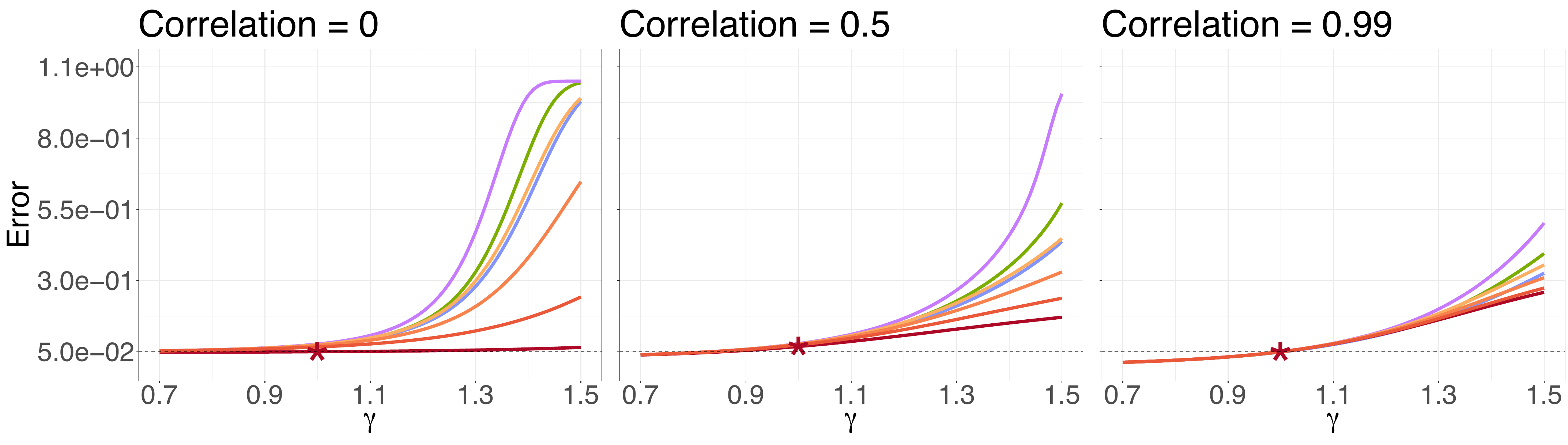}
         \caption{$\alpha=0.05$}
         \label{fig:error_100_0.05}
     \end{subfigure}
     \hfill
     \begin{subfigure}[h]{\textwidth}
         \centering
         \includegraphics[width=\textwidth]{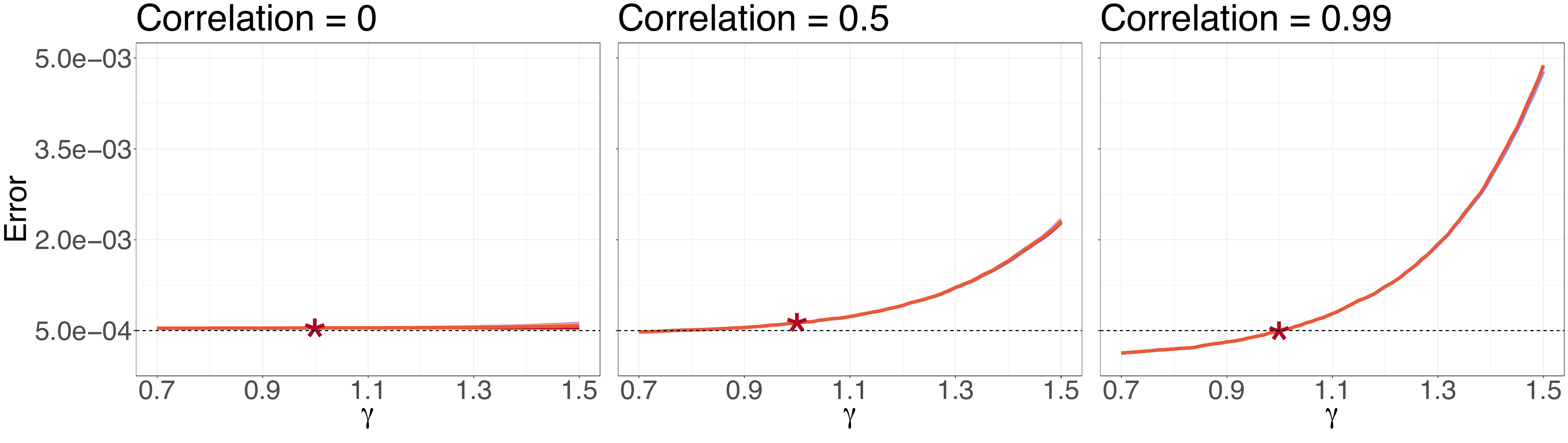}
         \caption{$\alpha=5\times10^{-4}$}
         \label{fig:error_100_5e-4}
     \end{subfigure}
\caption{The type-I error of the combination test when $n = 100$ with different distributions: Cauchy (star point), inverse Gamma (blue), Fr\'echet (green), Pareto (purple), student t (red), left-truncated t with truncation threshold $p_0=0.9$ (dark orange), left-truncated t with truncation threshold $p_0=0.7$ (orange), left-truncated t with truncation treshold $p_0=0.5$ (light orange). The vertical axis represents the empirical type-I error, and the horizontal axis stands for the tail index $\gamma$.}
\label{fig:n_100_gamma_vs_error}
\end{figure}

\vspace*{\fill}

\newpage
\vspace*{\fill}
\begin{figure}[ht!]
     \centering
     \begin{subfigure}[h]{\textwidth}
         \centering
         \includegraphics[width=\textwidth]{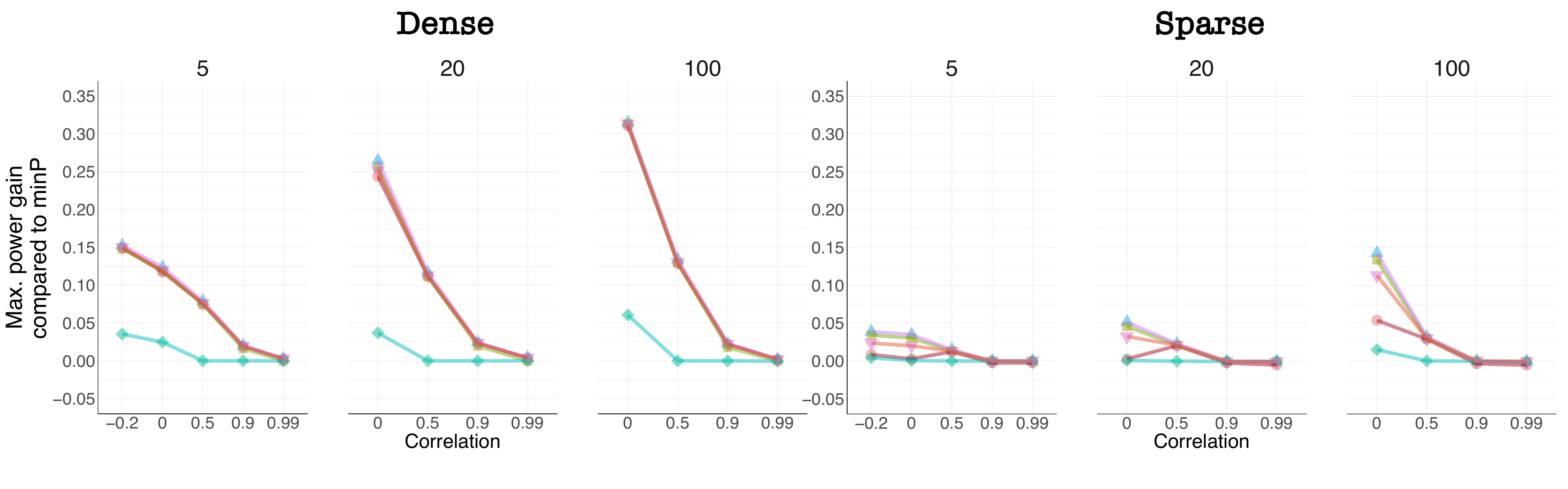}
         \caption{$\alpha=0.05$}
         \label{fig:power_0.05_with_correction}
     \end{subfigure}
     \hfill
     \begin{subfigure}[h]{\textwidth}
         \centering
         \includegraphics[width=\textwidth]{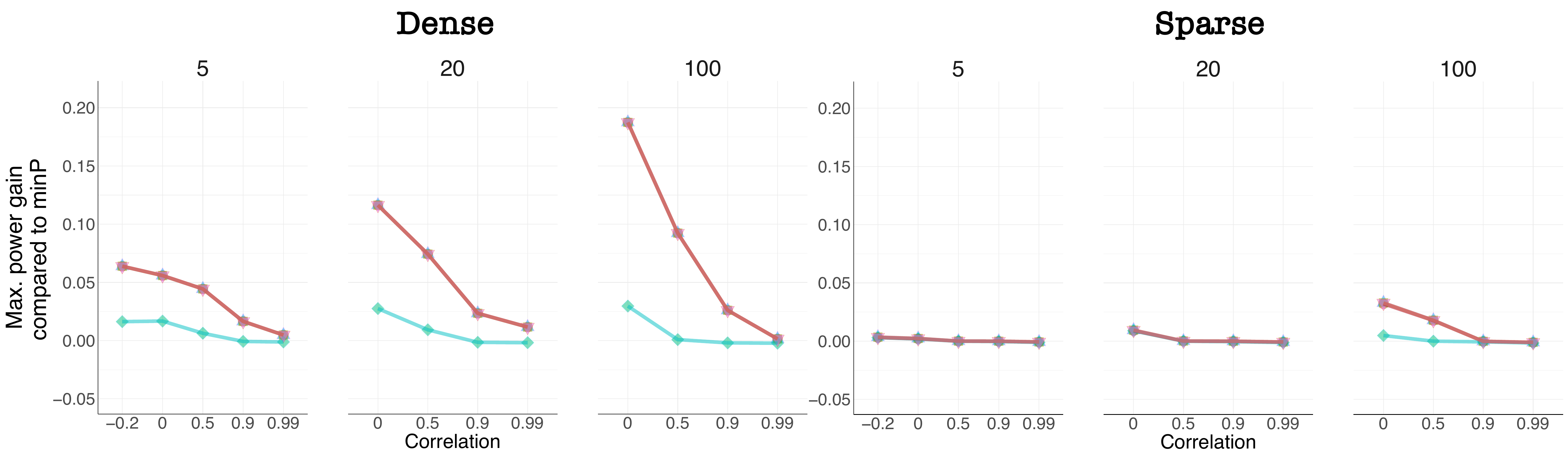}
         \caption{$\alpha=5\times10^{-4}$}
         \label{fig:power_5e-4_with_correction}
     \end{subfigure}
\caption{Power comparison with the minP test of the combination test with different distributions: Cauchy (red with round dot), Fr\'echet $\gamma=1$ (green with square dot), Pareto $\gamma=1$ (purple with triangular dot), left-truncated $t_1$ with truncation threshold $p_0=0.9$ (dark orange with inverted-triangle dot). Left plots correspond to dense signals, and right plots correspond to sparse signals. The maximum power gain is defined as the maximum of the empirical power difference between the proposed test and the Bonferroni test over all possible values of $\mu$.}
\label{fig:power_gain_with_correction}
\end{figure}
\vspace*{\fill}

\newpage
\vspace*{\fill}
\begin{figure}[ht!]
    \centering
    \includegraphics[width=\textwidth]{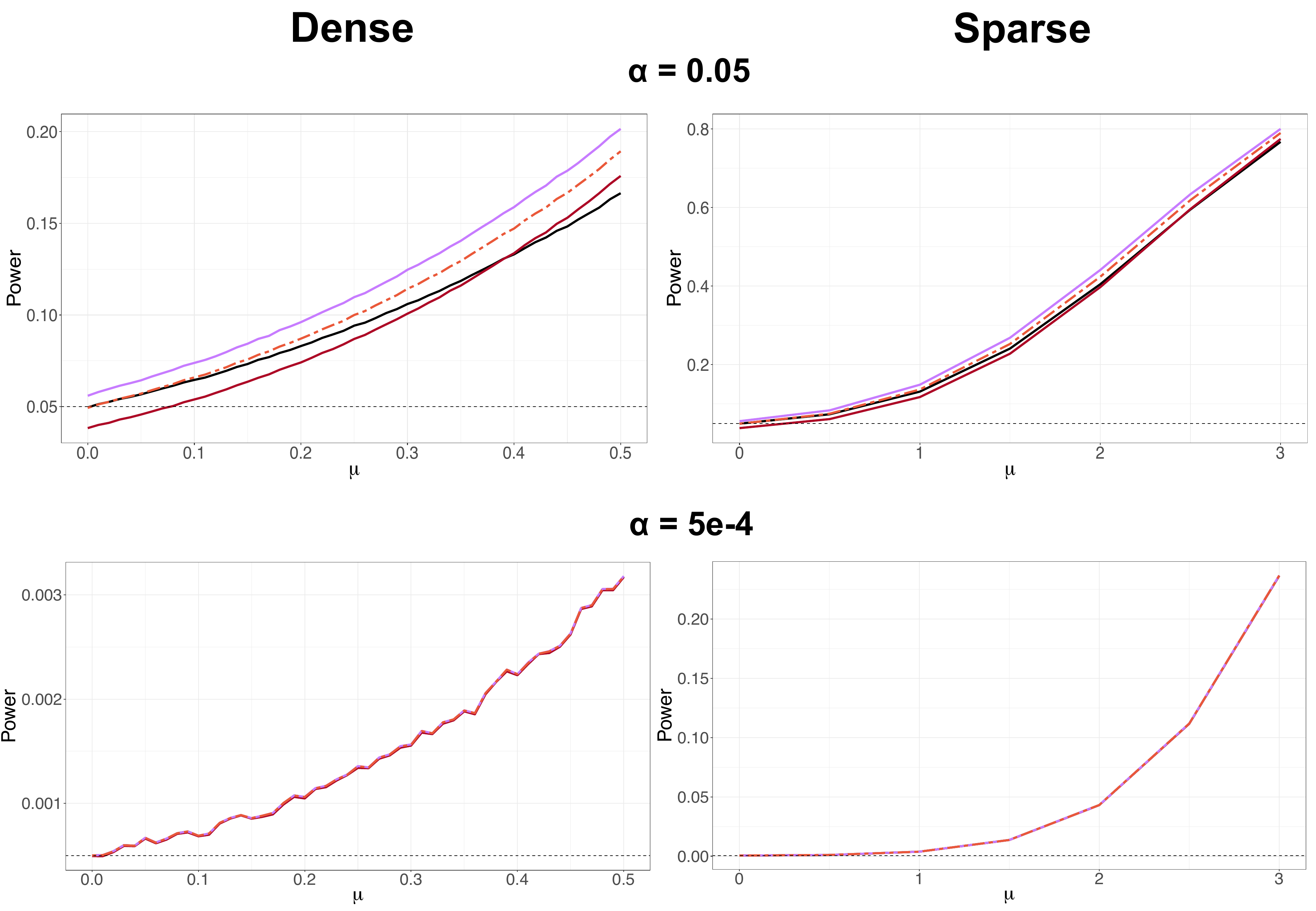} 
    \caption{Comparison of power (recall) when type-I error is controlled at level $\alpha=0.05$ and $5\times 10^{-4}$ of different methods: Bonferroni's test (black solid), Cauchy combination test (red solid), left-truncated $t_1$ with truncation level $p_0=0.9$ combination test (red dotted), and Pareto or Fr\'echet $\gamma=1$ combination test (purple solid). The number of base hypotheses is 5. Base p-values are one-sided p-values converted from multivariate z-scores with the mean $(\vec{0}_4,\mu)$ (to simulate sparse signals) and the mean $\vec{\mu}_5$ (to simulate dense signals). The common correlation $\rho=-0.2$.}
\label{fig:sparse_and_dense_power}
\end{figure}
\vspace*{\fill}

\newpage
\vspace*{\fill}
\begin{figure}[ht]
    \centering
    \includegraphics[width=0.8\textwidth]{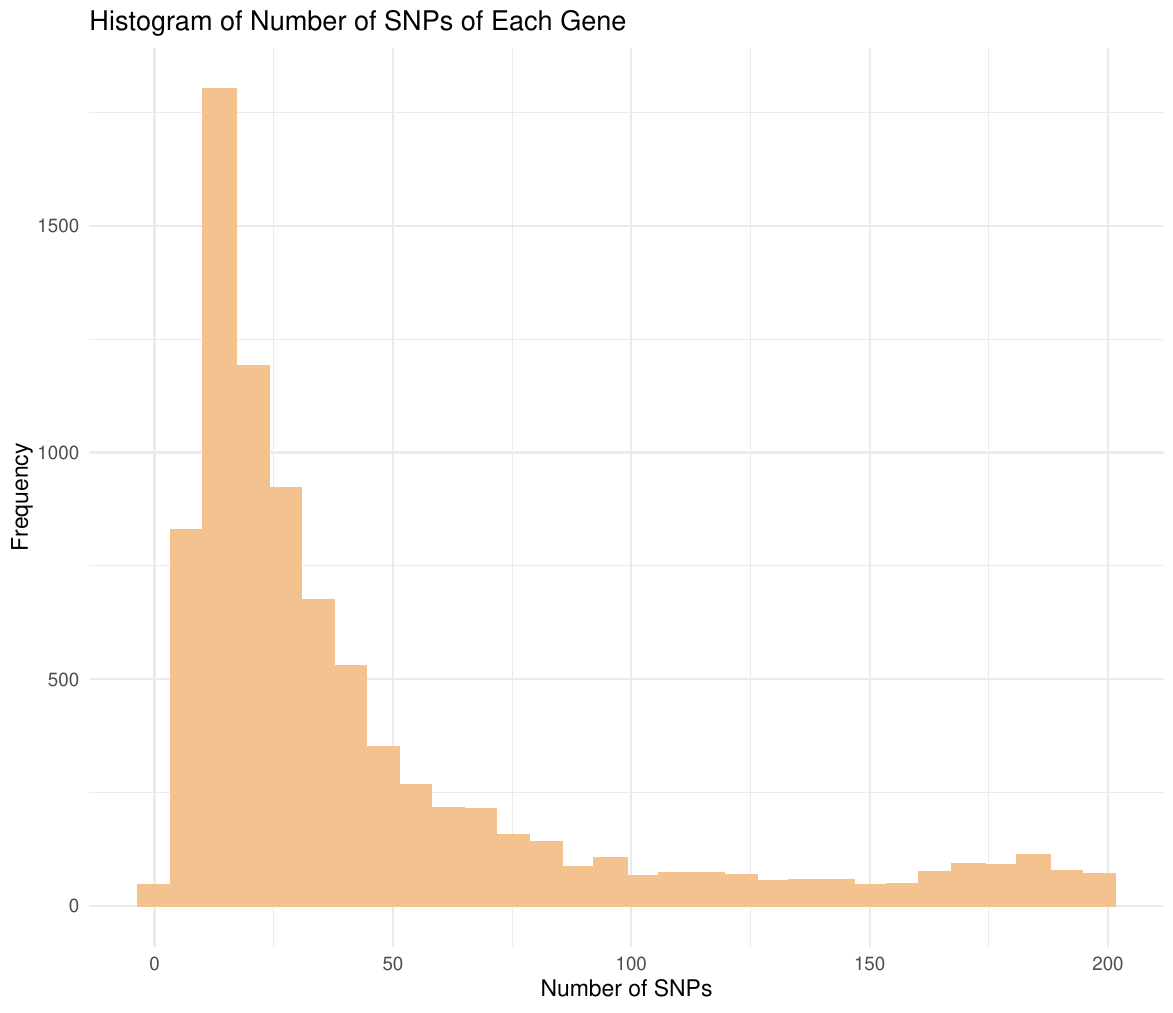} 
    \caption{The numbers of SNPs of all genes are smaller than 200. More than half of these numbers are smaller than 50.}
\label{fig:snps_number}
\end{figure}
\vspace*{\fill}

\newpage
\vspace*{\fill}
\begin{figure}[ht]
\centering
\includegraphics[width=\textwidth]{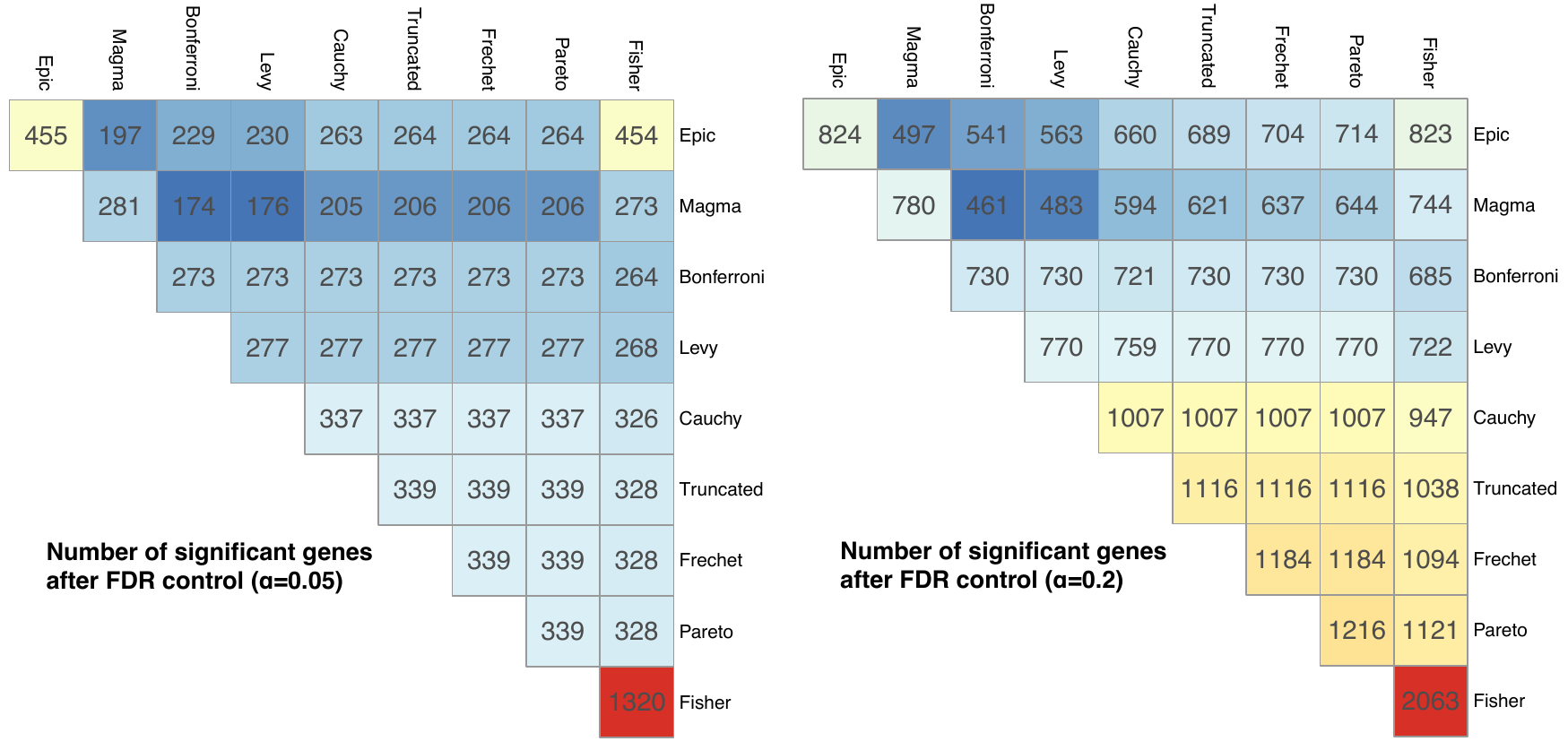} 
\caption{Number of significant genes for gene-level association testing combining SNP-level p-values when considering the subset of genes with at most $50$ associated SNPs. Diagonal values indicate the number of significant genes identified by each method; upper-triangular values indicate the number of overlapping discoveries between each pair of methods. Background colors correspond to the logarithms of the numbers. ``Truncated'' refers to the truncated $t_1$ distribution with truncation threshold $p_0 = 0.9$. For Fr\'echet and Pareto distributions, the tail index is set to $\gamma = 1$.}
\label{fig:gwas_sig_gene_less_snps}
\end{figure}
\vspace*{\fill}

\newpage
\vspace*{\fill}
\begin{figure}[ht]
\centering
\includegraphics[width=\textwidth]{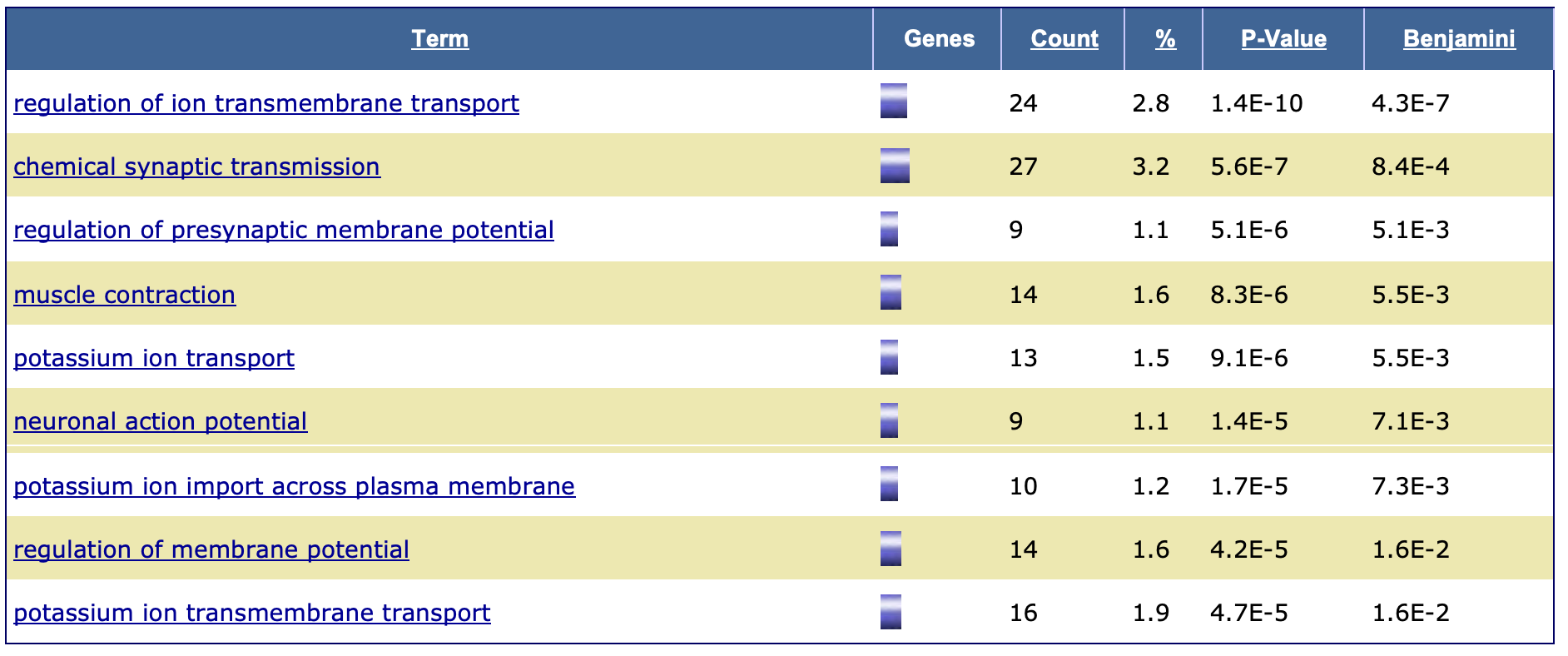} 
\scriptsize
\caption{Gene set enrichment analysis using genes significantly associated with schizophrenia (SCZ) from GWAS. Because a relatively large number of genes are needed to reach significance from the gene set enrichment analysis, we set the genome-wide FDR significance threshold to be 0.2 and included 939 genes that are detected by Cauchy/Fr\'echet/Pareto but not by Bonferroni. The enriched gene ontology terms are in agreement with previous studies and reports on SCZ: ion transporter pathway \citep{liu2022impact}, synaptic transmission \citep{favalli2012role}, and potassium ion transmembrane transport \citep{romme2017connectome}.}
\label{tab:heavy_tail_class_with_dep_condition}
\end{figure}
\vspace*{\fill}

\newpage
\vspace*{\fill}
\begin{table}[ht]
\caption{Type-I error control of the combination tests when test statistics follow multivariate t distribution when $n = 100$. Values inside the parentheses
are the corresponding standard errors. For the Fr\'echet
and Pareto distributions, they are the corresponding distribution with tail index $\gamma = 1$}
\begin{subtable}{\linewidth}
\centering
\scriptsize
\caption{$\alpha=5\times10^{-2}$}
\scalebox{1}{\begin{tabular}{cc|ccccccc}
\hline
& & \multicolumn{7}{c}{Distributions} \\
$\rho$ & $C^t_{\nu.\rho}$ & \multicolumn{1}{c}{Cauchy} & \multicolumn{1}{c}{Pareto} & \multicolumn{1}{c}{Truncated $t_1$} & \multicolumn{1}{c}{Fr\'echet} & \multicolumn{1}{c}{Levy} & \multicolumn{1}{c}{Bonferroni} & \multicolumn{1}{c}{Fisher} \\ \hline
&  & \multicolumn{1}{c}{}  & \multicolumn{1}{c}{}  & \multicolumn{1}{c}{} & \multicolumn{1}{c}{} & \multicolumn{1}{c}{} & \multicolumn{1}{c}{} & \multicolumn{1}{c}{} \\
0 & 0.18 & 7.07E-03 & 5.36E-02 & 4.58E-02 & 5.19E-02 & 1.18E-02 & 6.30E-03 & 1.77E-01\\
& & (8.38E-05) & (2.25E-04) & (2.09E-04) & (2.22E-04) & (1.08E-04) & (7.91E-05) & (3.82E-04) \\
& & & & & & & & \\
0.5 & 0.39 & 4.20E-02 & 5.29E-02 & 5.00E-02 & 5.15E-02 & 8.53E-03 & 3.74E-03 & 2.92E-01\\
& & (2.01E-04) & (2.24E-04) & (2.18E-04) & (2.21E-04) & (9.20E-05) & (6.10E-05) & (4.55E-04)\\
& & & & & & & &\\
0.9 & 0.72 & 5.01E-02 & 5.11E-02 & 5.07E-02 & 4.98E-02 & 5.84E-03 & 1.51E-03 & 3.08E-01\\
& & (2.18E-04) & (2.20E-04) & (2.19E-04) & (2.18E-04) & (7.62E-05) & (3.89E-05) & (4.62E-04)\\
& & & & & & & & \\
0.99 & 0.91 & 5.03E-02 & 5.03E-02 & 5.03E-02 & 4.91E-02 & 5.16E-03 & 7.75E-04 & 3.10E-01\\
& & (2.18E-04) & (2.19E-04) & (2.19E-04) & (2.16E-04) & (7.16E-05) & (2.78e-05) & (4.63E-04)\\ \hline
\end{tabular}}

\end{subtable}
\hfill
\begin{subtable}{\linewidth}
\centering
\scriptsize
\caption{$\alpha=5\times10^{-4}$}
\scalebox{1}{
\begin{tabular}{cc|ccccccc}
\hline
 & & \multicolumn{7}{c}{Distributions} \\
$\rho$ & $C^t_{\nu,\rho}$ & \multicolumn{1}{c}{Cauchy} & \multicolumn{1}{c}{Pareto} & \multicolumn{1}{c}{Truncated $t_1$} & \multicolumn{1}{c}{Fr\'echet} & \multicolumn{1}{c}{Levy} & \multicolumn{1}{c}{Bonferroni} & \multicolumn{1}{c}{Fisher} \\ \hline
\multicolumn{1}{l}{} & \multicolumn{1}{l|}{} & & & & & & & \\
0 & 0.18 & 7.00E-05 & 4.91E-04 & 4.89E-04 & 4.91E-04 & 1.26E-04 & 7.70E-05 & 9.38E-02\\
& & (8.37E-06) & (2.22E-05) & (2.21E-05) & (2.22E-05) & (1.12E-05) & (8.77E-06) & (2.92E-04)\\
\multicolumn{1}{l}{} & \multicolumn{1}{l|}{} & & & & & & & \\
0.5 & 0.39 & 3.87E-04 & 4.79E-04 & 4.79E-04  & 4.79E-04 & 8.60E-05  & 4.10E-05 & 2.14E-01\\
&  & (1.97E-05) & (2.19E-05) & (2.19E-05) & (2.19E-05) & (9.27E-06)  & (6.40E-06) & (4.10E-04) \\
& & & & & & & &\\
0.9 & 0.72 & 5.36E-04 & 5.42E-04 & 5.42E-04 & 5.42E-04  & 4.90E-05   & 1.00E-05 & 2.49E-01 \\
& & (2.31E-05) & (2.33E-05) & (2.33E-05) & (2.33E-05) & (7.00E-06)   & (3.16E-06) & (4.33E-04) \\
& & & & & & & &\\
0.99 & 0.91 & 5.42E-04  & 5.43E-04 & 5.43E-04 & 5.43E-04 & 4.40E-05 & 7.00E-06 & 2.56E-01 \\
& & (2.33E-05) & (2.33E-05) & (2.33E-05) & (2.33E-05) & (6.63E-06)   & (2.64E-06) & (4.36E-04) \\ \hline
\end{tabular}}
\end{subtable}
\label{tab:t_copula_100}
\end{table}

\begin{table}[ht]
\centering
\caption{Cutoff ratio between minP and the Bonferroni test. The nominal significance level for the global test is $\alpha=0.05$ or $5\times10^{-4}$.}
\footnotesize{
\begin{tabular}{c|cccccc}
\hline
& \multicolumn{3}{c}{$\alpha=0.05$} & \multicolumn{3}{c}{$\alpha=5\times10^{-4}$}\\
& & & & & & \\
$\rho$ & $n=5$ & $n=20$ & $n=100$ & $n=5$ & $n=20$ & $n=100$ \\ \hline
& & & & & & \\
0    & 1.02 & 1.03 & 1.03 & 0.97 & 0.98 & 0.99 \\
& & & & & & \\
0.5  & 1.26 & 1.63 & 2.3 & 1.05 & 1.08 & 1.27 \\
& & & & & & \\
0.9  & 2.46 & 5.90 & 17.87 & 1.83 & 3.36 & 7.60 \\
& & & & & & \\
0.99 & 3.94 & 13.46	& 58.90	& 3.32 & 9.69 & 42.18 \\
& & & & & & \\
\hline
\end{tabular}
}
\label{tab:corrected_error}
\end{table}
\vspace*{\fill}


\end{document}